\documentclass[12pt]{article} 

\usepackage{setspace}
\usepackage{hyperref}
\usepackage{url}

\usepackage{amsmath,amssymb,amscd,braket,amsthm}
\usepackage{amsfonts}
\usepackage{color}
\usepackage{graphicx}
\usepackage{cite}
\usepackage{bbold}
\usepackage{tikz} 
\usepackage{multirow}
\usepackage{mathrsfs}
\usepackage{subcaption}
\usepackage{comment}
\usepackage{mathtools}

\numberwithin{equation}{section}

 \newcommand{\bea}{\begin{eqnarray}}
\newcommand{\eea}{\end{eqnarray}}
\newcommand{\be}{\begin{equation}}
\newcommand{\ee}{\end{equation}}
\newcommand{\ba}{\begin{align}}
\newcommand{\ea}{\end{align}}

\newcommand{\PP}{\mathbb{P}} % Proiettivo complesso

\newcommand{\N}{\mathcal{N}}

\newcommand{\M}{\mathcal{M}}

\newcommand{\R}{\mathcal{R}}

\newcommand{\Tr}{{\rm {Tr}}}

\newcommand{\tr}{ {\rm Tr}}

\newcommand{\A}{\mathcal{A}}

\newlength{\slength}
\renewcommand{\title}[1]{\vbox{\center\LARGE{#1}}\vspace{5mm}}
\renewcommand{\author}[1]{\vbox{\center#1}\vspace{5mm}}
\newcommand{\address}[1]{\vbox{\center\footnotesize\em#1}}
\newcommand{\email}[1]{\vbox{\center\footnotesize\tt#1}\vspace{5mm}}

\newcommand\E{\text{E}}

\newcommand\la{\lambda}
\newcommand\Ga{\Gamma}
\newcommand\F{\mathcal{F}}
\newcommand\V{\mathcal{V}}
\newcommand\U{\mathcal{U}}
\newcommand\W{\mathcal{W}}
\newcommand\PS{\mathcal{PS}}

\theoremstyle{definition}
\newtheorem{definition}{Definition}[section]
\theoremstyle{plain}
\newtheorem{theorem}{Theorem}[section]
\theoremstyle{plain}

\theoremstyle{plain}

\theoremstyle{plain}

\newcommand{\bluecheck}{}%
\DeclareRobustCommand{\bluecheck}{%
  \tikz\fill[scale=0.4, color=blue]
  (0,.35) -- (.25,0) -- (1,.7) -- (.25,.15) -- cycle;%
}
\newcommand{\redcross}{}%
\DeclareRobustCommand{\redcross}{%
  \textcolor{red}{$\times$}%
}

\begin{document}

\begin{titlepage}

\begin{center}

\hfill \\
\hfill \\
\vskip 1cm

\title{Microcanonical Operator Spectrum, BPS Chaos \\ and Free Probability Theory}

\author{Alexandre Belin$^{a,b}$, Yichao Fu$^{c,d}$, Federico La Rocca$^{a}$
}

\address{
${}^a$Dipartimento di Fisica, Universit\`a di Milano - Bicocca \\
I-20126 Milano, Italy

\vspace{1em}
${}^b$INFN, sezione di Milano-Bicocca, I-20126 Milano, Italy
}

\address{${}^c$Department of Physics and Photon Science, Gwangju Institute of Science and Technology, \\
123 Cheomdan-gwagiro, Gwangju 61005, Korea}

\address{${}^d$Center for High Energy Physics, Peking University,\\
No.5 Yiheyuan Rd, Beijing 100871, P. R. China}

\email{alexandre.belin@unimib.it, yichao.fu.physics@gmail.com, federicolarocca99@gmail.com}

\end{center}

\abstract{

Matrix elements of simple operators are central to thermalization in chaotic quantum many-body systems. Their moments are often described statistically through the Eigenstate Thermalization Hypothesis. In this paper, we study a more refined characterization of these matrix elements: the eigenspectrum of simple operators once projected to microcanonical energy windows. The spectrum encodes all moments of the matrix elements at once, and is particularly relevant when the spectrum of the Hamiltonian is degenerate. Assuming that the eigenbasis of the simple operator is Haar-randomly oriented with respect to that of the Hamiltonian, we prove that the spectrum of an operator projected to a small microcanonical window obeys Gaussian random matrix statistics. We do this directly by studying the probability distribution inherited from the random orientation, as well as using tools from Free Probability Theory. We discuss the implications for BPS chaos and more generally for eigenstate thermalization.

}

\vfill

\end{titlepage}

\eject

\tableofcontents

\newpage

\section{Introduction}

Chaotic quantum many-body systems are believed to obey universal properties. A prominent example is the fact that they thermalize. Microscopically, this is best explained by the Eigenstate Thermalization Hypothesis (ETH) \cite{PhysRevA.43.2046,PhysRevE.50.888}
\be
\braket{E_m | O | E_n} = \delta_{m,n} f(\bar{E}) + g(\bar{E}, \delta E) e^{-S(\bar{E}/2)} R_{mn} \,,
\ee
where $\bar{E},\delta E$ are the mean energy and energy difference of the two states, and $S,f,g$ are smooth functions of their arguments encoding the entropy, thermal one- and two-point functions, respectively. $R_{mn}$ are erratic numbers, which statistically have an approximate Gaussian distribution with zero mean and unit variance. 

The ETH ansatz explains thermalization in systems that obey it (see \cite{DAlessio:2015qtq} for a review). The erratic off-diagonal matrix elements $R_{mn}$ were considered independent random Gaussian variables for a long time. Since the main observable of interest was the two-point function (as a standard diagnostic of thermalization), this was sufficient. But it was later realized that the higher point moments are just as important if one wants to consider higher-point correlation functions \cite{Foini:2018sdb,Chan:2018fsp,Murthy:2019fgs}. Consistency with higher-point correlation functions thus implies a whole tower of non-Gaussian moments.\footnote{These non-trivial moments have also been observed for OPE coefficients in CFTs, where they follow from crossing symmetry \cite{Belin:2021ryy,Anous:2021caj}.} They are often described in a moment expansion since the net weight of higher non-Gaussianities are exponentially suppressed in the entropy (even if their collective effect can be just as important as the Gaussian part in correlation functions).

In this paper, we will consider a different approach to this problem. We will study the eigenspectrum of a simple operator, once projected to a microcanonical window. This corresponds to the eigenvalues of the operator
\be \label{projectorintro}
O_K = P_K O P_K \,,
\ee
where $P_K$ is a projector onto a microcanonical window of size $K$. The basis we project to is not the eigenbasis of $O$ but rather that of the Hamiltonian, and the projection is crucial. To see this, one can imagine the spin operator on one site in a chaotic spin chain. Its eigenvalues are half $1/2$, half $-1/2$, so the spectrum is extremely simple. However, one projected to a smaller window, the spectrum will become very complicated. The spectrum should be viewed as a collective description of all the moments of $R_{mn}$. It contains however very fine-grained information, for example eigenvalue statistics, which cannot be easily read off from the first few moments described by the ETH.

This question has unfortunately not been explored much. It was first discussed in \cite{Richter:2020bkf,Wang:2021mtp,Wang:2023qon}\footnote{See also \cite{Dymarsky:2017zoc} for an earlier discussion, focusing mostly on the largest eigenvalue.}, where the authors studied the scale at which Gaussian ETH (i.e. where $R_{mn}$ are actual Gaussian independent) becomes accurate in definite Hamiltonians. Once the $R_{mn}$ become actual independent random Gaussian, this fixes the spectrum. Another direct attempt was done in \cite{SredIni}: assuming that the relative basis change between the eigenbasis of a spin operator and that of the Hamiltonian is a Haar-random unitary, \cite{SredIni} studied the expectation value (under the Haar average) of the one-point function of the density of states. In the limit of small relative size of the projection, a semicircle law was found. This will be the general idea that we will explore in this paper, by building the full probability distribution for the eigenvalues of the projected operator. It is also worth mentioning the ETH matrix model \cite{Jafferis:2022uhu,Jafferis:2022wez}, which also in principle encodes the spectrum of $O_K$, even if it has not been studied in that model.

An important motivation for our work is the study of BPS chaos (or more generally, the study of chaos when the spectrum of the Hamiltonian is degenerate). In that case, there are many states of fixed energy and we can take the projector $P_K$ to project on that degenerate subspace. Standard probes of quantum chaos such as eigenvalue statistics or ETH are not applicable there, as all eigenvalues have exactly zero spacing, and there is no preferred basis in a degenerate subspace so matrix elements are not meaningful. The eigenspectrum of the projected operator becomes the only physical quantity to study. \cite{Lin:2022rzw,Lin:2022zxd,Chen:2024oqv} conjectured that the eigenvalue statistics of the projector operator should encode whether a BPS subspace is chaotic.\footnote{\cite{Chen:2024oqv} stresses the distinction between strong chaos and weak chaos: to have strong chaos, the projected operator needs to display random matrix signals over the full size of the projected Hilbert space.} The motivation for this criterion was precisely what we will study in this paper, namely that chaos of the underlying Hamiltonian should imply that eigenspaces of operators are randomly oriented, even with respect to degenerate subspaces. Our work provides further justification for BPS chaos and makes the diagnostic quantitative: by assuming that particular subspaces are randomly oriented, we will find a definite statistical distribution for the eigenvalues.

\subsection*{Summary of Results}

In this paper, we establish the following result. Consider an operator $O$ acting on an $N$-dimensional Hilbert space. $O$ has eigenvalues $O_i$ with $1\leq i \leq N$. We consider the projected operator $O_K$, given by \eqref{projectorintro}. We take the projector to be an Haar-randomly rotated basis, compared to the original eigenbasis of $O$. We then consider the large $N$ limit, such that $N,K\to\infty$, with the ratio $\alpha=K/N$ fixed. We prove the following result for the probability measure on the space of eigenvalues ${\lambda_j}$ of $O_K$: to quadratic order in the small $\alpha$ expansion, we have
\be \label{Pgaussianintro}
\mu(\lambda_1,\cdots,\lambda_K) = \mu(\text{GUE})_{\sigma_O, \bar{O}} +\mathcal{O}(\alpha^3) \,.
\ee
The probability measure $\mu(\text{GUE})_{\sigma_O, \bar{O}}$ is a Gaussian probability distribution in the GUE universality class, centered at the mean of the trace of $O$: $\bar{O}=\frac{1}{N}\Tr O$. The width of the probability distribution is set by the variance of the operator $\sigma_ O$. More precise equations can be found around equation \eqref{eq-jpdf in LargeN}. Note that the eigenvalues of $O$ were deterministic (i.e. fixed), but those of $O_K$ are probabilistic. This follows from taking a Haar-random unitary to rotate the operator before taking the projection. The probability distribution \eqref{Pgaussianintro} should be viewed as inherited from the Haar measure. The results we prove are independent of the original eigenvalue distribution for $O$.\footnote{There is a mild assumption on the scaling with $N$ of the moments of $O$, which can be viewed in our approach as an equivalent of dealing with a simple operator.} Our result can thus be viewed as a type of central limit theorem for random projection.\footnote{We note that at the level of the cumulants, this result was already derived in \cite{Wang:2023qon}. Here, we derive the full probability distribution, which in particular can be used to probe correlations in the spectrum. }

\eqref{Pgaussianintro} follows directly from the Haar measure on the random unitary, by evaluating the projection in the large $N$ limit. This thus provides a direct derivation of the full probability distribution. However, evaluating the first few moments of the density of states $\rho(\lambda)$ can also be accomplished by using Free Probability Theory. Free Probability Theory is the theory of probability distributions over non-commuting objects (matrices). It has so far only made rare appearances in high-energy community\footnote{See however \cite{Gopakumar:1994iq} for an early discussion in the context of large $N$ gauge theories, and \cite{Wang:2022ots,Camargo:2025zxr} for more recent connections.}, but it is becoming an important idea in the chaos/ETH context \cite{Pappalardi:2022aaz,Fava:2023pac,Vallini:2025vvq, Espindola:2026xuc, Espindola:2026qfo}. In this paper, we use Free Probability Theory to provide a second proof of our random projection result. The use of free probability makes the proof quite elegant, although only partial: we can prove that the density of states becomes the semicircle law, and we can show that the two-point correlation function of the density of states, expanded moment by moment, also agree with the GUE distribution. Along the way, we give a general introduction to Free Probability and mention several important theorems, which we hope can be applied in other contexts.

This paper is organized as follows: in Section \ref{sec:2}, we set up the problem and define the projected operator we will study in the paper. We then discuss the probability measure on the eigenvalues of the projected operator, inherited from the Haar measure. We establish that in the large $N$ limit, it becomes that of the GUE. In Section \ref{sec: Numerics}, we numerically check our results, for different types of original operator spectra, and different size projections. In Section \ref{sec:freeprob}, we discuss the Free Probability approach to the problem. We first review the salient features of Free Probability, summarize many relevant theorems and apply them to our problem. We conclude with some open questions in Section \ref{sec:discussion}. The appendices discuss other examples of projected operators not covered in the main text, as well as some proofs of theorems.

\section{Operators and Microcanonical Windows} \label{sec:2}

\subsection{Projecting an operator onto a random subspace}
Consider a quantum chaotic system with Hilbert space $\mathscr{H}$ of dimension $\dim \mathscr{H}=N$. We imagine the system has a Hamiltonian $H$, with energy levels $\{  \ket{E_k} \}_{k=1...N}$. We now consider a Hermitian operator $O$, which acts linearly on $ \mathscr{H}$. $O$ can always be written in diagonal form, once represented in its eigenbasis
\begin{equation}
    O=\mathrm{diag}(O_1, O_2, \dots, O_N)\,.
\end{equation}
For now, we make no assumptions on the eigenvalues $O_i$. Physically, we will have in mind that $O$ is a simple operator which, just like in the ETH context, should be thought of as a few body operator.\footnote{This does constrain the eigenvalues $O_i$, but only mildly. We will come back to this point shortly.}

We now make the main assumption that we will use throughout this paper: we assume that the system is chaotic. This usually has several consequences on the Hamiltonian: the two main consequences are related to the eigenvalues and eigenvectors of the Hamiltonian. In this paper, we will not need to say anything about the eigenvalues. We will however assume that the eigenvectors of the Hamiltonian are randomly oriented. Of course, any physical Hamiltonian has a fixed eigenbasis, but here we will take these vectors to be Haar-randomly oriented with respect to the basis of the operator $O$: 
\be \label{defunit}
\braket{E_j| a } = U_{j a} \,,
\ee
where $ a=1,\cdots, N$ represents the eigenvectors of $O$. For any fixed Hamiltonian and fixed operator $O$, there exists such a unitary $U$, but here, we will model chaos by assuming this unitary is drawn randomly from the Haar ensemble, with probability $P(U)$. To compute any quantity, we will need to perform Haar averages
\be \label{Haarmeasure}
\int dU \mu_{\text{Haar}}(U) \,.
\ee
By a slight abuse of notation, and because this is a standard notation for Haar random unitaries, every time we write an integral over a unitary matrix, we will keep the Haar measure implicit:
\be
\int dU \equiv \int dU \mu_{\text{Haar}}(U)\,.
\ee

To summarize, we will imagine fixing the basis that is singled out by the observable, and we will characterize the operator entirely through its eigenvalue distribution. We will consider cases where the eigenspectrum of $O$ is fixed, but also cases where it is itself drawn from a probability distribution.
 This allows us to describe the spectrum of $O$ by a distribution function $f(O)$. Deterministic spectra can be included as a special case by taking $f(O)$ to be a sum of Dirac delta functions
 \begin{equation}
     f(O)=\frac{1}{N}\sum^N_{i=1} \delta(O-O_i) \,.
 \end{equation}

We now want to project $O$ onto a subspace spanned by $K$ energy eigenstates. Because $O$ is originally in its eigenbasis, before we project we must first rotate it to the eigenbasis of the Hamiltonian, using the unitary operator \eqref{defunit}. The projected operator thus reads
\begin{align}\label{trunc_op}
    O_K=P_K U OU^\dagger P_K \,.
\end{align}
Here $P_K$ is an orthogonal projector of rank $K$, which in its diagonal form reads
\begin{align}
    P_K = \mathrm{diag}(\underbrace{1,\ 1,\ \dots,\ 1}_{K \text{ times}},\underbrace{0,\ 0,\ \dots,\ 0}_{N-K \text{ times}}) \,.
\end{align}
Notice that in $P_K$ we put all the ones in the first $K$ places of the main diagonal without loss of generality, since we can reshuffle the energy eigenstates as we desire. It is useful to define the parameter
\begin{align}
    \alpha:=\frac{K}{N}\,,
\end{align}
which takes the value from $0$ to $1$ and parametrizes the ratio of the sizes of the projected versus full Hilbert space.

We would now like to study the eigenvalues of $O_K$ (the remaining $K$ non-trivial eigenvalues $\{\lambda_i\}_{i=1,\dots,K}$). These eigenvalues are never deterministic. They are probabilistic, simply because they inherit a probability distribution from the Haar measure \eqref{Haarmeasure}. Our goal will be to study the statistics of the $\lambda_i$, and understand their dependence on $\alpha$ and $f(O)$. We will mostly work in the limits $N\to\infty$ and $K=\alpha N$ with $\alpha$ finite. This can be carried out numerically given the prescription specified above, and we will report the results for various $f(O)$ in Section \ref{sec: Numerics}. Analytically, we will employ random matrix techniques to compute the joint probability density of the eigenvalues of $O_K$ in the following section. We will derive exact finite $N$ expressions and proceed with the simplifications that arise in the large $N$ limit, as we now explain.

\subsection{Deriving the projected eigenvalue probability distribution\label{sec: jpdf from RMT}}

The goal of this section is to establish a universality result, which one can view as a type of central limit theorem for the projected spectrum. In the limit of small $\alpha$, the probability distribution gives the GUE ensemble, centered at the mean
\be
\bar{O}=\frac{1}{N}\Tr O \,,
\ee
and whose variance is fixed by 
\be
\sigma^2_O=\frac{\Tr O^2}{N}-\bar{O}^2 \,.
\ee

To show this, we will study the joint eigenvalue distributions of the projected matrix, which controls all features of the spectrum. From it, we can extract the density of states, level-spacing distribution and spectral correlation functions.
Given that $O_K$ is a $K\times K$ matrix, we can assign its eigenvalues to be 
\begin{equation}
    \{\lambda_1,\lambda_2\dots \lambda_k\} \,.
\end{equation}
To evaluate the joint eigenvalue distribution of this randomly projected matrix, we consider the following expression\footnote{We will use the notation $d\lambda=\prod_{i=1}^K d\lambda_i$.}
\begin{equation}
    \mu(\lambda)d\lambda=\int \delta_K(PUOU^\dagger P-V\Lambda V^\dagger)d(V\Lambda V^\dagger)dU \,,
    \label{eq-jpdf delta}
\end{equation}
where we have the $K$-dimensional $\delta$ function $\delta_K$. Here, we have integrated in the matrix $V\Lambda V^{\dagger}$ using the delta function, and have directly written the integral in matrix form in diagonal form: $V$ is a unitary matrix and $\Lambda$ is diagonal. Such a $K$-dimensional matrix Dirac delta function can be rewritten in the Fourier integral form \cite{Zhang:2016mdq}:
\begin{equation}
    \delta_K(PUOU^\dagger P-V\Lambda V^\dagger)=c_M\int e^{i\Tr[M(PUOU^\dagger P-V\Lambda V^\dagger)]} dM\,,
    \label{eq-delta int form}
\end{equation}
where $c_M=\frac{1}{2^K\pi^{K^2}}$ and $dM=\prod^K_{i}M_{ii}\prod_{1\leq i <j \leq K} d\mathrm{Re}(M_{ij})d\mathrm{Im}(M_{ij}) $. We also have \cite{Zhang2019HarmonicAF}
\begin{equation}
    d(V\Lambda V^\dagger)=c_{\Delta}\Delta^2_K(\lambda)d\lambda dV \,,
\end{equation}
where $c_\Delta=\frac{\pi^{K(K-1)/2}}{\prod^K_{j=0} j!}$ and $\Delta$ is the Vandermonde determinant. We remind the reader one final time that $V$ is a unitary, so its integral is taken with respect to the Haar measure.

In what follows, it will be convenient to work in the full Hilbert space (which is $N\times N$) as well, so we introduce:
\begin{equation}
    \tilde{V}=\begin{bmatrix}V&0\\0&0 \end{bmatrix}_{N\times N},~~~\tilde{V}^\dagger=\begin{bmatrix}V^\dagger&0\\0&0 \end{bmatrix}_{N\times N},~~~\tilde{\Lambda}=\begin{bmatrix}\Lambda&0\\0&0 \end{bmatrix}_{N\times N} \,.
\end{equation}

One can do this since such an embedding in $N\times N$ does not change the trace in (\ref{eq-delta int form}).
Here, $dU_N$ and $dV_K$ are the same as in (\ref{eq-jpdf delta}), where the subscripts are to point out the dimensionality of the integral.
The $\delta$ function can be expressed using its integral representation by introducing an auxiliary Hermitian matrix $\tilde{M}$. Since all non-trivial information is only in the upper-left $K\times K$ corner, we can therefore consider the auxiliary matrix in the embedding form
\begin{equation}
    \tilde{M}=\begin{bmatrix}M&0\\0&0 \end{bmatrix}_{N\times N}\,.
\end{equation}
Therefore, we have 
\begin{equation}
    \mu(\lambda)=c_\Delta c_M\int \Delta^2_K(\lambda) e^{i\Tr(\tilde{M}(PUOU^\dagger P-\tilde{V}\tilde{\Lambda}\tilde{V}^\dagger))}dM_KdU_NdV_K
    \label{eq-IntRep of delta} \,.
\end{equation}

This is the master equation that we will use to derive the probability distribution in nicer form.

\subsubsection{Exact evaluation}
We first evaluate the integrals in (\ref{eq-IntRep of delta}) exactly. We will only take the limit $N\gg K$ at later stage, and we will mention it explicitly when we do. We will soon see that this gives the form of the GUE joint probability distribution. However, details, such as the mean and variance of the distribution, are missing. They can be recovered once we consider large-$N$ techniques later. 

Our main aims are to evaluate the following two integrals coming from (\ref{eq-IntRep of delta}):
\begin{eqnarray}
    I_1&=&\int  e^{i\Tr(\tilde{M}PUOU^\dagger P)}dU_N \\
    I_2&=&\int  e^{-i\Tr(\tilde{M}\tilde{V}\tilde{\Lambda}\tilde{V}^\dagger)}dV_K \,.
\end{eqnarray}
In the meantime, we also consider the diagonalization of the matrix $\tilde{M}$, which yields
\begin{equation}
    dM_K=d(WmW^\dagger)=c_\Delta \Delta^2_K(m)dmdW \,,
\end{equation}
where $W$ is a $K$-dimensional Haar random unitary and $m$ is a diagonal matrix of eigenvalues of $M$. Its embedding in $N$-dimensional form is
\begin{equation}
    \tilde{M}=\begin{bmatrix}W&0\\0&0 \end{bmatrix}\begin{bmatrix}m&0\\0&0 \end{bmatrix}\begin{bmatrix}W^\dagger&0\\0&0 \end{bmatrix} \,.
\end{equation}
To evaluate the integral $I_1$ exactly, we will consider the Harish-Chandra-Itzykson-Zuber (HCIZ) integral \cite{Harish-Chandra:1957dhy, Itzykson:1979fi}.  We can thus write the integral result as
\begin{equation}
    I_1=c_N\frac{\det[\exp(i \tilde{m}_i O_j)]_{1\leq i,j \leq N}}{i^{(N^2-N)/2}\Delta_N(\tilde{m})\Delta_N(O)} \,,
    \label{eq-I1}
\end{equation}
where $\tilde{m}_i=m_i$ for $1\leq i\leq K$, $\tilde{m}_i=0$ for $K+1\leq i\leq N$, and $c_N=\prod^{N-1}_{i=1}i!$. This is only a formal result, since the Vandermonde determinant in the denominator vanishes. We will keep it as is for now, but we will deal with the integral $I_1$ properly below. Similar to the first integral, we use the HCIZ integral to evaluate $I_2$, which gives
\begin{equation}
    I_2=c_K\frac{\det[\exp(-i m_i \lambda_j)]_{1\leq i,j \leq K}}{(-i)^{(K^2-K)/2}\Delta_K(m)\Delta_K(\lambda)} \,,
    \label{eq-I2}
\end{equation}
where $c_K=\prod^{K-1}_{i=1}i!$.
Unlike for $I_1$, this integral is perfectly well-behaved. At this stage, nothing depends on $W$ anymore and its integral gives unity (considering the normalized Haar measure). 

After combining (\ref{eq-I1}) and (\ref{eq-I2}), the joint probability density is then
\begin{equation}
    \mu(\lambda)=\int  \frac{\Delta_K(\lambda)\Delta_K(m)c_N c_Kc_\Delta}{i^{(N^2-N)/2}(-i)^{(K^2-K)/2} } \frac{\det[\exp(i \tilde{m}_i O_j)]_{1\leq i,j \leq N}}{\Delta_N(\tilde{m})} \frac{\det[\exp(-i m_i \lambda_j)]_{1\leq i,j \leq K}}{\Delta_N(O)}d_K m\,.
\end{equation}
As mentioned previously, the term $\frac{\det[\exp(i \tilde{m}_i O_j)]_{1\leq i,j \leq N}}{\Delta_N(\tilde{m})}$ is problematic, and is actually of the form $\frac{0}{0}$ \footnote{In the following, we will consider for simplicity eigenvalues of $O$ being non-degenerate. However, a straightforward generalization can be made for operators with a degenerate spectrum (e.g. the spin operators). We find divergences of the type $\frac{0^2}{0^2}$ and we can still apply the confluent Vandermonde limit that we use in our discussion, by also taking derivatives with respect to the elements for which the degeneracy of eigenvalues of $O$ occurs.}, . Following \cite{Faraut2015} (essentially the confluent Vandermonde limit), this term can be further simplified to
\begin{eqnarray}
    &&\lim_{\tilde{m}_{i>K}\rightarrow 0} \frac{\det[\exp(i \tilde{m}_i O_j)]_{1\leq i,j \leq N}}{\Delta_N(\tilde{m})}=\lim_{\tilde{m}_{i>K}\rightarrow 0} \frac{\det[\exp(i O_i\tilde{m}_j)]_{1\leq i,j \leq N}}{\Delta_N(\tilde{m})}\notag \\
    &=&\frac{i^{(K-N)(K-N+1)/2}}{c_{N-K}\Delta_K(m)\prod^K_{i=1} m_i^{N-K} } \det 
    \begin{pmatrix}
        1&\cdots&1\\
        O_1&\cdots&O_N\\
        \vdots&\cdots&\vdots\\
        O_1^{N-K-1}&\cdots&O_N^{N-K-1}\\
        e^{iO_1m_1}&\cdots&e^{i O_N m_1}\\
        \vdots&\cdots&\vdots\\
        e^{iO_1m_K}&\cdots&e^{i O_N m_K}\\
    \end{pmatrix}_{N\times N}\,,
\end{eqnarray}
where in the first equal sign we applied the transpose operation, which does not affect the determinant, and $c_{N-K}=\prod^{N-K-1}_{i=1}i!$. The determinant above is known to be related to the determinant made of divided differences:\footnote{In what follows, we assume the eigenvalues have been ordered and $O_i$ increases as $i$ increases.}
\begin{equation}
    \det 
    \begin{pmatrix}
        1&\cdots&1\\
        O_1&\cdots&O_N\\
        \vdots&\cdots&\vdots\\
        O_1^{N-K-1}&\cdots&O_N^{N-K-1}\\
        e^{iO_1m_1}&\cdots&e^{i O_N m_1}\\
        \vdots&\cdots&\vdots\\
        e^{iO_1m_K}&\cdots&e^{i O_N m_K}\\
    \end{pmatrix}=\prod_{0<j-i\leq N-K}(O_j-O_i) \det(f_i[O_j,\dots,O_{j+N-K}])_{1\leq i,j\leq K}
\end{equation}
where $f_i[O_j,\dots,O_{j+N-K}]$ is the exponential of the divided difference with $f_i[\cdot]:=e^{im_i [\cdot]}$. 
Therefore, the joint probability density can be written as
\begin{eqnarray}
\mu(\lambda)&=& \frac{\Delta_K(\lambda)c_N c_K c_\Delta^2c_M}{\Delta_N(O)c_{N-K}} i^{(K^2-KN)}\prod_{0<j-i\leq N-K}(O_j-O_i)\times\\ \nonumber
&&\int
\frac{\det(f_i[O_j,\dots,O_{j+N-K}])_{1\leq i,j\leq K}\det[\exp(-i m_i \lambda_j)]_{1\leq i,j \leq K}}{\prod^K_{i=1} m_i^{N-K} }  d_K m \,.
\end{eqnarray}
Next, we use the Andr\'eief identity \cite{andreief1883note} to evaluate the $d_K m$ integral:
\begin{eqnarray}
    &&\int
\frac{\det(f_i[O_j,\dots,O_{j+N-K}])_{1\leq i,j\leq K}\det[\exp(-i m_i \lambda_j)]_{1\leq i,j \leq K}}{\prod^K_{i=1} m_i^{N-K} }  d_K m\notag \\
&=&K!\det\left( \int \frac{1}{m^{N-K}} \exp(im[O_j,\dots,O_{j+N-K}])\exp(-im\lambda_l)dm \right)_{1\leq j,l \leq K}\,.
\end{eqnarray}
Consider an integral representation of the divided difference \cite{de1978practical}
\begin{equation}
    f[x_1,\dots,x_n]=\frac{1}{(n-1)!}\int f^{(n-1)}(t) M_n(x_1,\dots,x_n;t)dt \,,
\end{equation}
where $M_n(x_1,\dots,x_n;t)$ is called the cardinal (fundamental) $B$-spline function, which is the normalized $B$-spline function. Therefore, the above integral has a nice expression:
\begin{eqnarray}
    &&K!\det\left( \int \frac{1}{m^{N-K}} \exp(im[O_j,\dots,O_{j+N-K}])\exp(-im\lambda_l)dm \right)_{1\leq j,l \leq K}\\
    &=& K!\det\left( \int \frac{(im)^{N-K}}{m^{N-K}(N-K)!} \int e^{imt} M_{N-K+1}(O_j,\dots,O_{j+N-K};t)e^{-im\lambda_l}dtdm\right)_{1\leq j,l\leq K}\,.
\end{eqnarray}
The $dm$ integral above simply gives us a Dirac delta function $\delta(t-\lambda_l)$. Therefore, the integral results in a nice form
\begin{equation}
    K!\frac{i^{(N-K)K}(2\pi)^K}{((N-K)!)^K}\det \left( M_{N-K+1}(O_j,\dots,O_{j+N-K};\lambda_l) \right)_{1\leq j,l\leq K} \,.
\end{equation}
The joint probability density reaches a simple form in terms of the cardinal $B$-spline function:
\begin{eqnarray}
    \mu(\lambda)&=&\frac{(2\pi)^K\Delta_K(\lambda)c_N c_K c_\Delta^2 c_M}{\Delta_N(O)c_{N-K}} \frac{K!\prod_{0<j-i\leq N-K}(O_j-O_i)}{((N-K)!)^K}\nonumber\\
    &&\times\det \left( M_{N-K+1}(O_j,\dots,O_{j+N-K};\lambda_l) \right)_{1\leq j,l\leq K}\,,
\end{eqnarray}
where factors involving $i$ are canceled and left with a non-negative probability density.
From its definition, the cardinal $B$-spline function is only non-vanishing if
\begin{equation}
    O_j\leq\lambda_l<O_{j+N-K}\,.
\end{equation}
For the matrix $M_{N-K+1}(O_j,\dots,O_{j+N-K};\lambda_l)_{1\leq j,l\leq K}$, such a constraint for the diagonal elements when $j=l$ shall familiarize us:
\begin{equation}
    O_j\leq\lambda_j<O_{j+N-K}\,,
\end{equation}
which not only assembles Cauchy's interlacing theorem (also known as the Poincaré separation theorem), but also provides a stronger condition.
The cardinal $B$-spline function admits a recursion relation through convolution \cite{de1978practical}:
\begin{equation}
    M_{p+1}(x_1,\dots,x_p;x)=(\underbrace{M^{(1)}_0\star \cdots \star M^{(p)}_0}_p)(y)\,,
\end{equation}
where we use the upper indices to denote each section of the knots:
\begin{equation}
    M^{(i)}_0(y)\equiv M^{(i)}_0(x_i,x_{i+1};y)\,.
\end{equation}
Each $M^{(i)}_0(y)$ can be understood from the probability perspective as a uniform distribution:
\begin{equation}
    M^{(i)}_0(x_i,x_{i+1};y)=
    \begin{cases} \frac{1}{x_{i+1}-x_i},~~~\mathrm{if}~x_i\leq y< x_{i+1}\\
    0,~~~\mathrm{otherwise}
    \end{cases}\,,
\end{equation}
with mean and variance as
\begin{equation}
    \mu_i=\frac{O_{i+1}+O_i}{2},~~~\sigma^2_i=\frac{1}{12}(O_{i+1}-O_i)^2\,.
\end{equation}
Using such a convolution definition, we have
\begin{equation}
    M_{N-K+1}(O_j,\dots,O_{j+N-K};\lambda_l)=\underbrace{M_0(O_j,O_{j+1};\lambda_l)\star \dots \star M_0(O_{j+N-K-1,O_{j+N-K}};\lambda_l)}_{N-K}\,.
\end{equation}
Such a convolution expression can be regarded as a probability distribution of a sum of $N-K$ randomly independent variables with identical uniform distributions. This is the farthest we can get by making no approximations.

If we start taking limits, the expression can be further simplified. In the limit $N\gg K$, we can use the central limit theorem \cite{Lindeberg1922EineNH}, and the probability distribution converges to a Gaussian distribution \cite{Unser1992OnTA}:
\begin{equation}
    M_{N-K+1}(O_j,\dots,O_{j+N-K};\lambda_l) \xrightarrow[]{N\gg K} \frac{1}{\sqrt{2\pi \sigma_j^2}}\exp\left( -\frac{(\lambda_l-\bar{O}_j)^2}{2\sigma_j^2} \right)\,,
\end{equation}
where $\bar{O}_j=\sum_{m}^{N-K-1}\frac{1}{2}(O_{j+m}+O_{j+m+1})$, $\sigma^2_j=\frac{1}{12} \sum_{m}^{N-K-1}(O_{j+m+1}-O_{j+m})^2$. To proceed, let us assume that $\sigma_i^2 \approx \sigma^2$ is constant for all $1\leq j \leq K$.\footnote{This will not always be true for arbitrary choices of $O$. We will in any case provide a more formal derivation of the GUE distribution in the next section, here we make this assumption merely to illustrate how a GUE distribution can arise.} This would lead to 
\begin{equation}
    \mu(\lambda)\propto (2\pi)^K \frac{c_N c_\Delta^2 c_M }{c_{N-K}}\frac{\Delta_K(\lambda)^2 \Delta_K(\bar{O}/\sigma^2)}{\Delta_N(O)} e^{-\sum_l \lambda_l^2/2\sigma^2}e^{-\sum_j \bar{O}_j^2/2\sigma^2}\,,
    \label{eq-exactPlambda}
\end{equation}
where we used that \cite{tao2023topics}
\begin{equation}
    \det(e^{\lambda_l \bar{O}_j/\sigma^2})=\frac{1}{\prod_{x=1}^K x!}\Delta_K(\lambda)\Delta_K(\bar{O}/\sigma^2)+O(\Delta_K(\bar{O}/\sigma^2)) \,.
\end{equation}
We find that the result (\ref{eq-exactPlambda}) indeed resembles the joint probability distribution for the GUE ensemble. But we will now provide a more precise derivation in the large $N$ limit.

\subsubsection{Large-N expansion}
\label{sec: large-N}
In the above discussions, we were able to simplify the joint probability distribution, and in the limit $N\gg K$, found the first indication that a GUE probability distribution arises. To study the joint probability distribution in the large-$N$ limit, we will start from (\ref{eq-IntRep of delta}), and consider a large-$N$ expansion of the Haar integrals instead of the exact HCIZ integral. To achieve this, we introduce a scaled auxiliary matrix as 
\begin{equation}
    M'=N\tilde{M} \,,
\end{equation}
so that
\begin{equation}
    \mu(\lambda)=c_\Delta c_M\int \Delta^2_K(\lambda) e^{\frac{i}{N}\Tr(M'(PUOU^\dagger P-\tilde{V}\tilde{\Lambda}\tilde{V}^\dagger))}dM_KdU_NdV_K \,.
\end{equation}
Since $P\tilde{M}P=\tilde{M}$, we first focus on the $dU_N$ integral:
\begin{equation}
    I_3=\int e^{\frac{i}{N}\Tr(M'UOU^\dagger )}dU_N\,,
\end{equation}

We will consider the logarithm of $I_3$, which is essentially the cumulant generating function. The reason is that cumulants are ordered in the $1/N$ expansion, while, instead, moments would mix terms of different orders of $N$. We, therefore, focus on
\begin{equation}
    \log(I_3)=\frac{i}{N}\kappa_1-\frac{1}{2N^2}\kappa_2+\dots\,.
    \label{eq-cumulant generating}
\end{equation}
By computing each moment, we can find the cumulant through the cumulant-moment relation. By using the Weigarten calculus \cite{collins2003moments}, 
\begin{eqnarray}
    \kappa_1=m_1=\int \Tr(M'UOU^\dagger)dU=N\Tr(M)\bar{O}\,,
\end{eqnarray}
where $\bar{O}=\Tr(O)/N$. The second moment is
\begin{equation}
    m_2=\int \Tr(M'UOU^\dagger) \Tr(M'UOU^\dagger)dU\approx N^2 \Tr(M)^2 \bar{O}^2+N \Tr(M^2)\sigma_O^2\,,
\end{equation}
where we neglected the sub-leading terms of $\mathcal{O}(N^0)$, and $\sigma_O^2=\Tr(O^2)/N-\bar{O}^2$. The second cumulant is
\begin{equation}
    \kappa_2=m_2-m_1^2=N \Tr(M^2)\sigma_O^2\,.
\end{equation}
One can of course continue to check the third cumulant, which yields:
\begin{eqnarray}
    \kappa_3&=&\frac{1}{(N^2-1)(N^2-4)}2\left[N^2\Tr(O^3)-3N\Tr(O^2)\Tr(O)+2\Tr(O)^3\right]\\
    &&\times \left[N^2\Tr(M^3)-3N\Tr(M^2)\Tr(M)+2\Tr(M)^3\right]\,.
\end{eqnarray}
Counting the appropriate scalings, we find
\begin{equation}
    \mathcal{O}(\kappa_3)=\mathcal{O}(NK)+\mathcal{O}(K^2)+\mathcal{O}(K^3N^{-1})\,.
\end{equation}
Here, we assumed that the entries of $m$ are order-1 numbers, so $\Tr(M)\sim \Tr(M^2)\sim \mathcal{O}(K)$. Now, consider the next term in (\ref{eq-cumulant generating}), it scales as
\begin{equation}
    \frac{i}{3!N^3}\kappa_3 \sim \mathcal{O}(KN^{-2})+\mathcal{O}(K^2N^{-3})+\mathcal{O}(K^3N^{-4})\,,
\end{equation}
which can be neglected when $\alpha=K/N\ll 1$. Similar observations were made in \cite{Wang:2023qon}.

We can thus conclude that
\begin{equation}
    \log(I_3) \approx i \Tr(M)\bar{O}-\frac{1}{2N}\Tr(M^2)\sigma_O^2\,.
\end{equation}
Therefore, the scaling of $\log(I_3)$ is $\sim K+\alpha+O(\alpha^2)$. All the higher-order cumulants are subleading in the large-$N$ limit as long as $\alpha=K/N\ll1$.
The joint probability density can thus be written as
\begin{equation}
    \mu(\lambda)=c_\Delta c_M \Delta^2_K(\lambda)\int e^{i \Tr(M)\bar{O} -i\Tr(MV\Lambda V^\dagger)-\frac{1}{2N}\Tr(M^2)\sigma_O^2} dV_K dM_K\,.
\end{equation}

We then proceed by evaluating the $dM_K$ integral:

\begin{eqnarray}
    &&\int e^{i\Tr[M(\bar{O}\mathbb{1}_K-V\Lambda V^\dagger)]-\frac{1}{2N}\Tr(M^2)\sigma_O^2} dM_K\\
    &=& \int e^{ i\sum\limits_{i,j}M_{ij}(\bar{O}\mathbb{1}_K-V\Lambda V^\dagger)_{ji} -\frac{\sigma_O^2}{2N}\sum\limits_{i,j} M_{ij}M_{ji}}\prod^K_{i,j=1}dM_{ij}\\
    &=& \int e^{-\frac{\sigma_O^2}{2N} \sum\limits_{i,j} \left( M_{ij}M_{ji}-\frac{2iN}{\sigma_O^2} M_{ij}(\bar{O}\mathbb{1}_K-V\Lambda V^\dagger)\right)}\prod^K_{i,j=1}dM_{ij}\\
    &=& \int e^{-\frac{\sigma_O^2}{2N} \sum\limits_{i,j} 
    \left[\left(  M_{ij}-\frac{iN}{\sigma_O^2}(\bar{O}\mathbb{1}_K-V\Lambda V^\dagger)_{ij} \right) \left(  M_{ji}-\frac{iN}{\sigma_O^2}(\bar{O}\mathbb{1}_K-V\Lambda V^\dagger)_{ji} \right)+\frac{N^2}{\sigma_O^4} (V\Lambda V^\dagger-\bar{a}\mathbb{1})_{ij} (V\Lambda V^\dagger-\bar{O}\mathbb{1})_{ji}\right]
    }\prod^K_{i,j=1}dM_{ij}\nonumber\\
    &&\\
    &=& e^{-\frac{N}{2\sigma^2_O} \sum\limits_{i,j} (V\Lambda V^\dagger-\bar{O}\mathbb{1})_{ij} (V\Lambda V^\dagger-\bar{O}\mathbb{1})_{ji}}\int e^{-\frac{\sigma_O^2}{2N} \sum\limits_{i,j}\mathfrak{M}_{ij}\mathfrak{M}_{ji}}d \mathfrak{M}\\
    &=& 2^{K/2} \pi^{K^2/2} \left( \frac{N}{\sigma_O^2} \right)^{K^2/2} e^{-\frac{N}{2\sigma^2_O} \Tr\left[ (V\Lambda V^\dagger-\bar{O}\mathbb{1})^2\right]}\,,
\end{eqnarray}
where in the fifth line, we define $\mathfrak{M}=M-\frac{iN}{\sigma_O^2}(\bar{O}\mathbb{1}_K-V\Lambda V^\dagger)$ and $d\mathfrak{M}=dM$. Notice that
\begin{equation}
    \Tr\left[ (V\Lambda V^\dagger-\bar{O}\mathbb{1})^2\right]=\Tr\left[ (V(\Lambda -\bar{O}\mathbb{1})V^\dagger)^2\right]=\Tr\left[ (\Lambda -\bar{O}\mathbb{1})^2\right]\,.
\end{equation}
Therefore, the $dM_K$ integral yields
\begin{equation}
    2^{K/2} \pi^{K^2/2} \left( \frac{N}{\sigma_O^2} \right)^{K^2/2} e^{-\frac{N}{2\sigma^2_O} \sum\limits_i (\lambda_i-\bar{O})^2}\,.
\end{equation}
The probability density is independent of the Haar integral $dV_K$. Setting $N=K/\alpha$ and writing constants $c_\Delta$ and $c_M$ explicitly, the joint probability density is
\begin{equation}
    \mu(\lambda)=\frac{1}{(2\pi)^{K/2} \prod^K_{j=0} j!}\Delta^2_K(\lambda)e^{-\frac{K}{2\alpha\sigma^2_O} \sum\limits_i (\lambda_i-\bar{O})^2}= \mu(\text{GUE})_{\sigma_O, \bar{O}}\,.
    \label{eq-jpdf in LargeN}
\end{equation}

This is one of our main results. This exactly reproduces the normalized joint probability distribution of GUE (see e.g. \cite{mehta2004random}), with a shifted center. It follows from the above expression that, in the large-$N$ limit with $\alpha =K/N$ fixed and in the small $\alpha$ limit, the joint probability density of eigenvalues of $PUAU^\dagger P$ reduces to that of the Gaussian unitary ensemble upon a truncation at quadratic order in $\alpha$. Equipped with these, one can easily derive the density of states following standard procedures (i.e., from \cite{mehta2004random}):
\begin{equation}
    \boxed{\rho(\lambda) \xrightarrow[\alpha\ll 1]{N\gg1,~K\gg 1} \frac{1}{2\pi\sigma_O^2\alpha} \sqrt{4\sigma_O^2\alpha-(\lambda-\bar{O})^2}}\,.
\end{equation}
The emergence of the semicircle law should be understood as a convergent result iff all three limits are met.  

\subsubsection{Universality in the two-point eigenvalue correlation}
Knowing from the discussions above that the semicircle law would only show up if all three conditions are satisfied, a natural question to ask is whether these conditions are also necessary for universal correlations in the spectrum. We will see in the numerics in Section \ref{sec: Numerics} that the answer seems to be no: we seem to have universality in the spectral form factor even at large values of $\alpha$ (even surprisingly close to 1).

We will now try to heuristically explain this behavior. First, recall the joint eigenvalue distribution (\ref{eq-jpdf in LargeN}) in the large-$N$ limit. During the derivation, we did not assume any particular property on the spectrum of  $O$. Therefore, such a probability density can be regarded as a universal result as long as we are in the large-$N$ and small-$\alpha$ limit. For such a Gaussian potential $V(\lambda)\propto \sum \lambda^2$, there exist orthogonal polynomials $P_n(\lambda)$ such that
\begin{equation}
    \int P_n(\lambda)P_m(\lambda)e^{-K V(\lambda)}d\lambda=\delta_{nm}\,.
\end{equation}
One can define the function
\begin{equation}
    \psi_n(\lambda)=P_n(\lambda) e^{-K V(\lambda)/2}\,.
\end{equation}
The kernel, which is the square root of the connected part of the two-point spectral correlation function, is given as
\begin{equation}
    \mathrm{Ker}(\lambda_1,\lambda_2)=\frac{1}{K} \sum^{K}_{i=1} \psi_i(\lambda_1)\psi_i(\lambda_2)\,.
\end{equation}
This can be further simplified using the Christoffel-Darboux identity:
\begin{equation}
    \mathrm{Ker}(\lambda_1,\lambda_2)\propto \frac{1}{K} \frac{\psi_K(\lambda_1)\psi_{K-1}(\lambda_2)-\psi_{K-1}(\lambda_1)\psi_{K}(\lambda_2)}{\lambda_1-\lambda_2}\,.
\end{equation}
In the case of a Gaussian potential, the orthogonal polynomials are usually taken to be the Hermite polynomials. Generally, this kernel can be captured by the solution $Y_n$ that solves the Riemann-Hilbert problem \cite{bleher2011lectures}:
\begin{equation}
    \mathrm{Ker}(\lambda_1,\lambda_2)= \frac{e^{-K (V(\lambda_1)+V(\lambda_2)/2)}}{2\pi i (\lambda_1-\lambda_2)} 
    \begin{pmatrix}
        0&&1
    \end{pmatrix}
    Y^{-1}_{K+}(\lambda_2) Y_{K+}(\lambda_1)
    \begin{pmatrix}
        0 \\1
    \end{pmatrix}\,.
\end{equation}
The kernel is shown to be only a function of $\lambda_1,\lambda_2$ and the density of state $\rho(\lambda)$:
\begin{equation}
    \mathrm{Ker}(\lambda_1,\lambda_2)=\frac{1}{2\pi i (\lambda_1-\lambda_2)} \left( e^{\pi i K \int^{\lambda_1}_{\lambda_2} \rho(x)dx} -e^{-\pi i K \int^{\lambda_1}_{\lambda_2} \rho(x)dx}\right)\,.
    \label{eq-kernel}
\end{equation}
Let us consider that the density of states for the projected matrix is smooth. The distance between eigenvalues is $|\lambda_1-\lambda_2|\sim1/K$. Therefore, the density of states within this range can be assumed as a constant $\rho(\bar{\lambda})$ with $\bar{\lambda}=(\lambda_1+\lambda_2)/2$. The integral on the exponential of (\ref{eq-kernel}) can be approximated by
\begin{equation}
    \int^{\lambda_1}_{\lambda_2} \rho(x)dx\approx \rho(\bar{\lambda})(\lambda_1-\lambda_2)\,.
    \label{eq-dof approx}
\end{equation}
To consider the rescaled eigenvalues, we can zoom into the neighborhood of any bulk point (i.e. away from the edge of the spectrum) $\lambda_0$ and redefine:
\begin{equation}
    \lambda_1 \rightarrow\lambda_0+\frac{\lambda_1}{K\rho(\lambda_0)},~~~\lambda_2 \rightarrow\lambda_0+\frac{\lambda_2}{K\rho(\lambda_0)}\,,
\end{equation}
where we assume $\rho(\bar{\lambda})\approx \rho(\lambda_0)$, which is valid as long as the density of states is approximately constant in a small range of eigenvalues.
The kernel can therefore be derived as
\begin{equation}
    \lim_{K,N\rightarrow \infty} \frac{1}{K\rho(\lambda_0)}\mathrm{Ker}(\lambda_0+\frac{\lambda_1}{K\rho(\lambda_0)},\lambda_0+\frac{\lambda_2}{K\rho(\lambda_0)} )=\frac{\sin(\pi(\lambda_1-\lambda_2))}{\pi(\lambda_1-\lambda_2)}\,,
\end{equation}
which leads to the linear ramp and plateau in the connected spectral form factor. Although this sine kernel is well-known for the Gaussian potential \cite{Dyson:1970tza}, it is also valid for a general potential $V(\lambda)$ \cite{bleher2011lectures, Brezin:1993qg, bleher1999semiclassical}. 

Let us quickly recall the derivation process of the sine kernel. The necessary conditions are a well-defined potential $V(\lambda)$ in the polynomial form that admits orthogonal polynomials, and an approximation (\ref{eq-dof approx}) of the density of states. Now, consider relaxing the constraints for the projection matrix so that $\alpha$ can be any finite value between $(0,1)$ while keeping $K$ and $N$ large. The large-$N$ expansion in Section \ref{sec: large-N} would lead to a potential with higher power terms, beyond the quadratic order that we kept at small $\alpha$. However, this should not change the sine kernel, as we just discussed in this section (assuming the projected matrix has a smooth $\rho$), as well as the level repulsion, which is fully determined by the Vandermonde determinant and is independent of the potential. We can therefore conclude the universal properties of the projected matrix
\begin{equation}
   \boxed{ \mathrm{Two~point~ correlation~ function ~of ~}PUOU^\dagger P \xrightarrow[K<N, \mathrm{smooth}~ \rho]{N\gg1,~K\gg 1} \mathrm{sine ~kernel}}\,,
\end{equation}

\begin{equation}
    \boxed{\mathrm{Nearest~ level ~spacing ~of ~}PUOU^\dagger P \xrightarrow[K<N]{N\gg1,~K\gg 1} \mathrm{level~ repulsion~} P(s=0)=0}\,.
\end{equation}

Finally, let us discuss what type of spectrum $f(O)$ would violate these universal properties. The only case we imagine for it to happen is that the input matrix $O$ has a non-smooth density of states and the projected matrix inherits such a non-smoothness \footnote{For a non-smooth function, we mean the non-continuity in its derivatives, especially, in our consideration where there exists a gap or jump in its first derivative.}. This would not happen in the small $\alpha$ limit due to the central limit theorem we derived. But it is possible that the projected matrix will show such an inheritance of non-smoothness as $\alpha$ increases, depending on the properties of the eigenvalues of $O$. For example, in the case of a spin-$1/2$ operator, it is clear that we have non-smoothness when $\alpha>1/2$, and at $\alpha>2/3$ for a spin-$1$ operator. The violation of the smoothness condition clearly spoils our approximation (\ref{eq-dof approx}). The non-continuity in the first derivative of $\rho$ indicates a sudden change of the eigenvalues, which fails to satisfy the assumption that the density of states is a constant in that range of eigenvalues. This leads to a complicated form of the connected piece, different from the sine kernel, resulting in non-typical behavior for the SFF. This explains the SFF plots in Figure \ref{fig:spin_stats2}, when $\alpha>1/2$. For $\alpha<1/2$, but still not parametrically small, there is also clearly a deviation. It would be interesting to see exactly how this happens, given the framework of this section. It is likely that the coefficients in the potential start becoming fine-tuned in such a way that there is a deviation from the sine kernel.

\section{Numerical examples of the projected operators}
\label{sec: Numerics}

In this section, we will study the eigenvalue distribution for the projected operators numerically. More specifically, we will focus on the following quantities:
\begin{itemize}
    \item \textbf{The normalized level density of states}, which is the one-point distribution of the eigenvalues.
    \item \textbf{The normalized level spacing distribution}, which is obtained by ordering the eigenvalues $\lambda_1<\lambda_2<\ldots<\lambda_K$ and looking at the distribution of the jumps $S_i=\lambda_{i+1}-\lambda_i$ between neighboring eigenvalues, normalized by the mean spacing: $s_i=\frac{S_i}{\langle S \rangle}$.
    \item \textbf{The connected spectral form factor}
    \begin{align}
        \text{SFF}(t)=\frac{1}{K^2}\left\{\mathbb{E}\left[\Tr\,e^{iO_Kt}\, \Tr\,e^{-iO_Kt} \right] - \mathbb{E}\left[\Tr\,e^{iO_Kt}\right] \mathbb{E}\left[ \Tr\,e^{-iO_Kt} \right] \right\}\,,
    \end{align}
    where $\mathbb{E}[\cdot]$ stands for the ensemble averaging. Here, the ensemble averaging will be over the Haar ensemble of unitaries.
\end{itemize}

In the following numerical analysis, we will always restrict $O_K$ to the projected subspace in order to get rid of the trivial $N-K$ zero eigenvalues. We choose the operator $O$ of size $1000$ before projection and average over $1000$ extractions of random Haar unitaries. We have considered two typical examples of operator spectra: one with a degenerate spectrum and one with a randomly drawn spectrum. The former case ceases to show RMT behaviors after the transition point where degenerate eigenvalues are recovered.\footnote{There are still some remnants of RMT in this case, but only for a smaller subset of eigenvalues.}. The latter case surprisingly shows RMT behavior for the two-point correlation, at essentially any finite $\alpha<1$.

\subsection{The projection of a spin-$1/2$ operator}
First, we deal with the simplest example, which has been analyzed by Srednicki and Iniguez \cite{SredIni}, but only at the level of the density of states (i.e. for the one-point function). Here we will both summarize their findings and also add our results about spectral correlations. Consider $O$ to be a single spin-$1/2$ operator (this is a standard probe in lattice spin systems).\footnote{To avoid factors of $1/2$, we actually take $O=2\sigma_z$.} As we mentioned above, we will work with $\dim \mathscr{H}=1000$. The operator, written in its eigenbasis, then reads
\begin{align}
    O=
    \begin{pmatrix}
        \mathbb{I}_{\frac{N}{2}} & 0_{\frac{N}{2}}\\
        0_{\frac{N}{2}} & -\mathbb{I}_{\frac{N}{2}}
    \end{pmatrix} \,.
\end{align}
Such an operator has two eigenvalues $\{+1,-1\}$, each with multiplicity $\mu\left( \pm1\right)=N/2$.

We now rotate and then project $O$ as in (\ref{trunc_op}), and examine the eigenvalues for different values of $K$. For this case, it is useful to split the problem into two different ranges of $\alpha$.  

\subsubsection*{Case $\alpha\leq \frac{1}{2}$}

\begin{figure}[htbp]
    \centering
    \begin{subfigure}[b]{0.3\textwidth}
        \includegraphics[width=0.85\textwidth]{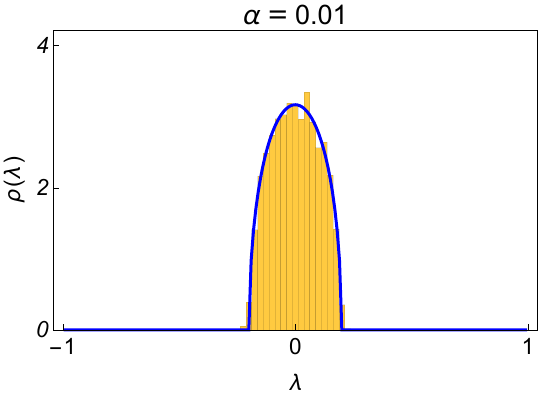}
        \caption{}
    \end{subfigure}
    \hfill
    \begin{subfigure}[b]{0.3\textwidth}
        \includegraphics[width=\textwidth]{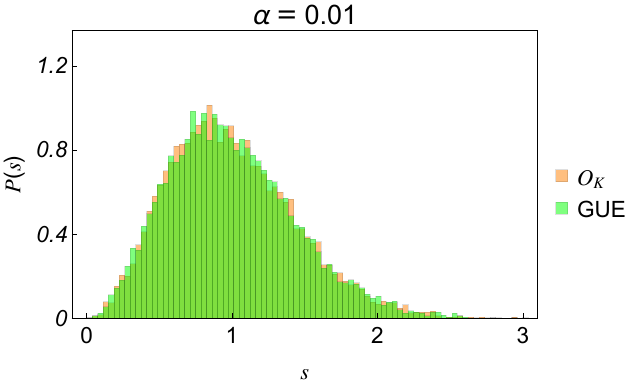}
        \caption{}
    \end{subfigure}
    \hfill
    \begin{subfigure}[b]{0.3\textwidth}
        \includegraphics[width=1.1\textwidth]{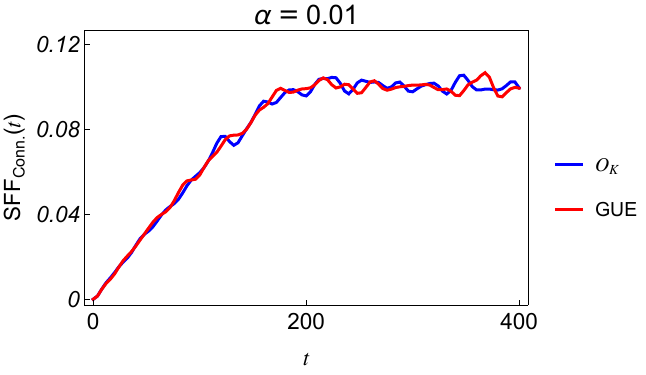}
        \caption{}
    \end{subfigure}
    
    \medskip

    \begin{subfigure}[b]{0.3\textwidth}
        \includegraphics[width=0.85\textwidth]{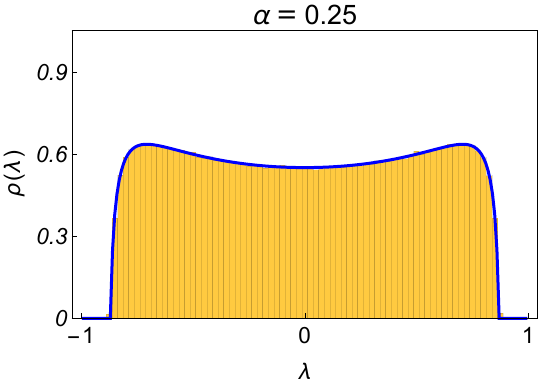}
        \caption{}
    \end{subfigure}
    \hfill
    \begin{subfigure}[b]{0.3\textwidth}
        \includegraphics[width=\textwidth]{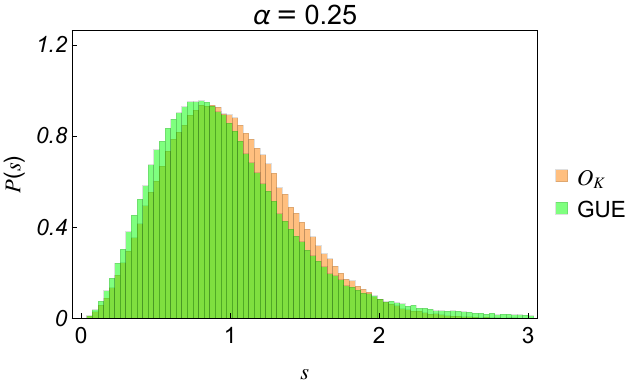}
        \caption{}
    \end{subfigure}
    \hfill
    \begin{subfigure}[b]{0.3\textwidth}
        \includegraphics[width=1.1\textwidth]{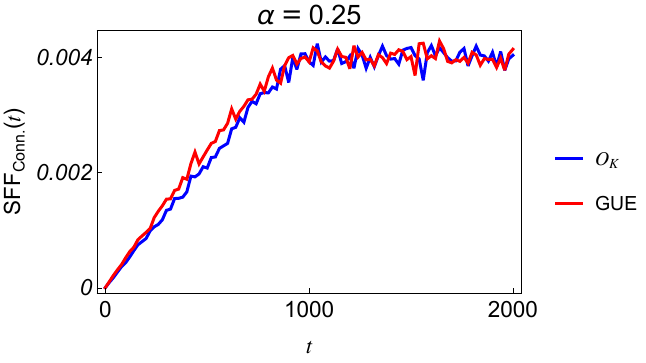}
        \caption{}
    \end{subfigure}
    
    \medskip
    
    \begin{subfigure}[b]{0.3\textwidth}
        \includegraphics[width=0.85\textwidth]{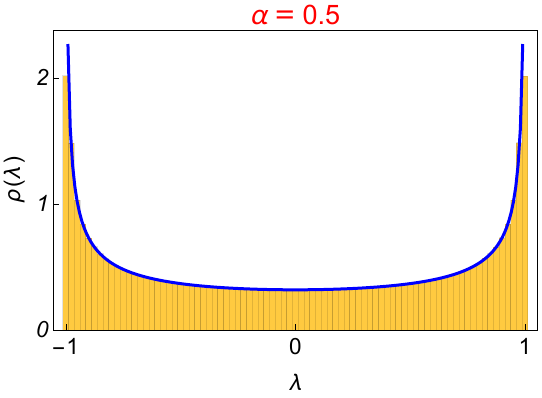}
        \caption{}
    \end{subfigure}
    \hfill
    \begin{subfigure}[b]{0.3\textwidth}
        \includegraphics[width=\textwidth]{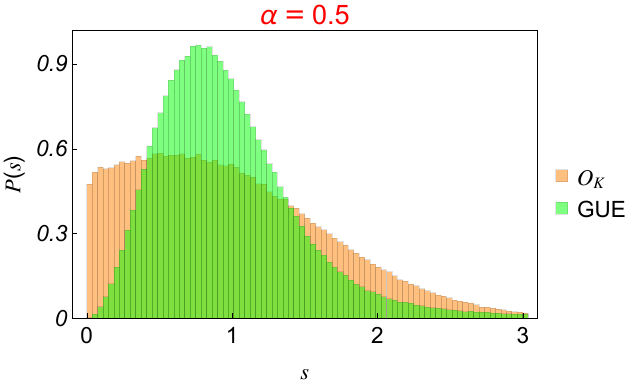}
        \caption{}
    \end{subfigure}
    \hfill
    \begin{subfigure}[b]{0.3\textwidth}
        \includegraphics[width=1.1\textwidth]{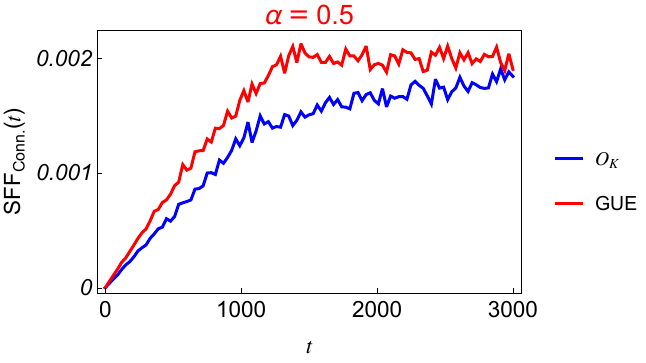}
        \caption{}
    \end{subfigure}
    
    \caption{Different statistics for the projected spin-$1/2$ operator for different values of $\alpha\leq\frac{1}{2}$. Left column: the numerically normalized level density (yellow), the theoretical prediction (blue) of (\ref{Spin_trunc}). Central column: the normalized spacing distribution of $O_K$ (orange), and of the GUE (green). Right column: the connected SFF for $O_K$ (blue), and for the GUE (red).}
    \label{fig:spin_stats}
\end{figure}
In Figure \ref{fig:spin_stats} we show the numerical results for $\alpha \leq\frac{1}{2}$. Let us first focus on the level density. In Sections \ref{sec: DOS from FPT} and \ref{sec: FreeCompression}, we use two different methods from free probability to prove that, in the limit $N\to\infty$ and $K\to\infty$ with $\alpha$ finite,
\begin{align}\label{Spin_trunc}
    \rho(\lambda)=\frac{\sqrt{4\alpha(1-\alpha)-\lambda^2}}{2\pi\alpha(1-\la^2)},\qquad \alpha\leq\frac{1}{2} \,.
\end{align}
These results were already obtained in \cite{collins2005product, SredIni}.
\eqref{Spin_trunc} is plotted as a solid blue line in the left column of Figure \ref{fig:spin_stats}, and we can see that it closely matches the numerically computed distribution.

When $\alpha\ll 1$, the distribution (\ref{Spin_trunc}) can be approximated as the Wigner semicircle
\begin{align}\label{Small_alpha}
    \rho_{\alpha\ll 1}\left( \lambda\right) \simeq \frac{1}{2\pi\alpha}\sqrt{4\alpha-\lambda^2} \,.
\end{align}
This is the regime where the analytic results of the previous section apply, and it is no surprise that we find the semicircle, since it is the density of states of the GUE ensemble. For comparison, we consider a Wigner semicircle distribution for a reference $K\times K$ GUE Hermitian random matrix $\mathbf{H}$ with $\E[\mathbf{H}_{ij}]=0$ and $\E\left[| \mathbf{H}_{ij}|^2\right]=\sigma^2$, which reads
\begin{align}\label{Wigner}
     \rho_W\left( \lambda\right) = \frac{1}{2\pi K\sigma^2}\sqrt{4K\sigma^2-\lambda^2}\,.
\end{align}
This implies a substitution of $\sigma^2=1/N$ for the variance of the GUE matrix elements by directly matching (\ref{Small_alpha}) and  (\ref{Wigner}).

The GUE behavior for $O_K$ at small $\alpha$ is also evident at the level of spectral correlations. To see this, we compared the level spacing distributions of $O_K$ and of $\mathbf{H}$. As can be seen in the central column of Figure \ref{fig:spin_stats}, for $\alpha\ll 1$ the two distributions appear to be the same. The connected SFFs plotted in the right column also perfectly match the prediction from the GUE. 

When $\alpha$ gets bigger, the various statistics begin to differ from the GUE case. We can see analytically that the distribution (\ref{Spin_trunc}) deviates from the Wigner semicircle by studying its second derivative
\begin{align}
    \rho''\left(0 \right)=\frac{2\sqrt{\alpha\left(1-\alpha\right)}}{\pi\alpha}-\frac{1}{4\pi\alpha\sqrt{\alpha\left(1-\alpha \right)}}\,,
\end{align}
which is negative for $\alpha<\alpha_c=\left(2-\sqrt{2} \right)/4$ but turns positive for $\alpha>\alpha_c$. When increasing $\alpha$, something peculiar happens at $\alpha=1/2$: the eigenvalue distribution becomes the famous arcsine distribution \footnote{As we will show in Section \ref{sec: freecompression for spinhalf}, this distribution is essentially the result of free probability using free compression. It is a distribution of the sum of two free spin-$1/2$ operators.}
\begin{align}
    \rho_{\alpha=1/2}(\lambda)=\frac{1}{\pi \sqrt{1-\lambda^2}}
\end{align}
with integrable singularities at $\lambda=\pm1$ (see (g) in Figure \ref{fig:spin_stats}). Even more interesting is the distribution of the level spacings at $\alpha=1/2$. As we can see in subfigure (h), the distribution differs significantly from the Wigner-Dyson distribution for the GUE, and starts to show level clustering (exactly at $\alpha=1/2$, $P(s=0)$ still vanishes so level repulsion is maintained but at threshold). The SFF, on the other hand, is more interesting. It appears to have a ramp with a smaller slope, followed by a second ramp with a yet smaller slope, before reaching the plateau.

\subsubsection*{Case $\alpha > \frac{1}{2}$}

When the truncation goes beyond half of the full operator, there is an interesting behavior of the spectrum: as we can see in the left column of Figure \ref{fig:spin_stats2}, the truncated matrix $O_K$ starts to recover the eigenvalues $\lambda=\pm1$ of the original $N\times N$ operator $O$. In particular, for every value of $K>N$ the degeneracy of such eigenvalues is exactly $K-N$, for each of them. The general proof of eigenvalue recovery is given in Appendix \ref{App: Eigenvalue Recovery}. In addition to this deterministic discrete spectrum, we have a continuous spectrum that resembles that of the range $\alpha<1/2$: its distribution is indeed given by (\ref{Spin_trunc}) modulo a rescaling $\alpha\rightarrow1-\alpha$:
\begin{align}\label{Spin_trunc2nd}
    \rho\left( \lambda\right)=\frac{\sqrt{4\alpha\left(1-\alpha\right)-\lambda^2}}{2\pi(1-\alpha)\left(1-\lambda^2\right)} \,.
\end{align}

\begin{figure}[htbp]
    \centering

    \begin{subfigure}[b]{0.3\textwidth}
        \includegraphics[width=0.85\textwidth]{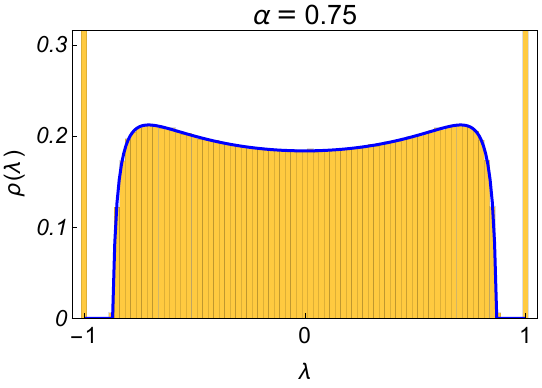}
        \caption{}
    \end{subfigure}
    \hfill
    \begin{subfigure}[b]{0.3\textwidth}
        \includegraphics[width=\textwidth]{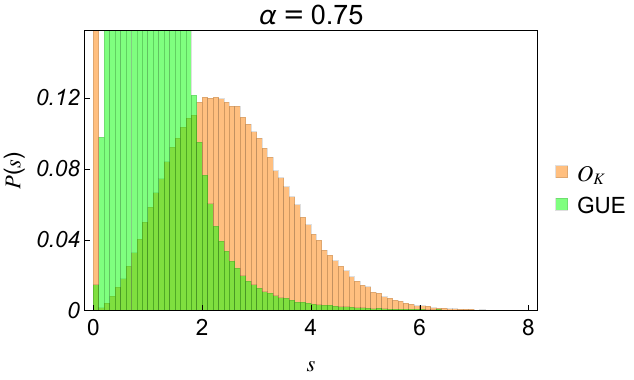}
        \caption{}
    \end{subfigure}
    \hfill
    \begin{subfigure}[b]{0.3\textwidth}
        \includegraphics[width=1.1\textwidth]{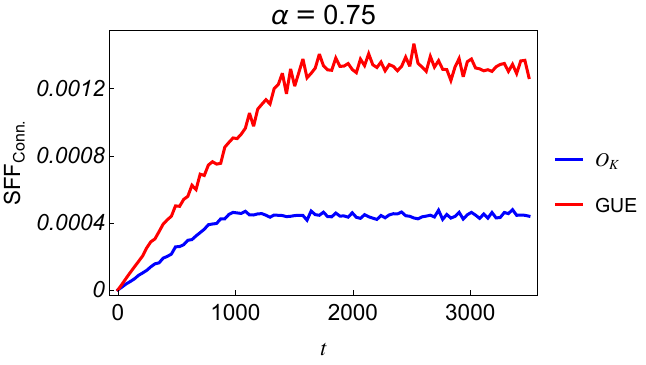}
        \caption{}
    \end{subfigure}
    
    \medskip

    \begin{subfigure}[b]{0.3\textwidth}
        \includegraphics[width=0.87\textwidth]{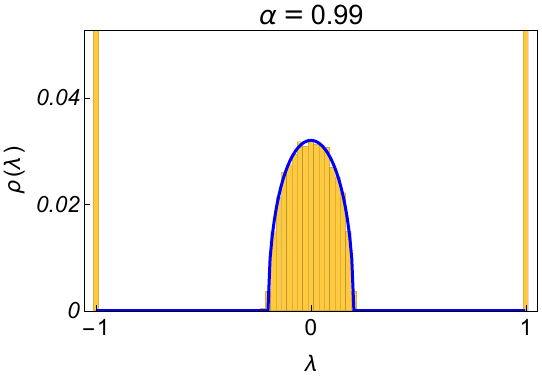}
        \caption{}
    \end{subfigure}
    \hfill
    \begin{subfigure}[b]{0.3\textwidth}
        \includegraphics[width=1.05\textwidth]{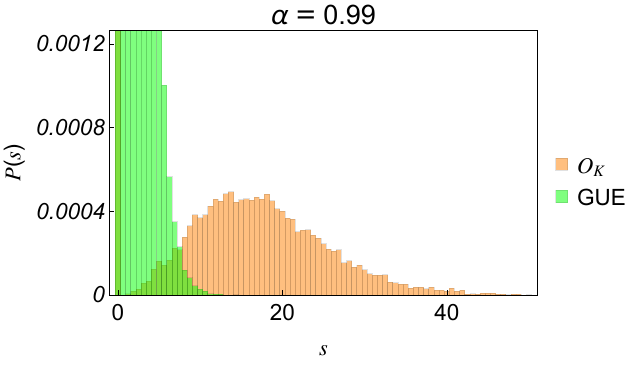}
        \caption{}
    \end{subfigure}
    \hfill
    \begin{subfigure}[b]{0.3\textwidth}
        \includegraphics[width=1.1\textwidth]{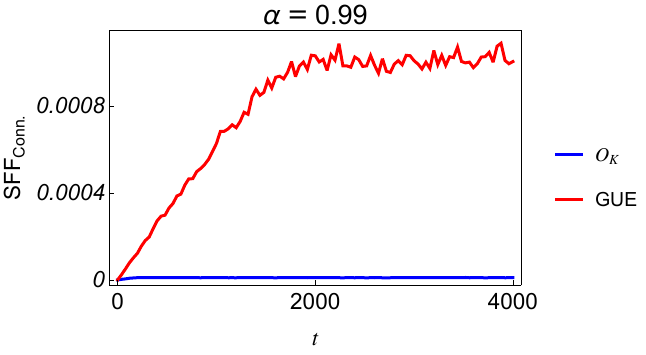}
        \caption{}
    \end{subfigure}
    
    \medskip
    
    \begin{subfigure}[b]{0.3\textwidth}
        \includegraphics[width=0.85\textwidth]{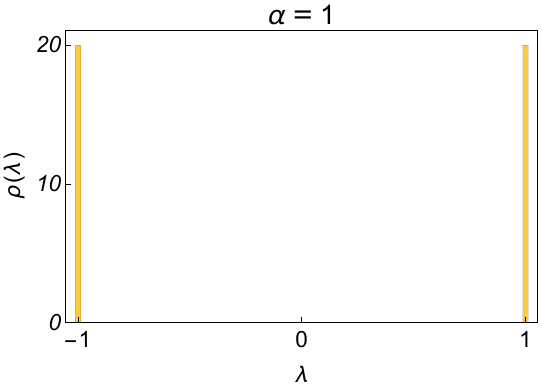}
        \caption{}
    \end{subfigure}
    \hfill
    \begin{subfigure}[b]{0.3\textwidth}
        \includegraphics[width=\textwidth]{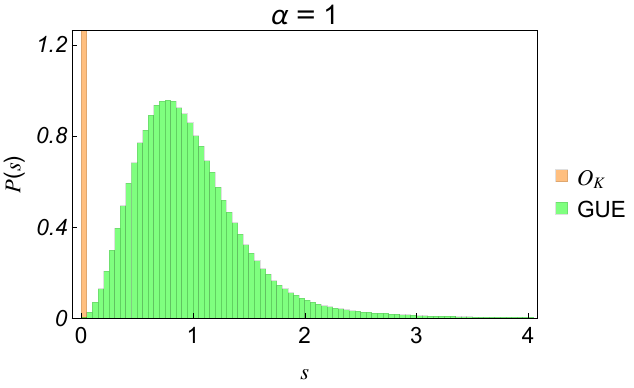}
        \caption{}
    \end{subfigure}
    \hfill
    \begin{subfigure}[b]{0.3\textwidth}
        \includegraphics[width=1.1\textwidth]{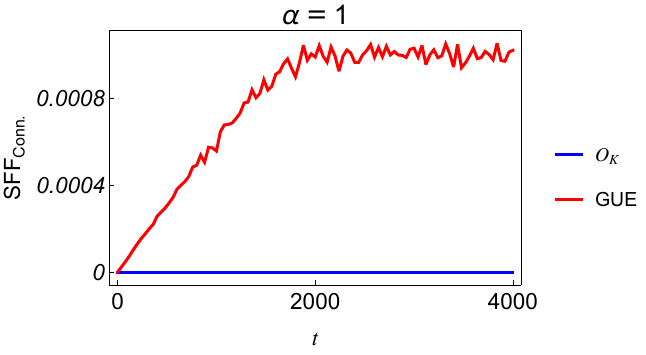}
        \caption{}
    \end{subfigure}
    \caption{Different statistics for the projected spin-$1/2$ operator for different values of $\alpha>\frac{1}{2}$. Left column: the numerically normalized level density (yellow), the theoretical prediction from (\ref{Spin_trunc2nd}) (also from free compression in Section \ref{sec: FreeCompression} ). Central column: the normalized spacing distribution of $O_K$ (orange), and of the GUE (green). Right column: the connected SFF for $O_K$ (blue), and for the GUE (red)}
    \label{fig:spin_stats2}
\end{figure}

In the left column of Figure \ref{fig:spin_stats2}, the density of states is again shown as a solid blue line.
There are algebraic reasons for the behavior of the eigenvalues in the $\alpha>1/2$ range, which we will come back to shortly. Authors in \cite{SredIni} suggest that the value $\alpha=1/2$ might have something special, because at that point there is a sort of ``phase transition" from a continuous spectrum to a superposition of a continuous spectrum and a discrete one. What we will find in the next paragraphs is that $\alpha=1/2$ is indeed special, but its value is not universal: a different operator $O$ may have such a transition at a different value of $\alpha$; it also may have more than one value of $\alpha$ for which similar transitions occur.

If we look at the spacing statistics in the central column in Figure \ref{fig:spin_stats2}, we see a severe deviation from the GUE. This time we have trivial spacings which are exactly zero, due to large degeneracies at $\pm1$. In addition to that, we have a continuous spacing distribution that recovers the property of eigenvalue repulsion. This means that the eigenvalues of the continuous part of the spectrum still have chaotic properties. Notice that the SFFs also exhibit large deviations from the GUE, despite still showing the ramp-plateau behavior. We expect that throwing out the degenerate eigenvalues and rescaling would lead to a standard ramp-plateau compatible with the GUE.

The last thing to mention is that, for $\alpha=1$, the distributions are trivial, since $O_K$ coincides with $O$.

\subsection{The projection of an operator with non-degenerate eigenvalues}
We will now consider an operator $O$ with a non-degenerate spectrum. To build such an operator, we can extract the $N$ eigenvalues independently from a probability distribution function $\mu(\la)$, so that the probability of having equal eigenvalues is zero, even though there is no eigenvalue repulsion since the $N$ extractions are uncorrelated. Take for example $\mu$ to be the flat distribution in the interval $[0,1]$:
\begin{align}
    \mu(\lambda)=
    \begin{cases}
        1,\qquad &\lambda\in[0,1]\\
        0, &\text{otherwise}
    \end{cases}\,.
\end{align}

\begin{figure}[htbp]
    \centering
    
    \begin{subfigure}[b]{0.3\textwidth}
        \includegraphics[width=0.85\textwidth]{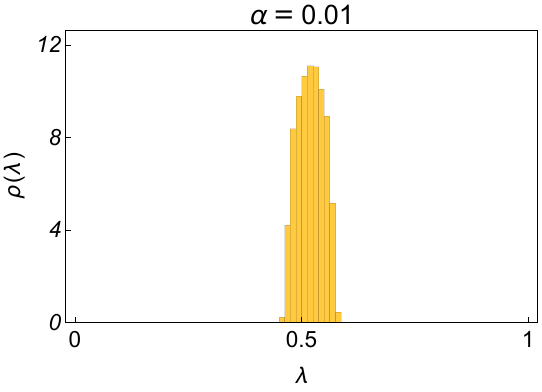}
        \caption{}
    \end{subfigure}
    \hfill
    \begin{subfigure}[b]{0.3\textwidth}
        \includegraphics[width=\textwidth]{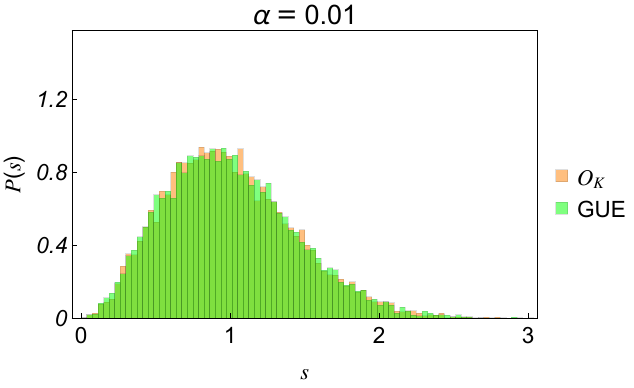}
        \caption{}
    \end{subfigure}
    \hfill
    \begin{subfigure}[b]{0.3\textwidth}
        \includegraphics[width=1.1\textwidth]{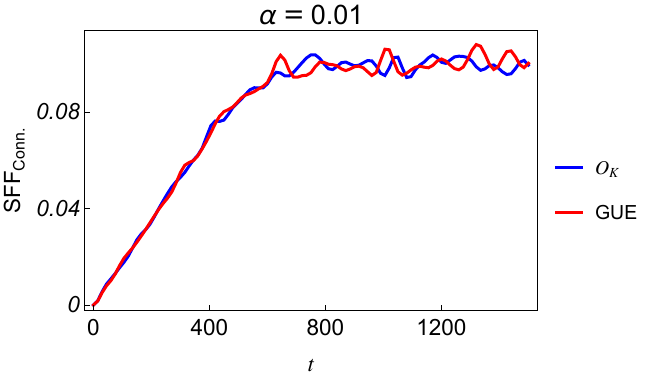}
        \caption{}
    \end{subfigure}
    
    \medskip
    
    \begin{subfigure}[b]{0.3\textwidth}
        \includegraphics[width=0.85\textwidth]{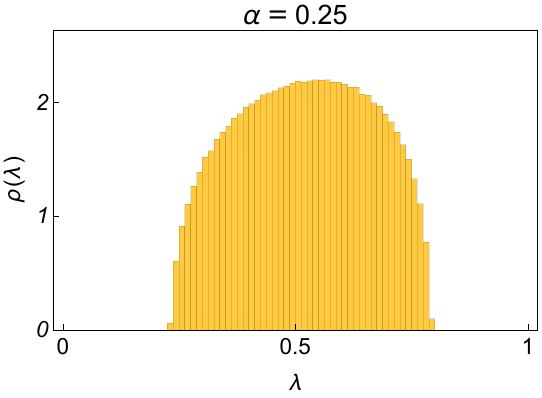}
        \caption{}
    \end{subfigure}
    \hfill
    \begin{subfigure}[b]{0.3\textwidth}
        \includegraphics[width=\textwidth]{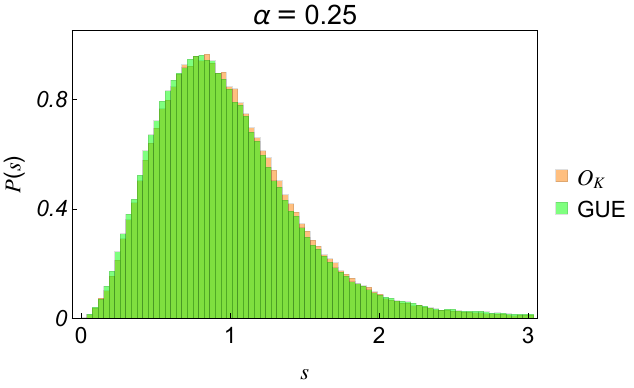}
        \caption{}
    \end{subfigure}
    \hfill
    \begin{subfigure}[b]{0.3\textwidth}
        \includegraphics[width=1.1\textwidth]{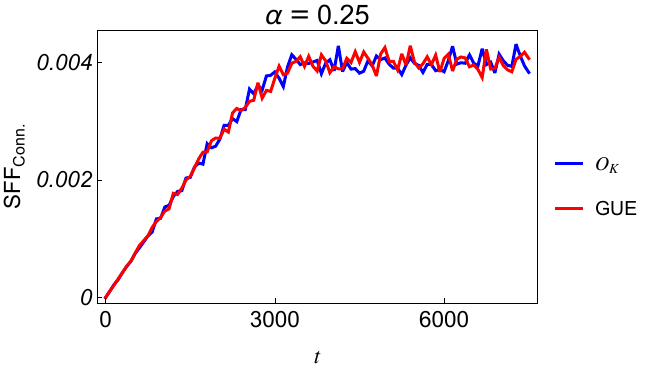}
        \caption{}
    \end{subfigure}
    
    \medskip
    
    \begin{subfigure}[b]{0.3\textwidth}
        \includegraphics[width=0.85\textwidth]{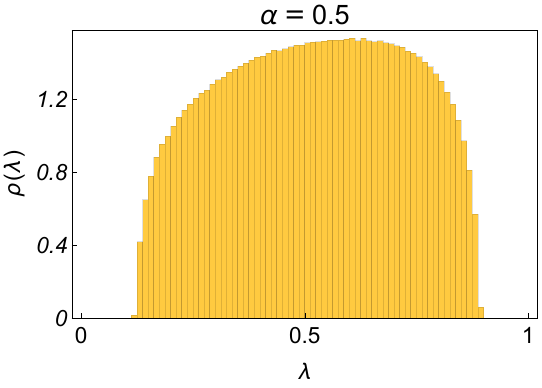}
        \caption{}
    \end{subfigure}
    \hfill
    \begin{subfigure}[b]{0.3\textwidth}
        \includegraphics[width=\textwidth]{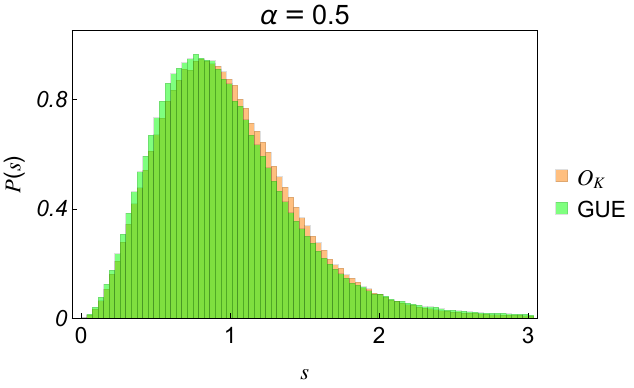}
        \caption{}
    \end{subfigure}
    \hfill
    \begin{subfigure}[b]{0.3\textwidth}
        \includegraphics[width=1.1\textwidth]{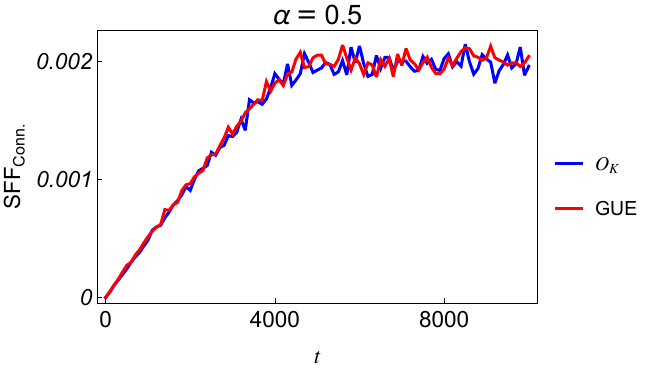}
        \caption{}
    \end{subfigure}
    
    \caption{Different statistics in the case of a non-degenerate spectrum, for various values of $\alpha\leq\frac{1}{2}$. Left column: the normalized level density (yellow). Central column: the normalized spacing distribution of $O_K$ (orange), and of the GUE (green). Right column: the connected SFF for $O_K$ (blue), and for the GUE (red). There is no eigenvalue recovery in this case.}
    \label{fig:nondeg}
\end{figure}

\begin{figure}[htbp]
    \centering
    
    \begin{subfigure}[b]{0.3\textwidth}
        \includegraphics[width=0.85\textwidth]{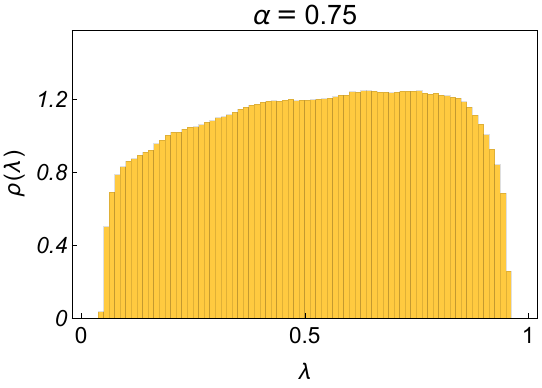}
        \caption{}
    \end{subfigure}
    \hfill
    \begin{subfigure}[b]{0.3\textwidth}
        \includegraphics[width=\textwidth]{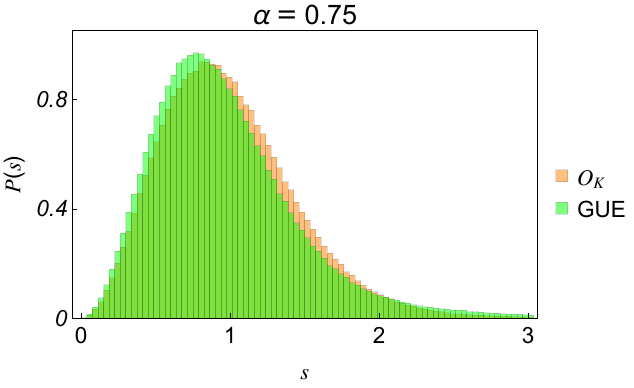}
        \caption{}
    \end{subfigure}
    \hfill
    \begin{subfigure}[b]{0.3\textwidth}
        \includegraphics[width=1.1\textwidth]{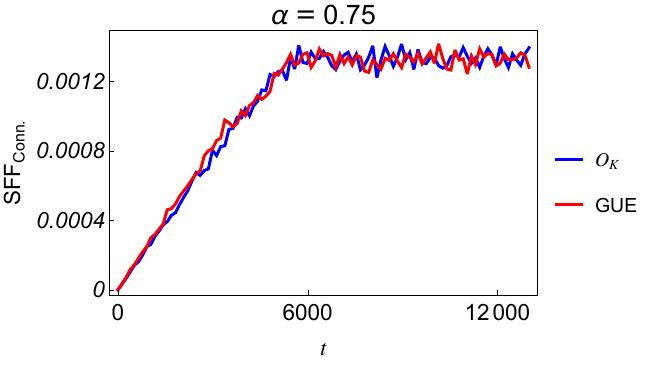}
        \caption{}
    \end{subfigure}
    
    \medskip
    
    \begin{subfigure}[b]{0.3\textwidth}
        \includegraphics[width=0.85\textwidth]{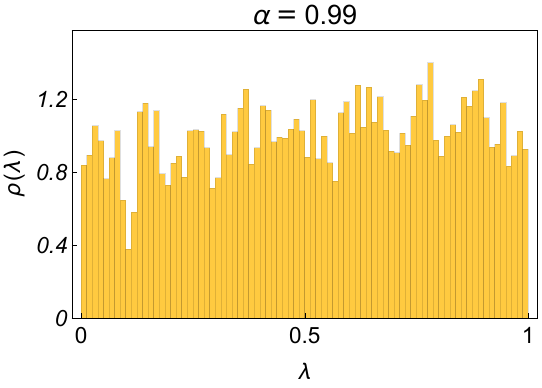}
        \caption{}
    \end{subfigure}
    \hfill
    \begin{subfigure}[b]{0.3\textwidth}
        \includegraphics[width=\textwidth]{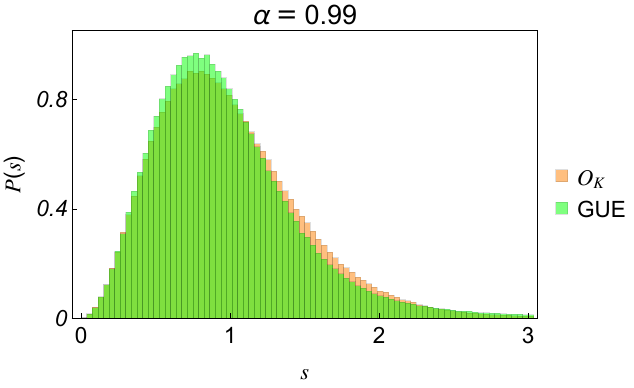}
        \caption{}
    \end{subfigure}
    \hfill
    \begin{subfigure}[b]{0.3\textwidth}
        \includegraphics[width=1.1\textwidth]{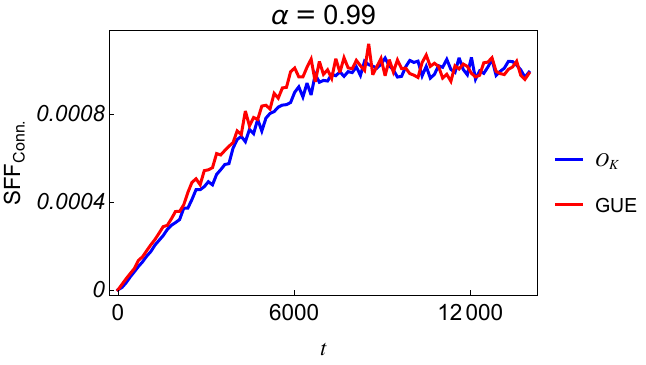}
        \caption{}
    \end{subfigure}
    
    \medskip
    
    \begin{subfigure}[b]{0.3\textwidth}
        \includegraphics[width=0.85\textwidth]{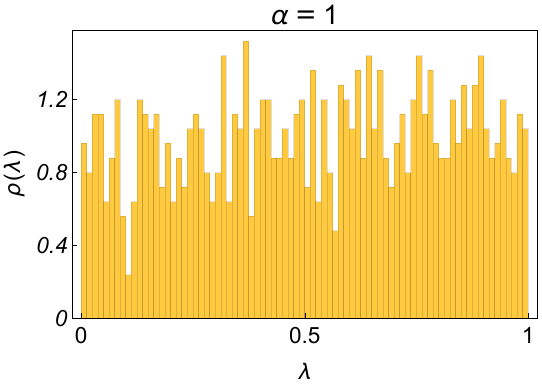}
        \caption{}
    \end{subfigure}
    \hfill
    \begin{subfigure}[b]{0.3\textwidth}
        \includegraphics[width=\textwidth]{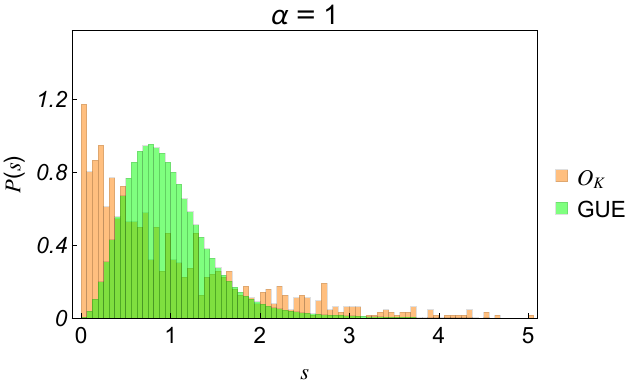}
        \caption{}
    \end{subfigure}
    \hfill
    \begin{subfigure}[b]{0.3\textwidth}
        \includegraphics[width=1.1\textwidth]{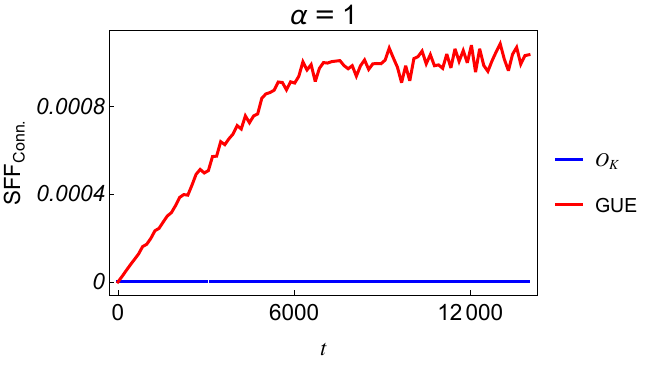}
        \caption{}
    \end{subfigure}
    \caption{Different statistics in the case of a non-degenerate spectrum, for various values of $\alpha>\frac{1}{2}$. Left column: the normalized level density (yellow). Central column: the normalized spacing distribution of $O_K$ (orange), and of the GUE (green). Right column: the connected SFF for $O_K$ (blue), and for the GUE (red). There is no eigenvalue recovery in this case.}
    \label{fig:nondeg2}
\end{figure}
As in the previous cases, we computed the normalized level density, the normalized spacing distribution, and the connected spectral form factors for different values of $\alpha\in(0,1)$, shown in Figures \ref{fig:nondeg} and \ref{fig:nondeg2}. As expected, we see no eigenvalue recovery in this case. The level density (left column) resembles a Wigner semicircle only for $\alpha\ll 1$, which is universal for any $O$ with or without degeneracy.  As for level spacing and SFFs, they behave very close to the GUE even for values of $\alpha$ close to one, which indicates the persistence of random matrix universality even when only a small portion of eigenvalues is removed. In this case, the reference GUE $\mathbf{H}$ satisfies
\begin{align}
    \E[\mathbf{H}_{ij}]=\mu\delta_{ij},\qquad \E[|\mathbf{H}_{ij}|^2]=\frac{1}{N}
\end{align}
where $\mu=\frac{1}{N}\Tr O$ as in the previous cases.

In all the numerical examples we saw, when $\alpha\to 0$ the statistics of $O_K$ are well described by the GUE, as explained by our analytic results of the previous section. The behavior is universal, meaning that it is true for any operator $O$, which we see clearly in the two extreme scenarios studied here.

\section{Free Probability Theory}
\label{sec:freeprob}

In the previous section, we have successfully shown that in the limit $\alpha \ll 1$, the probability distribution of a projected operator follows GUE statistics. In this section, we will provide another perspective on understanding the emergence of the GUE density of states from Free Probability Theory.

Because Free Probability Theory has only made rare appearances in high-energy physics, we will first give a short review of the topic. We will then use powerful theorems to show that the first two moments of the projected operator agree with those of the GUE, up to the second order in $\alpha$. To do so, we will use both the so-called $S$-transform and \textit{second-order free probability theory}. Moreover, beyond the small $\alpha$ limit, we will also obtain the density of states for all finite values of $\alpha$ through \textit{free compression} using the so-called $R$-transform.

\subsection{Practical review of free probability theory}\label{sec: review of FPT}
Before introducing Free Probability Theory governing non-commutative variables, let us first review the basics of a classical probability distribution, as they will be of importance later.

\subsubsection*{Moments in classical probability theory}

   Given a classical probability distribution $f$ for a real continuous random variable $X$, its $k$-th moment is defined as
\begin{align}
    m_k=\text{E}[X^k]=\int_{-\infty}^{+\infty}x^kf(x)\,dx \,,
\end{align} 
where $x$ stands for the variable of the density function of $X$.
The collection of all moments essentially shapes (defines) the probability distribution. In other words, there is a unique distribution whose moments are given as $\{m_n\}_{n\geq1}$ if Carleman's condition \cite{Durrett_2019}: $\sum_{n=1}^\infty m_{2n}^{-\frac{1}{2n}}=+\infty$ is satisfied.

An example of such a probability distribution is the Wigner semicircle with variance $\sigma^2$:
\begin{align}
    \rho(\lambda)=\frac{1}{2\pi \sigma^2}\sqrt{4\sigma^2-\lambda^2} \,.
\end{align}
Its moments are
\begin{align}
    m^{(W)}_k=
\begin{cases}
  (\sigma^2)^{k/2}C_{k/2} & \text{for } k\text{ even} \\
  0 & \text{for } k\text{ odd} \,,
\end{cases}
\end{align}
where $C_n=\frac{1}{n+1}\binom{2n}{n}$ is the $n$-th Catalan number \cite{MinSpei}. These moments indeed satisfy Carleman's condition. This means that whenever a random variable has moments of the Wigner semicircle, its probability distribution must be the Wigner semicircle. 

\subsubsection*{Moments in Free Probability Theory}

Free Probability Theory deals with non-commutative random variables. Defining a probability distribution for such variables in a way similar to classical probability theory is not straightforward. One way is to describe probability not with distributions, but with moments. 

A simple realization of non-commuting variables is to consider random matrices. For $N\times N$ random matrices, the linear functional could be defined as
\begin{align}
    \varphi(X):=\text{E}[\text{tr}X] \,,
\end{align}
where we used the normalized trace $\text{tr}=\frac{1}{N}\text{Tr}$. With this definition, the $k$-th moment of the random matrix $X$ is
\begin{align}
    m_k=\text{E}\left[\text{tr}\left(X^k\right)\right] \,.
\end{align}
Using the trace in the definition of the functional is also very useful because it allows us to know the moments of the one-point distribution of the eigenvalues:
\begin{align}
   m_k= \text{E}\left[\text{tr}\left( X^k \right) \right]=\text{E}\left[\frac{1}{N}\sum_{i=1}^N (\lambda_i)^k \right]=\text{E}\left[ \lambda^k \right] \,,
\end{align}
where we used the linearity of the expectation value operator.

For example, we can consider an $N\times N$ GUE random matrix $X$ with 
\begin{align}\label{GUEcond}
    \text{E}[X_{ij}]=0 \,,\qquad \text{E}[|X_{ij}|^2]=\frac{1}{N} \,,
\end{align}
whose moments can be computed as
\begin{align}\label{finalmomGUE}
    \lim_{N\rightarrow\infty}\E\left[\tr\left(X^k\right)\right]=\frac{1}{2\pi}\int_{-2}^2t^k\sqrt{4-t^2}\,dt=
    \begin{cases}
        C_{k/2} & k\text{ even}\\
        0 & k\text{ odd}
    \end{cases}
\end{align}
which means that the moments of $X_N$ are, in the limit of large $N$, equal to the moments of the Wigner semicircle. 

\subsubsection{Freeness and asymptotic freeness}

Consider a non-commutative probability space $(\mathcal{A},\varphi)$, where $\mathcal{A}$ is a unital and associative algebra defined over $\mathbb{C}$ equipped with a unital linear functional $\varphi:\mathcal{A}\rightarrow \mathbb{C}$.
Suppose $\mathcal{A}_1,\ldots,\mathcal{A}_s$ are unital subalgebras of $\mathcal{A}$. We say that $\mathcal{A}_1,\ldots,\mathcal{A}_s$ are \emph{freely independent} or \emph{free} with respect to $\varphi$ if, 
\begin{align}
        \varphi(a_1\ldots a_r)=0\,,
    \end{align}
as long as the neighboring elements are not from the same subalgebra and the following conditions are satisfied for $r\geq2$ and $a_1,\ldots,a_r\in\mathcal{A}$ such that $\varphi(a_i)=0$ for $i=1\ldots r$, and $a_i\in\mathcal{A}_{j_i}$ with $j_i\in[s]\equiv\{1,2,\dots,s\}$ and $j_1\neq j_2,\ldots,j_{r-1}\neq j_r$. Elements $a_1,\ldots,a_s\in\mathcal{A}$ are said to be \emph{free} or \emph{freely independent} if the generated unital subalgebras $\mathcal{A}_i=\text{alg}(\mathbb{I},a_i)$ for $i=1,\ldots,s$ are free in $\mathcal{A}$ with respect to $\varphi$.

To see that freeness generalizes independent random variables, consider the freeness for two elements $a_1,a_2$. This time, we \textit{do not} demand that they are centered, namely $\varphi(a_{1,2})\neq 0$. We can consider
\begin{align}
    \varphi\left[ \left(a_1-\varphi(a_1)\mathbb{I} \right)\left(a_2-\varphi(a_2)\mathbb{I} \right) \right]=0\,,
\end{align}
which leads to the factorization property:
\begin{align}\label{phiAB}
    \varphi\left(a_1a_2 \right)=\varphi\left[\varphi(a_2)a_1+\varphi(a_1)a_2-\varphi(a_1)\varphi(a_2)\mathbb{I} \right]=\varphi\left(a_1 \right)\varphi\left(a_2 \right)\,.
\end{align}
This is the sense in which freeness generalizes independence for classical random variables, and extends it to non-commuting variables: the probability distribution of independent variables factorizes. 

But things become more complicated in higher-point expectation values. For example, if we consider $\{a_1,a_2\}$ free from $\{b_1,b_2\}$, a similar calculation shows that
\begin{equation}\label{phiABAB}
\varphi(a_1b_1a_2b_2)=\varphi(a_1a_2)\varphi(b_1)\varphi(b_2)+\varphi(a_1)\varphi(a_2)\varphi(b_1b_2)-\varphi(a_1)\varphi(a_2)\varphi(b_1)\varphi(b_2)\,,
\end{equation}
which is manifestly different from the classical commuting case.

Now, we can consider asymptotic freeness. Let us build an $N\times N$ random matrix
    \begin{align}
        Y_N=\left( X_{i_1}^{m_1}-c_{m_1}\mathbb{I} \right)\left( X_{i_2}^{m_2}-c_{m_2}\mathbb{I} \right)\ldots\left( X_{i_r}^{m_r}-c_{m_r}\mathbb{I} \right)\,,
    \end{align}
    where $c_{m_n}=\lim_{N\rightarrow \infty} E[\Tr(X_n^{m_n})]$ with $m_n$ being a positive integer.
$X_1,\ldots,X_s$ are called \emph{asymptotically free} if
\begin{align}
        \lim_{N\rightarrow\infty}E\left[\tr\left(Y_N \right) \right]=0\,.
    \end{align}

Intuitively, asymptotic freeness means that random variables become free, but only in the limit $N\to \infty$. This will be of importance to us, because of the application of Free Probability Theory to matrices (like operators in quantum mechanics). In this case, non-trivial freeness can be reached only asymptotically.

\subsubsection{Sum of free random variables}

We now discuss the sum of free random variables, which will be important for us later. Let us thus introduce the free additive convolution as a first implication of freeness. For two self-adjoint random variables $a,b$ that are free from each other. The distribution of the sum $a+b$ is given by the free additive convolution of the individual distributions:
\begin{equation}
    \rho_{a+b}=\rho_a \boxplus \rho_b\,.
\end{equation}
The machinery for the free additive convolution is to first consider the Cauchy transforms:
\begin{equation}
    G(z)=\int d\lambda \frac{\rho(\lambda)}{z-\lambda}\,,
    \label{eq-CauchyTransform}
\end{equation}
whose inverse is related to the $R$-transform:
\begin{equation}
    G^{-1}(z)=R(z)+\frac{1}{z}\,.
\end{equation}
The $R$-transform can also be expanded using free cumulants as
\begin{equation}
    R(z)=\sum_{n=1}^{\infty} \kappa_n z^{n-1}\,.
\end{equation}
As a main condition for freeness, the mixed free cumulants vanish, which implies:
\begin{equation}
    \kappa_n(a+b)=\kappa_n(a)+\kappa_n(b)\,,
\end{equation}
indicating the factorization of the $R$-transform:
\begin{equation}
    R_{\rho_{a+b}}(z)=R_{\rho_a\boxplus \rho_b}(z)=R_{\rho_a}(z)+R_{\rho_b}(z)\,.
\end{equation}
This provides a powerful tool for finding the free convoluted distribution by using the Stieltjes inversion formula:
\begin{equation}
    \rho_{a+b}=-\frac{1}{\pi} \lim_{\epsilon \rightarrow 0} \mathrm{Im}(G_{\rho_a\boxplus \rho_b}(\lambda+i\epsilon))\,.
    \label{eq-StieltjiesInv}
\end{equation}

\subsubsection{Product of free random variables}

Lastly, we review the machinery for computing the distribution of the product of two random variables. 

Consider two non-commutative random variables $a$ and $b$ which are free. If the probability distribution of $a$ is $\mu_a$ and the distribution of $b$ is $\mu_b$, then the probability distribution of the product $ab$ is given by an operation called \textit{free multiplicative convolution} \cite{Voi2}:
\begin{align}
    \mu_{ab}=\mu_a \boxtimes \mu_b\,.
\end{align}
However, even though $a,b$ are selfadjoint, the product $ab$ is not selfadjoint in a non-commutative probability space, which thus cannot lead to a guaranteed probability measure on $\mathbb{R}$. Instead, $\sqrt{a}b\sqrt{a}$ (or $\sqrt{b}a\sqrt{b}$) shares the same moments as $ab$ while maintaining the selfadjointness. We can therefore identify the distributions:
\begin{equation}
\mu_{ab}=\mu_{\sqrt{a}b\sqrt{a}}=\mu_{\sqrt{b}a\sqrt{b}}=\mu_{ba}\,,
\end{equation}
which implies that  $\boxtimes$ is commutative \cite{Voi}:
\begin{equation}
    \mu_a \boxtimes \mu_b=\mu_b \boxtimes \mu_a\,.
\end{equation}

To compute the free multiplicative convolution, it is useful to define the moment function $M_\mu$:
\begin{equation}
    M_\mu(z):=z G_\mu(z)-1=\sum_{n=1}^\infty m_n(\mu)z^n\,,
    \label{eq-FreeMulMomentFunc}
\end{equation}
where $G_\mu(z)$ is the Cauchy transform of $\mu$ defined in (\ref{eq-CauchyTransform}), and $m_n(\mu)$ is the $n$-th moment of $\mu$. The $S$-transform is defined through the inverse of the moment function:
\begin{equation}
    S_\mu(z):=\frac{1+z}{z}M^{-1}_\mu(z)\,,
    \label{eq-Stranform}
\end{equation}
which is also the inverse of the $R$-transform by \cite{Nica_Speicher_2006} 
\begin{equation}
    S_\mu(z)=\frac{1}{z}R^{-1}_{\mu}(z)\,.
\end{equation}
The $S$-transform provides a powerful tool for computing the distribution of the product of two random variables
\begin{equation}
    S_{\mu_a\boxtimes\mu_b}=S_{\mu_a}\cdot S_{\mu_b}\,.
\end{equation}
Given this, one can find $M_{\mu_{ab}}$ and $G_{\mu_{ab}}$ using (\ref{eq-FreeMulMomentFunc}) and (\ref{eq-Stranform}), which hence leads to the distribution $\mu_{ab}$ from the Stieltjes inversion formula (\ref{eq-StieltjiesInv}).

\subsection{The one-point function of projected operators}\label{sec: DOS from FPT}

We will now use theorems of free probability theory to show that the probability distribution agrees with that of the GUE in the large $N$, small $\alpha$ limit. In this section, we first apply the formalism of the S-transform to study the one-point function of the spectral density for the projected operator $O_K$. First, notice that the projected operator

\begin{align}
    O_K=P_KU O U^\dagger P_K
\end{align}
shares the same eigenvalues as
\begin{align}
    O_K'=P_K(U O U^\dagger)\,.
\end{align}

In order to use the S-transform, we need to prove that $P_K$ and $U O U^\dagger$ are free. This is proved in the following theorem:
\begin{theorem}
    Let $A$ and $B$ be two deterministic $N\times N$ matrices with $A\xrightarrow{\text{distr}}a$ and $B\xrightarrow{\text{distr}}b$ as $N\to\infty$. Let $U$ be an $N\times N$ Haar unitary random matrix. Then $A,\,U B U^\dagger \xrightarrow{\text{distr}}a,b$ as $N\to\infty$, where $a$ and $b$ are free. The convergence ``$\xrightarrow{\text{distr}}$" refers to a convergence in distribution, hence a convergence in moments \cite{MinSpei}.
\end{theorem}
This can be straightforwardly proved using the Weingarten calculus to compute the mixed cumulants from moments of $P_K$ and $U^\dagger OU$, which yields the non-crossing partitions as a characteristic of freeness in the large-$N$ limit.
This result implies that, asymptotically, $P_K$ and $U^\dagger OU$ are free, hence we can use the S-transform formalism to compute the eigenvalue distribution of $O_K$.

\subsubsection{Universality of the one-point distribution at small $\alpha$}
Knowing that $P_K$ and $U O U^\dagger$ are free from each other, this allows us to use the free multiplicative convolution to derive the expression of the one-point eigenvalue distribution of $O'_K$, which agrees with that of $O_K$. A simple case with two eigenvalues can be straightforwardly computed using the $S$-transform. We show the limiting density as the Wigner semicircle for small $\alpha$ in Appendix \ref{App-Proj2eig}. Here, we consider a more general $N\times N$ diagonal matrix $O$ with eigenvalues distributed according to a generic well-behaved p.d.f. $f(\lambda)$, in the limit $N\to\infty$. Suppose also that these eigenvalues are independent and that we have no degeneracy. This means that
\begin{align}
    U O U^\dagger\sim f\,,
\end{align}
where $\sim$ indicates its eigenvalue distribution.
The projector $P_K$ has the usual asymptotic distribution
\begin{align}
    P_K\sim\mu=(1-\alpha)\delta_0+\alpha\delta_1\,.
\end{align}
As we saw in the previous sections,
\begin{align}
    O_K=P_KU O U^\dagger P_K\sim\rho=\mu\boxtimes f\,.
\end{align}
Then, the following theorem holds:
\begin{theorem}
\label{Theorem-1pt from S trans}
    In the limit $N,K\to\infty$ and $\alpha=\frac{K}{N}\to0$, the probability distribution $\rho$ of eigenvalues of $O_K$ converges to a Wigner semicircle of variance $\alpha\sigma_f^2$ centered at $f_1$, superimposed to a trivial deterministic spectrum centered in zero:
    \begin{equation}
    \begin{aligned}
        \rho(x)=(\mu\boxtimes f)(x) \xrightarrow{\text{distr}}\rho_{\text{lim}}(x)
    \end{aligned}
    \end{equation}
    with
    \begin{align}
        \rho_\text{lim}(x)=\alpha\,\frac{1}{2\pi\alpha\sigma_f^2}\sqrt{4\alpha\sigma_f^2-(x-f_1)^2}\;\chi_{\left[-\Delta+f_1,\Delta+f_1\right]} +(1-\alpha)\delta(x)
    \end{align}
    where $f_1=f[O]$ and $\sigma_f^2=f[O^2]-f[O]^2$ are respectively the mean and the variance of the eigenvalues of $O$, $\chi[a,b]$ is the characteristic (indicator) function of the interval $[a,b]$, and $\Delta=\sqrt{4\alpha\sigma_f^2}$.

    More precisely, the moments of $\rho$ agree with the moments of $\rho_\text{lim}$ up to $O(\alpha^2)$.
\end{theorem}
To prove it, in the spirit of Carleman’s condition, we compare the moments of $\rho$ and $\rho_\text{lim}$ up to order $\alpha^2$. The moments of $\rho$ are essentially encoded in the moments of $f$. They can be compared using the moment-generating function (\ref{eq-FreeMulMomentFunc}) and the $S$-transform (\ref{eq-Stranform}). With many detailed calculations shown in the Appendix \ref{App-TheoremProof}, one finds that the $n$-th order moments for $\rho$ up to order $\alpha^2$ are given by
\begin{align}\label{final1stord}
    \begin{dcases}
        \rho_1=\alpha f_1\\
        \rho_n=\alpha(f_1)^n+\alpha^2\frac{n(n-1)}{2}\sigma_f^2\,(f_1)^{n-2}+O(\alpha^3),\qquad n\geq2\,.
    \end{dcases}
\end{align}

In order to prove the theorem, we have to compare this result with the moments of
\begin{align}
    \rho_\text{lim}(x)=\alpha\mathscr{W}_{\alpha\sigma_f^2}(x-f_1)+(1-\alpha)\delta(x)\,.
\end{align}
The previous expression contains a shifted Wigner semicircle of variance $\alpha\sigma_f^2$:
\begin{align}
    \mathscr{W}_{\alpha\sigma_f^2}(x)=\frac{1}{2\pi\alpha\sigma_f^2}\sqrt{4\alpha\sigma_f^2-x^2}\chi_{\left[-\sqrt{4\alpha\sigma_f^2},\sqrt{4\alpha\sigma_f^2}\right]}\,;
\end{align}
we can see, by a trivial generalization of equation (\ref{finalmomGUE}), that its moments are
\begin{align}
    m_\mathscr{W}^{(n)}=
    \begin{cases}
        0,\qquad &n\text{ odd}\\
        (\alpha\sigma_f^2)^{\frac{n}{2}} C_{\frac{n}{2}},\qquad &n\text{ even}\,.
    \end{cases}
\end{align}
Now we can compute the moments of $\rho_\text{lim}$:
\begin{equation}
    \begin{aligned}
        \rho_\text{lim}[X^n]=\int dx\,x^n\rho_\text{lim}(x)=\alpha\int dx\,(x+f_1)^n\mathscr{W}_{\alpha\sigma^2}(x)\,,
    \end{aligned}
\end{equation}
which can easily be shown to equal
\begin{align}
    \begin{dcases}
        \rho_\text{lim}[X]=\alpha f_1\\
        \rho_\text{lim}[X^n]=\alpha(f_1)^n+\alpha^2\frac{n(n-1)}{2}\sigma_f^2\,(f_1)^{n-2}+O(\alpha^3),\qquad n\geq2\,.
    \end{dcases}
\end{align}
This is exactly the same result obtained in (\ref{final1stord}), up to order $\alpha^2$.

The terms of order $\alpha^3$, on the contrary, differ. To see this, we do not need to compute the term of order $\alpha^3$ for the general $n$, it is enough to find a counterexample for a specific value of $n$. For example, in the case $n=3$, such a term is equal to $f_3-3\sigma_f^2f_1-f_1^3$ for the projected operator, while it is zero for the GUE.

We now turn to the study of the correlation between eigenvalues.

\subsection{Second-order free probability theory}
The formalism of (first-order) free probability theory that we saw so far is based on the definition of a non-commutative probability space $(\mathcal{A},\varphi)$. The linear functional $\varphi:\mathcal{A}\to\mathbb{C}$ is used to define the moments, which in the case of random matrices is $\varphi(x)=E[\tr X]$. 

However, one of the main features of RMT is the fact that the eigenvalues are correlated, supported by the appearance of eigenvalue repulsion. We thus expect the eigenvalues of the projected operator to obey RMT statistics, and in particular to repel. This has already been proven in Section \ref{sec: jpdf from RMT} explicitly from the joint eigenvalue distributions using random matrix techniques. However, we would now like to see what free probability can say about these correlations.

Importantly, these correlations cannot be seen from first-order free probability theory, which has no access to the two-point correlation function of the eigenvalues. For this reason, \textit{Second-order Free Probability Theory} \cite{MinSpei} was constructed. It begins with an addition to the tracial non-commutative probability space $(\mathcal{A},\varphi)$:
\begin{equation}
    \mathrm{a~~ bilinear~~ functional:~~}\varphi_2:\mathcal{A}\times\mathcal{A}\to\mathbb{C}\,,
\end{equation}
which is tracial and symmetric in its two variables $\varphi_2(a,b)=\varphi_2(b,a)$ for all $a,b\in\mathcal{A}$ and vanishes in case of either $a$ or $b$ being an identity. These define a \textit{second-order non-commutative probability space} $(\mathcal{A},\varphi,\varphi_2)$.

The bilinear functional $\varphi_2$ can then be used to define the \textit{second-order moments} of an $N \times N$ random matrix $X$, i.e., the moments of the two-point distribution of its eigenvalues:
\begin{align}\label{2ndordmom}
    m_{p,q}=\varphi_2(X^p,X^q)=k_2\left(\Tr (X^p),\Tr (X^q) \right)
\end{align}
where $k_2$ is the second classical cumulant, hence the covariance:
\begin{align}
    k_2(x,y)=\text{cov}(x,y)=\E[xy]-\E[x]\E[y]\,.
\end{align}
The \textit{asymptotic freeness} of the second-order is defined in a similar manner as the first-order through the factorization of second-order cumulants:
\begin{equation}
    \lim_{N\rightarrow \infty}\kappa_{p,q}^{X+Y}= \kappa_{p,q}^{X}+\kappa_{p,q}^{Y}\,,
\end{equation}
where $\kappa_{p,q}$ is the connected cumulants appearing in $m_{p,q}=\varphi_2(X^p,X^q)=\kappa_2(\Tr(X^p),\Tr(X^q))$ through moment-cumulant relation. This implies that all mixed free second-order cumulants vanish.

\subsubsection{Non-crossing partitions for GUE}

Notice that the connected SFF, one of the main tools for diagnosing quantum chaos, depends only on the second-order moments (\ref{2ndordmom}). This can be seen by expanding the exponentials in the definition of the SFF:
\begin{equation}
    \begin{aligned}
        \text{SFF}(t)&=\frac{1}{N^2}\left(\mathbb{E}\left[\Tr\left(e^{iXt} \right)\Tr\left(e^{-iXt} \right)\right]-\mathbb{E}\left[\Tr\left(e^{iXt} \right)\right]\mathbb{E}\left[\Tr\left(e^{-iXt} \right)\right] \right)\\
        &=\sum_{l,p=0}^\infty \frac{(-1)^p(it)^{l+p}}{l!p!N^2}\varphi_2\left(X^l,X^k \right)\,.
    \end{aligned}
\end{equation}
We will first study this expansion for GUE random matrices, as some of the technology based on second-order freeness can be generalized for the case of our projected operator.

Using definition (\ref{2ndordmom}), the second-order moments of the GUE are obtained to be related to combinatorial properties as in the first-order case  \cite{MinSpei}. Unlike the first-order case, where we consider the non-crossing partitions of the cumulants with points lying on a circle, the combinatorics for the second-order case include $p$ points on an inner and $q$ points on the outer circle forming an annulus called $(p,q)$-annulus. This is the annulus with the integers 1 to $p$ are arranged clockwise on the outer circle and $p+1$ to $p+q$ arranged counterclockwise on the inner circle.

Different ways of contracting the points form sets of partitions of the permutation $S_{p+q}$. If there is an element of the permutation $\pi\in S_{p+q}$ with cycles that do not cross each other, and at least one of them connects the two circles, and the region enclosed in each cycle is homeomorphic to the disc with boundary oriented clockwise, it is called a \textit{non-crossing permutation} on the $(p,q)$-annulus (or just a \textit{non-crossing annular permutation}). We denote by $S_{NC}(p,q)$ the set of non-crossing permutations on the $(p,q)$-annulus. The subset of variables consisting of non-crossing pairings on the $(p,q)$-annulus is denoted by $NC_2(p,q)$.

\begin{figure}[h!]
    \centering
    \includegraphics[width=0.3\linewidth]{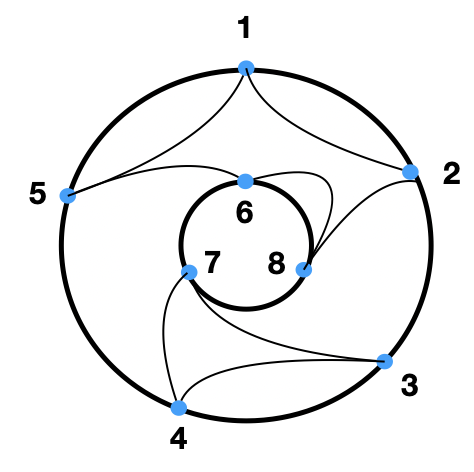}
    \caption{Graphic representation of the non-crossing permutation $(1,2,8,6,5)(3,4,7)$ on a $(5,3)$-annulus as a partition.}
    \label{fig:annulus}
\end{figure}
In Figure \ref{fig:annulus}, we show an example diagram of a non-crossing permutation on the annulus as a partition \footnote{Notice that there is an obvious ordering ambiguity in the connection between partitions and permutations on the annulus (see e.g. \cite{MinSpei}). However, this does not affect non-crossing pairings $NC_2(p,q)$, which can be regarded as partitions or as permutations without ambiguity, since the ordering of the elements within 2-cycles is not important.}. The non-crossing pairs dominate in the second-order cumulant, i.e., in the limit $N\to\infty$, the second-order moment $m_{p,q}$ of an $N\times N$ GUE random matrix $X$ is given by the number of non-crossing pairings on a $(p,q)$-annulus \cite{MinSpei}:
    \begin{align}
        m_{p,q}=|NC_2(p,q)|\,.
    \end{align}
This discussion can be generalized to define $n$-th order moments, building an $n$-th order free probability theory \cite{MinSpeiCol}. Given the $N\times N$ random matrices $B_1,\ldots,B_n$, we define the $n$-th order moments as
    \begin{align}
        \varphi_n(B_1,\ldots,B_n)=k_n\left(\Tr(B_1),\ldots,\Tr(B_n)\right)\,.
    \end{align}
In the previous definition, $k_n$ is the $n$-th classical cumulant, which is defined as
\begin{align}\label{classiccum}
    k_n(x_1,\ldots,x_n)=\sum_{\V\in\PP(n)}(\#\V-1)!(-1)^{\#\V-1}\prod_{V\in\V}\E\left[\prod_{i\in V}x_i \right]\,,
\end{align}
where $x_1,\ldots,x_n$ are classical random variables, $\PP(n)$ is the set of partitions of $n$ elements and $\#\V$ is the number of blocks in the partition $\V$. Notice that in this definition we are using the traditional trace instead of the normalized one.

Our final goal is to study the second-order moments of the projected operator $O_K$, and see whether they agree with the ones of the GUE for small $\alpha$. 
Given the increasing computational complexity of high-order moments, we will therefore eventually check a certain number of moments of $O_K$ and compare with the ones from the GUE. 

To initiate the computation, we notice that for two random and independent unitary invariant ensembles $\mathcal{A}$ and $\mathcal{B}$, the moments for $N\times N$ random matrices $A_1,\ldots,A_n\in\mathcal{A}$ and $B_1,\ldots,B_n\in\mathcal{B}$ follow \cite{MinSpeiCol}:
\begin{equation}\label{prodmom}
            \varphi(\U,\gamma)[A_1B_1,\ldots,A_nB_n]=\sum_{(\V,\pi),(\W,\sigma)}\kappa(\V,\pi)[A_1,\ldots,A_n]\cdot\varphi(\W,\sigma)[B_1,\ldots,B_n]
    \end{equation}
   where the sum is all over $(\V,\pi),(\W,\sigma)\in\PS(n)$ such that $\V\vee\W=\U$ and $\pi\sigma=\gamma$.
Notice that (\ref{prodmom}) is valid for any value of the Hilbert space dimension $N$, not just for the large $N$ case.

The cumulants in (\ref{prodmom}) can be written as
    \begin{align}
        \kappa(\V,\pi)[A_1,\ldots,A_n]=\sum_{\substack{(\M,\tau)\in\PS(n)\\ \M\leq\V}}\varphi(\M,\tau)[A_1,\ldots,A_n]\cdot C_{\pi\vee\M,\V}(\tau\pi^{-1})
    \end{align}
    where
    \begin{align}\label{Cfactor}
        C_{\pi\vee\M,\V}(\tau\pi^{-1})=\sum_{\N\in\PP(n)}(-1)^{\#\N-\#\V}\prod_{V\in\V}\left(\#(\N|_V)-1 \right)!\cdot\text{Wg}(\N,\tau\pi^{-1})\,,
    \end{align}
and 
    \begin{align}
        \varphi(\V,\pi)[B_1,\ldots,B_n]\equiv\prod_{V\in\V}\varphi(\pi|_V)[B_1,\ldots,B_n|_V]\equiv\prod_{V\in\V}\varphi_r\left(B|_{c_1},\ldots,B|_{c_r} \right)\,,
        \label{mompartperm}
    \end{align}
    where $(\V,\pi)\in\PS(n)$ is a given partitioned permutation, $\pi|_V$ is the permutation obtained by joining the cycles of $\pi$ which are contained in the block $V$ of $\V$, and $B_1,\ldots,B_n|_V=B_{i_1},\ldots,B_{i_k}$ with $V=\{i_1,\ldots,i_k\}$, $c_1,\ldots,c_r$ are the cycles of $\pi|_V$ and $B|_{c_i}$ is the restriction of $B$ to the cycle $c_i$. For example,
    \begin{align}
    \varphi\left(\{1,3,4\}\{2\},(1,3)(2)(4) \right)[B_1,B_2,&B_3,B_4]=\varphi_2(B_1B_3,B_4)\varphi_1(B_2)\nonumber\\
    &=k_2(\Tr(B_1B_3),\Tr(B_4))\cdot k_1(\Tr(B_2))\,.
\end{align}

\subsubsection{Second-order moments of randomly projected operators}
We can finally apply all the machinery above to compute the second-order moments of $O_K$. We are interested in the quantity
\begin{align}
    \varphi_2\left((O_K)^p,(O_K)^q \right)=\varphi_2\left((U O U^\dagger\cdot P_K)^p,(U O U^\dagger\cdot P_K)^q \right)\,.
\end{align}
Since $U O U^\dagger$ is unitarily invariant by construction, we can use (\ref{prodmom}). In order to adapt it to our case, we have to choose:
\begin{itemize}
    \item $n=p+q$
    \item $\U=\{1,2,\ldots,p+q\}$, in order to compute a single moment and not a product of moments;
    \item $\gamma$ such that $\#\gamma=2$ so that we compute a second-order moment; for example, we can take $\gamma=(1,\ldots,p)(p+1,\ldots,p+q)$;
    \item $A_1=A_2=\ldots =A_{p+q}=U O U^\dagger$;
    \item $B_1=B_2=\ldots=B_{p+q}=P_K$
\end{itemize}
In this way, once we call $\tilde{O}\equiv U O U^\dagger$, (\ref{prodmom}) becomes
\begin{align}\label{momOK}
    \varphi_2\left((O_K)^p,(O_K)^q \right)=\sum_{(\V,\pi),(\W,\sigma)\in\PS(p+q)}\kappa(\V,\pi)[\underbrace{\tilde{O},\ldots,\tilde{O}}_{p+q\text{ times}}]\cdot\varphi(\W,\sigma)[\underbrace{P_K,\ldots,P_K}_{p+q\text{ times}}]\,.
\end{align}
where the sum is such that $\V\vee\W=\U$ and $\pi\sigma=\gamma$.

Let us begin with the computation of the $\varphi$-factor in the sum. Using definition \ref{mompartperm}, we can find
\begin{equation}
    \begin{aligned}
        \varphi(\W,\sigma)[\underbrace{P_K,\ldots,P_K}_{p+q\text{ times}}]&=\prod_{W\in\W}k_{r(W)}(\Tr[(P_K)^{l_1^W}],\ldots,\Tr[(P_K)^{l_{r(W)}^W}])=\\
        &=\prod_{W\in\W}k_{r(W)}(\underbrace{\alpha N,\ldots,\alpha N}_{r(W)\text{ times}})
    \end{aligned}
\end{equation}
where the permutation $\sigma$, once restricted to the block $W$, reads $\sigma|_W=\sigma_1^W\cdot\ldots\cdot\sigma_{r(W)}^W$ and $l_i^W$ is the length of the cycle $\sigma_i^W$. In the last equality we used the idempotence of the projectors $(P_K)^n=P_K$ and the fact that $\Tr (P_K)=\alpha N$. By applying the definition (\ref{classiccum}) of the classical cumulants, it can be easily found that
\begin{align}
    k_{r(W)}(\underbrace{\alpha N,\ldots,\alpha N}_{r\text{ times}})=\alpha N\delta_{r(W),1}\,.
\end{align}
We see that $r$ is forced to be $1$, or the whole expression becomes zero; since $r(W)$ is the number of cycles of $\sigma$ in the block $W$, this means that the expression is non-zero only when each block of $\W$ contains exactly one cycle of $\sigma$. Hence we obtain
\begin{align}
    \varphi(\W,\sigma)[\underbrace{P_K,\ldots,P_K}_{p+q\text{ times}}]=(\alpha N)^{\#\W}\delta(\W,\sigma)
\end{align}
where $\delta(\W,\sigma)$ is an operator which enforces that $\W=\sigma$ (notice that in this equality $\sigma$ is regarded as a partition, hence the order of the elements in its cycles is not important).

We can now decompose the $\kappa$-factor in the sum of (\ref{momOK}) using \ref{mompartperm}. The final result is
\begin{equation}\label{final}
    \begin{aligned}
        \kappa(\V,\pi)[\underbrace{\tilde{O},\ldots,\tilde{O}}_{p+q\text{ times}}]=\sum_{\substack{(\M,\tau)\in\PS(p+q)\\\M\leq\V}}\delta(\M,\tau)N^{\#\M}\left(\prod_{M\in\M}f_{|M|} \right)\cdot C_{\pi\vee\M,\V}(\tau\pi^{-1})
    \end{aligned}
\end{equation}
where $|M|$ is the length of the block $M\in\M$ and $f_{|M|}$ is the $|M|$-th first-order moment of $f(x)$, the distribution of the eigenvalues of $O$, taken to be independent.

The last term we need to compute is $C_{\pi\vee\M,\V}(\tau\pi^{-1})$ using its definition in (\ref{Cfactor}). Since we are interested in the limit $N\to\infty$, we can use the asymptotic behavior of the Weingarten function \cite{collins2003moments} to obtain:
\begin{equation}
    \begin{aligned}
        C_{\pi\vee\M,\V}(\tau\pi^{-1})&=\sum_{\substack{\N\in\PP(p+q)\\\pi\vee\M\leq\N\leq\V}}(-1)^{\#\N-\#\V}\prod_{V\in\V}\left(\#(\N|_V)-1 \right)!\cdot\\
        &\cdot N^{\#(\tau\pi^{-1})-2p-2q}\left[\prod_{c\in\tau\pi^{-1}}(-1)^{l_c-1}\frac{(2l_c-2)!}{(l_c-1)!l_c!}+O(N^{-2}) \right]
    \end{aligned}
\end{equation}
where $c$ is a cycle of $\tau\pi^{-1}$ and $l_c$ is its length. We can, at this point, insert this expression in (\ref{final}), and plug everything in (\ref{momOK}).     However, the resulting expression is quite complicated, and we were unable to simplify it significantly further. Nevertheless, one can proceed with symbolic computations for specific values of $p$ and $q$ \footnote{ Notice that even though we adopt the un-normalized trace here, we have the moments of $O_K$, when not restricted to the subspace onto which we are projecting, are the same as the moments of $O_K$ when it is instead restricted. The reason is that the trivial zero eigenvalues of the non-restricted operator do not contribute to the traces.}.
 
 Hence, we can compare the result directly with a $K\times K$ GUE random matrix $H$ which has $\E[H_{ij}]=f_1\delta_{ij}$ and $k_2(H_{ij})=\frac{\alpha \sigma_f^2}{K}$. It is easy to check that
\begin{align}\label{finalGUE}
    \varphi_2(H^p,H^q)=\sum_{n=1}^p\sum_{m=1}^q\binom{p}{n}\binom{q}{m}\alpha^{\frac{n+m}{2}}\sigma_f^{n+m}(f_1)^{p+q-n-m}|NC_2(n,m)|\,.
\end{align}
We can see in \cite{Dru} that the number of $(n,m)$-annular non-crossing pairings is given by
\begin{align}
    |NC_2(n,m)|=\sum_{k=1}^{\min(n,m)}|\text{Ann}_k(n,m)|
\end{align}
with
\begin{equation}
\begin{aligned}
    &|\text{Ann}_k(n,m)|=\\
    &=
    \begin{dcases}
        \frac{k}{nm}\sum_{d|\left(n,\frac{n-k}{2},\frac{m-k}{2}\right)}\phi(d)\binom{n/d}{(n-k)/d}\binom{m/d}{(m-k)/d},\qquad &\substack{(n-k)\\(m-k)}\text{ even}\\
        0 &\text{otherwise}
    \end{dcases}
\end{aligned}
\end{equation}
where $d$ runs over the common divisors of $\left(n,\frac{n-k}{2},\frac{m-k}{2}\right)$ and $\phi(d)$ is Euler's totient number.

%Eq. (\ref{finalGUE}) can be computed too with Mathematica, and compared with the results for eq. (\ref{momOK}).
Now, we are ready to compute moments of (\ref{momOK}) and compare with the ones from GUE (\ref{finalGUE}). The results are shown in the following, where we call $m_{p,q}=\varphi_2((O_K)^p,(O_K)^q)$ and $\mu_{p,q}=\varphi_2(H^p,H^q)$ for simplicity:
\begin{equation}
    \begin{aligned}
        &m_{1,1}=\mu_{1,1}+O(\alpha^2)=\alpha\sigma_f^2+O(\alpha^2)\\
        &m_{1,2}=\mu_{1,2}+O(\alpha^2)=2\alpha f_1\sigma_f^2+O(\alpha^2)\\
        &m_{1,3}=\mu_{1,3}+O(\alpha^2)=3\alpha f_1^2\sigma_f^2+O(\alpha^2)\\
        &m_{2,2}=\mu_{2,2}+O(\alpha^2)=4\alpha f_1^2\sigma_f^2+O(\alpha^2)\\
        &m_{2,3}=\mu_{2,3}+O(\alpha^2)=6\alpha f_1^3\sigma_f^2+O(\alpha^2)\,.
    \end{aligned}
\end{equation}
Recall that $m_{p,q}=m_{q,p}$, which allows us to compute only half of the moments. We conclude that, for the moments that we checked, we find agreement up to order $\alpha$. 
For $p,q\geq 3$, the computational cost becomes very high, since the number of partitioned permutations grows exponentially with $p+q$. Nevertheless, this result suggests that for small values of $\alpha$ the projected operator $O_K$ behaves as a GUE also at the level of two-point correlations, regardless of the properties of $O$. Of course, given the results of the previous section, we already know that the full probability distribution is that of the GUE. But it would be interesting to find a rigorous proof for all moments, using only Free Probability Theory.

\subsection{Free compression: revisit microcanonical truncations from free probability}\label{sec: FreeCompression}
(Asymptotic) freeness is an important feature for free operators in a non-commutative probability space. Other than Haar unitaries and random matrices, another way of generating a new non-commutative probability space from a given one is using a free projection, known as \textit{free compression} \cite{Nica_Speicher_2006}.
Consider a non-commutative probability space $(\A,\varphi)$, where $\A$ is a unital algebra and $\varphi=\Tr(\cdot)/N$ is the normalized trace, and a projection $p\in\A$ (i.e. $p^2=p$) such that $\varphi(p)\neq0$, we denote the compression by $(p\A p, \varphi^{p\A p})$, where
    \begin{equation}
        p\A p:=\{pap~|~a\in \A\}\,,
    \end{equation}
    and  
    \begin{equation}
        \varphi^{p\A p}(\cdot):=\frac{1}{\varphi(p)}\varphi(\cdot) ~~~ \mathrm{restrict~ to~ the ~algebra}~~p\A p\,.
    \end{equation}
    We denote the cumulants corresponding to $\varphi^{p\A p}$ as $\kappa^{p\A p}$, and $\kappa$ to be the cumulants for $\varphi$ for the whole algebra. Here, $(p\A p, \varphi^{p\A p})$ is indeed a new non-commutative probability space. Now, consider the case where the projection $p$ is free from the elements in $\A$. There exists a relation between the cumulants of $\{a_1,\dots, a_m\}\in\A$ and the cumulants of the compressed variables $\{pa_1p,\dots, pa_mp\}\in p\A p$:
    \begin{equation}
        \kappa_n^{p\A p}(pa_{i(1)}p,\dots, pa_{i(n)}p)=\frac{1}{\alpha}\kappa_n(\alpha a_{i(1)},\dots, \alpha a_{i(n)} )\,,
    \end{equation}
    for all $n\geq1$ and all $1\leq i(1), \dots,i(n) \leq m$. Such a relation of cumulants can be easily generalized to their $R$-transforms. In the simplest case where we start with only one non-commutative variable $a$, by considering the free additive convolution, this gives a surprising result: The renormalized distribution of $pap$ is given by
\begin{equation}
    \mu^{p \A p}_{pap}=\mu^{\boxplus 1/\alpha}_{\alpha a}\,.
    \label{eq-dof of free compression}
\end{equation}
An intuitive example is when $\alpha=1/2$, which yields
\begin{equation}
        \mu^{p \A p}_{pap}=\mu^{\boxplus 2}_{a/2}=\mu_{a/2}\boxplus\mu_{a/2}\,,
\end{equation}
which is equivalent to the eigenvalue distribution for the sum of two free variables with the same distribution as $a/2$. Another interesting observation is the limit where $\alpha \rightarrow 0$. This limit results in a distribution of $pap$ being an infinite sum of free variables with identical distributions of $\alpha a$, which is precisely captured by the free central limit theorem \cite{Voiculescu1986AdditionOC}, resulting in its convergence to the semicircle law:
\begin{equation}
    \mu^{p \A p}_{pap} \xrightarrow[]{\alpha \rightarrow0} \text{semicircle law}\,.
\end{equation}
In our case, the projection of $O$ via $PU O U^\dagger P$ is exactly a realization of free compression. To see it clearly, one can explicitly compute the moments of the projection $p=PU$ and $O$ using the Weingarten calculus, which are governed precisely by the non-crossing partitions, indicating the free independence between $PU$ and $O$. Therefore, the distribution of the freely projected matrix should follow (\ref{eq-dof of free compression}). To test its validity, we take the simple cases of spin-$1/2$ and spin-$1$ operators. 

\subsubsection*{Spin-$1/2$}\label{sec: freecompression for spinhalf}

For a spin-$1/2$ operator $O=a$ with projection parameter $\alpha$, the distribution is given by the Bernoulli distribution:
\begin{equation}
    \mu_{\alpha a}(\lambda)=\frac{\delta(\lambda+\alpha)+
\delta(\lambda-\alpha)}{2}\,.
\end{equation}
We can follow the procedures of free additive convolution we reviewed in Section \ref{sec: review of FPT} and compute the distribution of the sum of two spin-$1/2$ operators:
\begin{equation}
    \mu_{p a p}(\lambda)=\mu_{\alpha a}(\lambda)^{\boxplus 1/\alpha}\,.
\end{equation}
This example can be fully calculated analytically. For general finite $\alpha$, we have the Cauchy transform and $R$-transform as
\begin{equation}
    G_{\alpha a}(z)=\frac{z}{z^2-\alpha^2},~~~R_{\alpha a}(z)=
    G_{\alpha a}(z)=\frac{z}{z^2-\alpha^2}\,,
\end{equation}
where we pick the solution of $R$-transform that is regular in the vicinity of $z=0$. The $R$-transform for the projected variable is 
\begin{equation}
    R_{p a p}(z)=R_{\alpha a}*\frac{1}{\alpha}\,,
\end{equation}
which leads to the Cauchy transform of $pap$ as
\begin{equation}
    G_{pap}(z)=\frac{2\alpha z-z\pm \sqrt{z^2-4\alpha+4\alpha^2}}{2\alpha(z^2-1)}\,.
\end{equation}
One can therefore derive the distribution $\mu_{p a p}(\lambda)$ using the Stieltjes inversion formula (\ref{eq-StieltjiesInv}), which exactly matches with (\ref{Spin_trunc}) and (\ref{Spin_trunc2nd}) (first derived in \cite{SredIni, collins2005product}).

For the case with two different multiplicities (see Appendix \ref{sec:2eignum}), we adjust the Bernoulli distribution as:
\begin{equation}
    \mu'_{\alpha a}(\lambda)=\frac{3}{4}\delta(\lambda-\alpha)+\frac{1}{4}\delta(\lambda+\alpha)\,.
\end{equation}
Following the same steps, one can straightforwardly derive the distributions (\ref{2eigdistr}).

\subsubsection*{Spin-$1$}

For a spin-$1$ operator $O=\tilde{a}$ with projection parameter $\alpha$, the distribution is given by the generalized Bernoulli distribution:
\begin{equation}
    \mu_{\alpha \tilde{a}}(\lambda)=\frac{\delta(\lambda+\alpha)+\delta(\lambda)+
\delta(\lambda-\alpha)}{3}\,.
\end{equation}
We follow the same steps as in the spin-$1/2$ case and arrive at the $R$-transform for $p \tilde{a} p$
\begin{equation}
    R_{p \tilde{a} p}(z)=\frac{3 \alpha ^2 z^2+\left(\sqrt[3]{3 \sqrt{-\alpha ^2 z^2 \left(3 \alpha ^4 z^4+3 \alpha ^2 z^2+1\right)}+1}-1\right)^2}{3 \alpha  z \sqrt[3]{3 \sqrt{-\alpha ^2 z^2 \left(3 \alpha ^4 z^4+3 \alpha ^2 z^2+1\right)}+1}}\,.
    \label{eq-freeDec R-tranform sp1}
\end{equation}
In principle, we could continue for higher spin operators as well. However, as we increase the spin, it quickly gets harder to have analytical control, especially for the formal inverse of the Cauchy transform. However, it is still possible to perform an exact computation for certain given values of $\alpha$. We listed all the expressions of the distribution $\mu_{p \tilde{a} p}(\lambda)$ for values of $\alpha=0.01,0.25,0.5,0.75,0.99$ in the Appendix \ref{App: DOS for spin1}. Finally, we note that while it can be difficult to proceed using the $R$-transform method for operators of higher spin, it is possible to bypass the difficulties by considering the subordination method \cite{belinschi2007new} which can be very efficient, especially in numerical computations \cite{Camargo:2025zxr}.

\section{Discussion}
\label{sec:discussion}

In this paper, we studied quantum mechanical operators, of arbitrary spectra, Haar-randomly rotated, and projected down to a smaller Hilbert space. Physically, these should be viewed as operators projected down to a microcanonical window and the Haar-randomness comes from assuming the underlying Hamiltonian is chaotic. We then considered the eigenspectrum of this projected operator, and found a type of central limit theorem. In the large $N$ limit and small relative size of the projection, the probability distribution for the eigenvalues of the projected operator follows a GUE distribution, independently of any details of the original operator, other than the mean and variance. We derived this property directly from the Haar measure, once projected to a smaller Hilbert space, as well as using Free Probability Theory. We tested these results numerically for various types of operators, going from very ordered operators such as a spin operator, to operators whose spectrum is drawn randomly.\footnote{This should not be mistaken with a random matrix. Here, we mean an operator whose spectrum is drawn randomly, and hence follows a Poisson distribution for its level spacings. This should be contrasted with a true random matrix where the matrix elements are drawn randomly.} In both cases, in the universal regime, we find perfect agreement with the GUE distribution, while at larger relative size of the Hilbert space there are differences. The numerical results are summarized in Table \ref{tab:summary}. We conclude with some open questions.

\begingroup
\setlength{\tabcolsep}{5pt} 
\renewcommand{\arraystretch}{1.2} 
\begin{table}[t!]
    \centering
    \resizebox{\columnwidth}{!}{
    \begin{tabular}{|c|c|c|c|c|c|c|}
        \hline
        \multicolumn{2}{|c|}{Properties of}  &  semicircle& Wigner-Dyson & Level repulsion& \multirow{2}{*}{Sine kernel}\\ 
        \multicolumn{2}{|c|}{$PU O U^\dagger P$}& law & distribution & ($P(s=0)=0$)     &  \\ \cline{1-6}
        \multicolumn{2}{|c|}{$N\gg 1$, $K\gg 1$} & \multirow{2}{*}{\bluecheck} &\multirow{2}{*}{\bluecheck} & \multirow{2}{*}{\bluecheck} & \multirow{2}{*}{\bluecheck}\\
       \multicolumn{2}{|c|}{$\alpha \ll 1$, smooth $\rho$} &  &  & &   \\
        \hline
        $N\gg 1$, $K\gg 1$ & degeneracy & \multirow{2}{*}{\redcross} &\multirow{2}{*}{\redcross} & \redcross & \redcross  \\ 
        \cline{2-2}\cline{5-6}
       $0\ll\alpha< 1$, smooth $\rho$& non-degeneracy &  &  & \bluecheck&  \bluecheck \\
       \hline
       $N\gg 1$, $K\gg 1$ & degeneracy &\multirow{2}{*}{\redcross} &\multirow{2}{*}{\redcross} & \redcross & \multirow{2}{*}{\redcross} \\ \cline{2-2}\cline{5-5}
       $0\ll\alpha< 1$, non-smooth $\rho$ & non-degeneracy & &  &\bluecheck &   \\
        \hline
    \end{tabular}
    }
    \caption{ Summary of numerical results}
    \label{tab:summary}
   
\end{table}
\endgroup

\subsection*{Implications for BPS Chaos}

The main motivation for this work was to sharpen the physical intuition behind BPS chaos, or more generally chaos when there is a degenerate subspace of a Hamiltonian. The original proposal of \cite{Lin:2022rzw,Lin:2022zxd} was that level repulsion between eigenvalues of an operator projected to a BPS subspace would serve as a good diagnostic of the chaotic nature of that subspace.\footnote{See also \cite{Chen:2026vml} for a different proposal for a diagnostic of chaos in a BPS subspace, when there is a marginal coupling in the theory. See also \cite{Chang:2024lxt,Chryssanthacopoulos:2026mni} for a connection between fortuity and BPS chaos in the supersymmetric SYK model.} This was later refined in \cite{Chen:2024oqv} where the point of strong vs weak chaos was emphasized, meaning that a truly chaotic subspace should have random matrix correlations reaching out to the full size of the subspace. 

Our results make the connection sharper and more quantitative. If the BPS subspace is truly randomly oriented with respect to the eigenbasis of the unprojected operator, and the relative size of the subspace is small,\footnote{This seems to be the case in all known examples, with some extra subtleties in the D1D5 CFT \cite{Belin:2026vfd}.} we expect a semicircle law and a true GUE spectral form factor. Interestingly, this does not appear to be the case for the simplest operators in the supersymmetric SYK model \cite{Lin:2022zxd}. This means that few fermion operators seem to still retain some non-chaotic details of the Hamiltonian, perhaps related to the fact that it is $q$-local. It would be interesting to understand this better, and see if the chaotic nature of simple operators improves as $q$ increases.

There is a similar statement related to free probability: one would expect two different simple operators to be free, meaning that once projected, their relative eigenbases are random. For few fermion operators, this is not the case \cite{kyriakosunpublished}, again showing that some form of order remains. As the probe operator is taken to be more and more complicated (more fermions and hence heavier in scaling dimension), the freeness between any two such operators becomes better. This is consistent with the super-JT calculations of \cite{Lin:2022zxd}. It would be interesting to understand this phenomenon better.

\subsection*{Physical Systems}

The work done in this paper is a toy model: we used an actual Haar-random rotation, which means the spectrum of the projected operator becomes probabilistic, inheriting a probability distribution from the Haar measure. In a physical system with a definite choice of Hamiltonian and probe operator, there will be a definite unitary that rotates the two bases. An important question is to quantify how close to a true Haar-random unitary operator the physical rotation is. One way to quantify this, would be to study the spectrum of the unitary, and see how close it is to the circular ensemble (for example by studying its spectral form factor). While this is in principle doable, it seems to be complicated in general.

The more standard way to ask this question has been through Designs. $k$-Designs are probability distributions on unitaries that are not fully Haar-random, but agree with the Haar measure up to the $k$-th moment, see \cite{Roberts:2016hpo} for an early discussion in the high-energy context, or \cite{DiVincenzo:2001lru} for an earlier reference in quantum information theory.  At this stage, this is still probabilistic, but there also exist approximate $k$-designs (see for example \cite{Dankert:2009yux}), which could be used as true diagnostics of how close to a Haar random unitary a physical system really is. This could also help with the questions raised in the context of supersymmetric SYK models. We hope to return to this question in the future.

\subsection*{Large Projections}

One surprise of this work is that we could achieve quasi-random matrix statistics by projecting over large relative subspaces. This is \textit{not} a universal feature, and only works for specific spectra of the full operator. For example, for a spin operator, we saw that original eigenvalue recovery is guaranteed beyond a critical size of the projection. But if the probe operator had itself a non-degenerate spectrum, we found that projecting out only a small fraction of the Hilbert space sufficed to obtain random matrix statistics. It is important to emphasize that the probability distribution is \textit{not} the GUE one in that case, as for example the density of states is not a semicircle. But the level statistics look GUE-like. 

It is well known that the density of states is not a universal feature of chaotic systems, but it is the level statistics that are universal. It is also well known that small perturbations of an integrable Hamiltonian quickly make it chaotic. It is thus tempting to speculate that in such a context, we can rewrite the probability distribution as
\be
P(\lambda_j) = \int d\lambda \Delta_K^2(\lambda) e^{-V(\lambda)}
\ee
where $V(\lambda)$ is a single-trace potential. It is known that the spectral two-point function of such a random matrix model universally gives the sine kernel, even if the one-point function depends sensitively on the details of the potential $V$. It would be very interesting to try to derive such a result analytically, and to understand precisely which conditions on the initial unprojected spectrum need to be satisfied for it to be true. 

\subsection*{Random Tensors and Free Probability}

Another interesting avenue to explore is free probability for tensors. This is relevant for 3D gravity \cite{Belin:2020hea,Chandra:2022bqq,Belin:2023efa,Jafferis:2025vyp,deBoer:2023vsm,Collier:2023fwi, Collier:2024mgv,deBoer:2024mqg,deBoer:2025rct,Hartman:2025ula,Hartman:2025cyj,Chandra:2025fef,Chandra:2024vhm,Hung:2024gma,Belin:2026pko,Jafferis:2026gzn}, where heavy-heavy-heavy OPE coefficients are modeled as erratic 3-tensors, generalizing the erratic matrices of ETH. While much more recent than free probability of matrices, there also exists a free probability theory for tensors (see for example \cite{collins2025freecumulantsfreenessunitarily,nechita2025tensorfreeprobabilitytheory,bonnin2026freenesstensors,kunisky2024tensorcumulantsstatisticalinference,Lancien_2024}). It would be very interesting to see if some technology can be imported to study 3D gravity, and more generally OPE coefficients in chaotic CFTs.

\section*{Acknowledgements}
We are delighted to thank  Yiming Chen, Jan de Boer, Alice de Vito, Kyriakos Papadodimas, Julian Sonner, Rita Vadala, Alberto Zaffaroni for useful discussions. We are also thankful to Tolya Dymarsky for comments on the draft. YF was supported by the 1st Yu Ilhan scholarship from the Yuhan Foundation. YF also thanks the Department of Physics at the University of Milano-Bicocca for hospitality where part of this work was conducted. AI was used to simplify formulas in Section \ref{sec: jpdf from RMT} and to search for relevant literature.

\appendix

\section{Eigenvalue recovery: an algebraic proof\label{App: Eigenvalue Recovery}}
By looking at the cases of spin-$1/2$ in Section \ref{sec: Numerics} and spin-$1$ in Appendix \ref{App: DOS for spin1}, it seems that every eigenvalue of $O$ which has multiplicity $\gamma$ appears as an exact eigenvalue of $O_K$ when $\alpha=\frac{K}{N}>1-\frac{\gamma}{N}$. This is actually a general property, and the algebraic proof of such behavior is given in Theorem \ref{theomult} (this theorem can also be found in \cite{Thomp}). In order to prove it, we need a few definitions.
\begin{definition}
    Given an $N\times N$ matrix $A$, the $K$-th compound $C_K(A)$ is the matrix of the determinants of all the possible $K\times K$ minors of $A$, where a minor is built by selecting $K$ rows and $K$ columns of $A$; if the indices of the selected rows and column coincide, then it is a \textit{principal minor}.\\
    $C_K(A)$ is a $\binom{N}{K}\times\binom{N}{K}$ matrix and satisfy the properties
    \begin{align}
        C_K(AB)=C_K(A)\cdot C_K(B),\qquad C_K(A^{-1})=C_K(A)^{-1}
    \end{align}
\end{definition}
\begin{definition}
    For fixed integers $N$ and $K$, $1\leq K<N$, let $Q_{NK}$ denote the set of all sequences $\omega=\{i_1,i_2,\ldots,i_K\}$ of integers such that $1\leq i_1<i_2<\ldots<i_K\leq N$.
\end{definition}
\begin{definition}
    Let $A[\omega|\tau]$, with $\omega,\tau\in Q_{NK}$, be the $K\times K$ minor of $A$ obtained by selecting the rows labeled by $\omega$ and the columns labeled by $\tau$.
\end{definition}
\begin{definition}
    Let $f_{[\omega]}(\lambda)$ be the characteristic polynomial of $A[\omega|\omega].$
\end{definition}

\begin{theorem}\label{theomult}
    Consider the matrix $\tilde{A}=\text{diag}\left( \lambda_1,\lambda_2,\ldots,\lambda_N \right)$, with possibly degenerate eigenvalues. Let $\{\mu_i\}_{i=1\ldots D}$ be the distinct eigenvalues of $\tilde{A}$ and $\{\gamma_i\}_{i=1\ldots D}$ their multiplicities $(D\leq N)$. Now perform the usual change of basis and projection $A_K=P_KS\tilde{A}S^{-1}P_K$, where we let $S$ to be a generic non-singular $(\det S\neq0)$ invertible matrix. Then the multiplicity of $\mu_i$ as an eigenvalue of $A_K$ has a lower bound $\Ga_K$:
    \begin{align}
        \gamma_{A_K}(\mu_i)\geq\Gamma_K
    \end{align}
    where
    \begin{align}
        \Gamma_K=
        \begin{cases}
            K-(N-\gamma_{i})& \quad K\geq N-\gamma_{i}\\
            0&\quad K< N-\gamma_{i}\,.
        \end{cases}
    \end{align}
\end{theorem}
\begin{proof}
    If we define $A=S\tilde{A}S^{-1}$, then
    \begin{align}
        \lambda \mathbb{I}-A=S(\lambda\mathbb{I}-\tilde{A})S^{-1}.
    \end{align}
    In terms of the $K$-th compound, this reads
    \begin{align}
        C_K\left(\lambda\mathbb{I}-A \right)=C_K(S)C_K(\lambda\mathbb{I}-\tilde{A})C_K(S)^{-1}.
    \end{align}
    By definition,
    \begin{align}
        C_K(\lambda\mathbb{I}-A)_{\omega\omega}=f_{[\omega]}(\lambda),\qquad\omega\in Q_{NK}
    \end{align}
    while
    \begin{align}
        C_K(\lambda\mathbb{I}-\tilde{A})_{\tau\sigma}=\delta_{\tau\sigma}\prod_{j\in\tau}(\lambda-\lambda_j),\qquad\tau\in Q_{NK}
    \end{align}
    (note that $ C_K(\lambda\mathbb{I}-\tilde{A})$ is diagonal). Hence,
    \begin{align}\label{charpolK}
        f_{[\omega]}(\lambda)=\sum_{\tau\in Q_{NK}}\det S[\omega|\tau]\det S^{-1}[\tau|\omega]\prod_{j\in\tau}(\la-\la_j)
    \end{align}
    but
    \begin{align}
        \prod_{j\in\tau}(\la-\la_j)=f(\la)\prod_{j\notin\tau}(\la-\la_j)^{-1}
    \end{align}
    where $f(\la)$ is the characteristic polynomial of the full $A$. Hence,
    \begin{align}
        f_{[\omega]}(\lambda)=\sum_{\tau\in Q_{NK}}\det S[\omega|\tau]\det S^{-1}[\tau|\omega]f(\la)\prod_{j\notin\tau}(\la-\la_j)^{-1}\,.
    \end{align}
    The product $\prod_{j\notin\tau}(\la-\la_j)^{-1}$ removes $N-K$ factors from $f(\la)$. If an eigenvalue $\mu_i$ has multiplicity $\gamma_i$, then $f(\la)$ contains the factor $(\la-\mu_i)^{\gamma_i}$. Hence, if $\gamma_i\geq N-K$, in every term of the sum $\sum_{\tau\in Q_{NK}}$, the exponent of $(\la-\mu_i)$ cannot go under $\gamma_i-(N-K)$. This means that, for $K\geq N-\gamma_i$, $\mu_i$ is a solution of $f_{[\omega]}(\lambda)=0$ with multiplicity at least $K-(N-\gamma_i)$. On the other hand, if $\gamma_i<N-K$, then in general there will be specific choices of $\tau$ for which the product $\prod_{j\notin\tau}(\la-\la_j)^{-1}$ cancels all the factors $(\la-\mu_i)$, and hence the lower bound for the multiplicity of $\mu_i$ is zero. This proves the theorem.

    Notice that the theorem also applies in the specific case where $S$ is unitary, which is the case we are interested in.
\end{proof}

With the previous theorem, we proved that the multiplicity of the eigenvalues of $A$ as eigenvalues of $A_K$ has a lower bound. However, if one tries to extract numerically the matrix $S$ as a Haar unitary and compute the eigenvalues of $A_K$, one finds that the multiplicities of the original eigenvalues never exceed the lower bound. The reason is that the matrix $S$ that one has to pick in order to exceed the lower bound belongs to a set of Lebesgue measure zero, hence cannot be randomly extracted. In the following theorem, we prove this property, which is also valid in the case where $S$ is just invertible and non-singular.
\begin{theorem}\label{theomult2}
    If, in Theorem \ref{theomult}, $S\in GL_N(\mathbb{C)}$ is chosen randomly according to some probability measure, then \emph{almost surely} the multiplicity of $\mu_i$ as an eigenvalue of $A_K$ equals the lower bound
    \begin{align}\label{lower_bound}
        \gamma_{A_K}(\mu_i)=\Ga_K=
        \begin{cases}
            K-(N-\gamma_{i})& \quad K\geq N-\gamma_{i}\\
            0&\quad K< N-\gamma_{i}
        \end{cases}\quad .
    \end{align}
    This means that the probability of extracting an $S$ for which the lower bound is exceeded is zero.
\end{theorem}
\begin{proof}
    Take again (\ref{charpolK}):
    \begin{align}\label{charpolK2}
        f_{[\omega]}(\lambda)=\sum_{\tau\in Q_{NK}}\det S[\omega|\tau]\det S^{-1}[\tau|\omega]\prod_{j\in\tau}(\la-\la_j)\,,
    \end{align}
    which is a sum of $\binom{N}{K}$ terms.
    If we want $\mu_i$ to be a solution of $ f_{[\omega]}(\lambda)=0$ with multiplicity at least $\Ga_K+1$, then $(\la-\mu_i)$ must be present in each term of the sum in (\ref{charpolK2}) at least $\Ga_K+1$ times. Theorem \ref{theomult} only guarantees that such a term is present at least $\Ga_K$ times. This means that, if we want to exceed the lower bound, all the terms which contain $(\la-\mu_i)$ just $\Ga_K$ times must vanish: the matrix $S'$ that allows this to happen must satisfy the conditions 
    \begin{align}
        g_\tau(S')=\det S'[\omega|\tau]\cdot\det S'^{-1}[\tau|\omega]=0\qquad \forall\tau\in \tilde{Q}_{NK}
    \end{align}
    where

     \begin{align}
        \tilde{Q}_{NK}=\left\{\sigma\in Q_{NK} \middle| \text{ card}\left\{j\in\sigma\middle|\la_j=\mu_i \right\}=\Ga_K \right\}
    \end{align}
    The cardinality of $\tilde{Q}_{NK}$ is
    \begin{align}
        \text{card}(\tilde{Q}_{NK})=\binom{\gamma_i}{\Ga_K}\cdot\binom{N-\gamma_i}{K-\Ga_K}>0\,.
    \end{align}
    So we get $h=\text{card}(\tilde{Q}_{NK})$ conditions on $S'$, which define a submanifold $\mathcal{M}\subseteq GL_N(\mathbb{C)}$. By the implicit function theorem, the dimension of $\mathcal{M}$ is $N-r$ where $r$ is the rank of the Jacobian $J$ of the functions $g_\tau$ for $\tau\in\tilde{Q}_{NK}$.

    We have that $r$ cannot be zero: in fact, $r=0\Leftrightarrow J=0$, which means that all the $g_\tau$ should be constant on $GL_N(\mathbb{C)}$. To check this, we can use Jacobi's formula for computing the derivative of a determinant to compute the component $J_{mn}$ of the Jacobian, with $m\in\omega$ and $n\in\tau$. We get
    \begin{align}
        J_{mn}=\frac{\partial}{\partial S_{mn}}g_\tau(S)=\left[ \left(S[\omega|\tau]^{-1} \right)_{nm} - \left(S^{-1}\right)_{nm} \right]g_\tau(S)\,.
    \end{align}
    The $g_\tau$'s are zero on the submanifold, but not in general. Moreover, $S[\omega|\tau]^{-1}\neq S^{-1}[\omega|\tau]$ in general, which implies that $(S[\omega|\tau]^{-1})_{nm}\neq (S^{-1}[\omega|\tau])_{nm}=(S^{-1})_{nm}$. This is enough to deduce that $J$ is not identically zero on $GL_N(\mathbb{C)}$. 
    Hence,
    \begin{align}
       \dim\mathcal{M}= N-r<N\,,
    \end{align}
    which implies that the Lebesgue measure $\mathcal{M}$ is zero. This proves that the probability of extracting an $S'\in\mathcal{M}$ is zero.
\end{proof}

With the last two theorems, we have learned that, if the spectrum of $O$ is degenerate, the statistics of $O_K$ are different from the ones predicted by RMT, since beyond certain values of $K$ a discrete spectrum is superimposed on the continuous one. If we want to keep the RMT description, we need to stay below those transition values, or take an operator $O$ with a fully non-degenerate spectrum, in which case all the multiplicities are equal to $1$. In the limit $N\to\infty$ with $\alpha=\frac{K}{N}$ finite, we expect to observe some RMT correlations in eigenvalues and level repulsions for $\alpha\in(0,1)$, which was analytically shown in Section \ref{sec: jpdf from RMT}.

\section{The projection of the operator with two eigenvalues from the $S$-transform\label{App-Proj2eig}}

Knowing that $P_K$ and $U O U^\dagger$ are free from each other, this allows us to use the free multiplicative convolution to derive the expression of the one-point eigenvalue distribution of $O'_K$, essentially $O_K$.
Consider $ O$ has two eigenvalues in the following form:
\begin{align}
    O = \text{diag}(&\underbrace{1,\ 1,\ \dots,\ 1}_{\gamma \text{ times}},\ 
                   \underbrace{0,\ 0,\ \dots,\ 0}_{N-\gamma \text{ times}})\,.
\end{align}
Notice that in Appendix \ref{sec:2eignum} we studied numerically a slightly different case, where the eigenvalues were $\pm 1$; here we need to consider operators with non-vanishing mean, otherwise the S-transform cannot be defined. This is not too bad since we can always relate the two cases with the shift-rescaling of (\ref{shift_resc}) after performing the S-transform. This case is simpler because both $O$ and $P_K$ are projectors, hence share the same asymptotic distributions:
\begin{align}
    P_K&\sim\mu_1(\lambda)=\left((1-\alpha)\delta_0+\alpha\delta_1 \right)\\
    U O U^\dagger&\sim \mu_2(\lambda)=\left((1-\beta)\delta_0+\beta\delta_1 \right)
\end{align}
where $\alpha=\lim_{K,N\to\infty}\frac{K}{N}$ and $\beta=\lim_{\gamma,N\to\infty}\frac{\gamma}{N}$, while $\delta_x$ is the Dirac delta centered in $x$ represented as a functional. The probability distribution of $O_K$ is then
\begin{align}
    \rho_{O_K}=\mu_{P_K}\boxtimes\mu_{U O U^\dagger}\,.
\end{align}

Let us apply the S-transform procedure. The moments of each projector are
\begin{align}
    m_n(\mu_i)=\int d\lambda \lambda^n\left((1-\alpha_i)\delta_0+\alpha_i\delta_1 \right)=\alpha_i\,,
\end{align}
where $\alpha_1=\alpha$, $\alpha_2=\beta$, and $X_1=P_K$, $X_2=U O U^\dagger$. Then the moment generating formal power series is
\begin{align}
    M_{\mu_i}(z)=\sum_{n=1}^\infty m_n(\mu_i)z^n=\frac{\alpha_i z}{1-z}\,,
\end{align}
where we used the formal expression for the geometric series. The inverse of the moment series is then
\begin{align}
    M_{\mu_i}^{-1}(z)=\frac{z}{z+\alpha_i}\,,
\end{align}
which implies that the S-transforms (see (\ref{eq-Stranform})) are
\begin{align}
    S_{\mu_{P_K}}(z)=\frac{1+z}{\alpha+z},\qquad S_{\mu_{U O U^\dagger}}(z)=\frac{1+z}{\beta+z}\,.
\end{align}
Now we know, that the S-transform of $\rho_{O_K}$ reads
\begin{align}
    S_{\rho_{O_K}}(z)=S_{\mu_{P_K}}(z) \cdot S_{\mu_{U O U^\dagger}}(z)=\frac{(1+z)^2}{(z+\alpha)(z+\beta)}\,.
\end{align}
Knowing $S_{\rho_{O_K}}$, we can go back to the inverse moment series of $\rho_{O_K}$:
\begin{align}
    M_\rho^{-1}(z)=\frac{z}{1+z}S_\rho(z)=\frac{z(1+z)}{(z+\alpha)(z+\beta)}\,.
\end{align}
This leads to the moment function
\begin{equation}
    M_{\rho_{O_K}}(z)=\frac{1-z(\alpha+\beta)\pm\sqrt{(r_+z-1)(r_-z-1)}}{2(z-1)}
\end{equation}
where
\begin{equation}
    r_\pm=\alpha+\beta-2\alpha\beta\pm\sqrt{4\alpha\beta(1-\alpha)(1-\beta)}\,.
\end{equation}

However, we know from Theorem \ref{theomult} and \ref{theomult2} that $\rho$ is the superposition of a continuous and a deterministic distribution due to the eigenvalue recovery phenomenon. Moreover, since here $O_K$ is not restricted to the subspace onto which we are projecting, we have to count the trivial $N-K$ vanishing eigenvalues. Hence, the multiplicity of $\{0\}$ as an eigenvalue of $O_K$ is obtained by summing $N-K$ to (\ref{lower_bound}):
\begin{equation}
\begin{aligned}
    \frac{1}{N}\gamma_{O_K}(0)&=
    \begin{cases}
        1-\beta,\qquad &\alpha\geq \beta\\
        1-\alpha, &\alpha<\beta
    \end{cases}\\
    &=1-\min(\alpha,\beta)\equiv\mathcal{N}_0(\alpha,\beta)\,.
\end{aligned}
\end{equation}
The multiplicity of $\{1\}$, on the other hand, is simply given by (\ref{lower_bound}):
\begin{equation}
\begin{aligned}
    \frac{1}{N}\gamma_{O_K}(1)&=
    \begin{cases}
        \alpha+\beta-1,\qquad &\alpha\geq 1-\beta\\
        1-\alpha, &\alpha<1-\beta
    \end{cases}\\
    &=\max(\alpha+\beta-1,0)\equiv\mathcal{N}_1(\alpha,\beta)\,.
\end{aligned}
\end{equation}
This allows us to separate the continuous part of the distribution from the deterministic one:
\begin{align}
    \rho_{O_K}=\mathcal{N}(\alpha,\beta)\tilde{\rho}_{O_K}\,+\mathcal{N}_0(\alpha,\beta)\delta_0+\mathcal{N}_1(\alpha,\beta)\delta_1\,,
\end{align}
where $\mathcal{N}=1-\mathcal{N}_0-\mathcal{N}_1$ and $\tilde{\rho}_{O_K}$ is a continuous probability distribution. The relation between $\rho_{O_K}$ and $\tilde{\rho}_{O_K}$ imply a relation between the moment generating series:
\begin{equation}
    \begin{aligned}
        M_{\tilde{\rho}_{O_K}}(z)&=\frac{1}{\mathcal{N}}\left(M_\rho(z)-\mathcal{N}_1\frac{z}{1-z} \right)\\
        &=\frac{1+(2\mathcal{N}_1-\alpha-\beta)z\,\pm\sqrt{(r_+z-1)(r_-z-1)}}{2(z-1)\mathcal{N}}\,.
    \end{aligned}
\end{equation}
The Cauchy transform is therefore:
\begin{equation}
    \begin{aligned}
        G_{\tilde{\rho}_{O_K}}(z)&=\frac{1}{z}\left(1+M_{\tilde{\rho}_{O_K}}\left(\frac{1}{z} \right) \right)\\
        &=\frac{(1-2\mathcal{N})z+2(1-\mathcal{N}_0)-(\alpha+\beta)\pm\sqrt{(r_+-z)(r_--z)}}{2\mathcal{N}(1-z)z}\,.
    \end{aligned}
\end{equation}
Finally, we compute $\tilde{\rho}_{O_K}$ by applying the Stieltjes inversion formula:
\begin{equation}
    \begin{aligned}
        \tilde{\rho}_{O_K}(\lambda)&=-\frac{1}{\pi}\lim_{\varepsilon\to 0}\mathrm{Im}~G_{\tilde{\rho}_{O_K}}(\lambda+i\varepsilon)=\\
        &=-\frac{1}{\pi}\mathrm{Im}~\left[\frac{(1-2\mathcal{N})\lambda+2(1-\mathcal{N}_0)-(\alpha+\beta)-\sqrt{(r_+-\lambda)(r_--\lambda)}}{2\mathcal{N}(1-\lambda)\lambda}\right]
    \end{aligned}
\end{equation}
where we have chosen the ``$-$" branch of the square root in order to have a positive distribution $\tilde{\rho}_{O_K}$. This distribution vanishes except when $(r_+-\lambda)(r_--\lambda)<0$, within the range for $\lambda\in(r_-,r_+)$, which yields
\begin{align}\label{2eigdistrth}
    \tilde{\rho}(\lambda)=
    \begin{cases}
    \frac{\sqrt{(r_+-\lambda)(\lambda-r_-)}}{2\pi\lambda(1-\lambda)\mathcal{N}},\qquad &\lambda\in(r_-,r_+)\\
    0 &\lambda\notin (r_-,r_+)
    \end{cases}
\end{align}
with $\mathcal{N}=\min(\alpha,\beta)-\max(\alpha+\beta-1,0)$.

When we take $\beta=\frac{1}{2}$ in (\ref{2eigdistrth}) we get (\ref{Spin_trunc}) and (\ref{Spin_trunc2nd}) after a proper shift and rescaling. Similarly, if we set $\beta=\frac{3}{4}$ and shift-rescale,  we obtain (\ref{2eigdistr}).

With this procedure, one can, in principle, compute the distribution $\rho$ of $O_K$ for any operator $O$ distributed according to
\begin{align}
    \mu=\sum_{i=1}^D\alpha_i\delta_{\lambda_i}
\end{align}
where $\lambda_1,\ldots,\lambda_D$ are the distinct fixed eigenvalues of $O$ following the same steps. However, finding the inverse of the moment function for more eigenvalue cases can be hard to solve analytically.

In the numerical computations, we saw that, for generic values of $\alpha$, the statistics of the eigenvalues depended on the choice of the spectrum of $O$, with a discrete spectral density that might appear superimposed to a continuous one in the case of degeneracy. 

\section{Proof of Theorem \ref{Theorem-1pt from S trans}\label{App-TheoremProof}}

 From the previous section, we know that the S-transform of a projector is
    \begin{align}
        S_\mu(z)=\frac{1+z}{\alpha+z}\,.
    \end{align}
    Consider the moment generating series of $f$
    \begin{align}
        M_f(z)=\sum_{k=1}^\infty f_kz^k\,,
    \end{align}
    where $f_k=f[X^k]$ is the $k$-th moment of $f$, and denote the formal inverse of such series as
    \begin{align}
        M_f^{-1}(z)=\sum_{k=1}^\infty \F_kz^k\,.
    \end{align}
    In order for $M_f^{-1}$ to be well defined, we have to assume that $f_1\neq 0$ (the reason will become clear soon). This is not too bad as it seems since, in this specific problem (random rotation with $U$ and then projection), a shift in the distribution of $O$ results in the same shift in the distribution of $O_K$, by a reasoning similar to that of (\ref{shift_resc}). So, if we want to study the case $f_1=0$, it is enough to shift $f$ by any positive quantity $a$, compute the free convolution with the S-transform, and shift back.

    The S-transform of $\rho$ reads
    \begin{align}
        S_\rho(z)=S_\mu(z)S_f(z)=\frac{(1+z)^2}{z(\alpha+z)}M_f^{-1}(z)
    \end{align}
    from which we can compute $M_\rho^{-1}$ using the definition of the S-transform:
    \begin{align}
        M_\rho^{-1}(z)=\frac{1+z}{\alpha+z}M_f^{-1}(z)\equiv\sum_{n=1}^\infty\R_nz^n\,.
    \end{align}
If we write $M_f^{-1}$ in terms of the $\F_k$ and expand the fraction $\frac{1+z}{\alpha+z}$ as a formal power series, we can find the expression of the $\R_n$'s written in powers of $\alpha$.
\begin{align}\label{RF}
    \R_n=\sum_{l=1}^n(-1)^{l+1}(\F_{n-l+1}+\F_{n-l})\alpha^{-l},\qquad n\geq1
\end{align}
with the condition $\F_0=0$.

Since the moments $\{\rho_n\}_n$ of $\rho$ (where $\rho_n=\rho[x^n]$) and the $\{\R_n\}_n$ are the coefficients of two formal power series which are one the inverse of the other (with respect to function composition), they can be related using the Lagrange inversion theorem for formal power series \cite{charalambides2002enumerative}:
\begin{align}\label{rhoR}
    \begin{dcases}
        \rho_1=\frac{1}{\R_1}\\
        \rho_n=\frac{1}{n!}\frac{1}{\R_1^n}\sum_{k=1}^{n-1}(-1)^kn^{\bar{k}} B_{n-1,k}(\hat{\R}_1,\hat{\R}_2,\ldots,\hat{\R}_{n-k}),\qquad n\geq2
    \end{dcases}
\end{align}
where $\hat{\R}_l=l!\R_{l+1}\R_1^{-1}$, $\displaystyle B_{n-1,k}(\hat{\R}_1,\hat{\R}_2,\ldots,\hat{\R}_n-k)=n!\sum_J\prod_{i=1}^{n-k}\frac{\hat{\R}_i^{j_i}}{(i!)^{j_i}j_i!}$ are the exponential Bell polynomials, $J=(j_1,\ldots,j_{n-k})$ is a multi-index, subject to the constraints $\sum_{i=1}^{n-k}j_i=k$ and $\sum_{i=1}^{n-k}ij_i=n$, with $j_i=0\ldots k$, and $n^{\bar{k}}=n(n+1)\ldots(n+k-1)$ is the rising factorial.

Now we can substitute (\ref{RF}) in (\ref{rhoR}) to find an expansion of the general $\rho_n$ in powers of $\alpha$. After some tedious calculations, we obtain 
\begin{align}\label{rhoF}
    \begin{dcases}
        \rho_1=\frac{\alpha}{\F_1}\\
        \rho_n=\sum_{k=1}^{n-1}\frac{(-1)^kn^{\bar{k}}}{\F_1^{n+k}}\sum_J \sum_{\textbf{P}}\left(\prod_{m=1}^{n-k}\prod_{l_m=1}^{m+1}\Omega_{m,l_m,\textbf{P}}(\F) \right)\,\alpha^{n+k-\textbf{E}_{\textbf{P}}},\qquad n\geq2
    \end{dcases}
\end{align}
where $\displaystyle \sum_J$ is the sum that appears in the definition of the Bell polynomials, $\displaystyle \sum_\textbf{P}=\prod_{m=1}^{n-k}\sum_{P^{(m)}}$ where $P^{(m)}=\left(p_{1}^{(m)},\ldots,p_{m+1}^{(m)}\right)$ is a multi index subject to the constraint $\sum_{l=1}^{m+1}p_l^{(m)}=j_m$ and such that each $p_l^{(m)}$ runs from $0$ to $j_m$, $\displaystyle \textbf{E}_\textbf{P}=\sum_{m=1}^{n-k}\sum_{l=1}^{m+1}l\,p_l^{(m)} $,
$\displaystyle \Omega_{m,l,\textbf{P}}(\F)=\frac{1}{p_l^{(m)}!}(-1)^{p_l^{(m)}(l+1)}(\F_{m-l+2}+\F_{m-l+1})^{p_l^{(m)}}$.

Since we want to study the limit $\alpha\to 0$ of the moments of $\rho$, let us compute the lowest orders in $\alpha$ of expression (\ref{rhoF}). In particular, we will compute $\rho_n$ up to order $\alpha^2$ and compare it with the same expansion of $\rho_\text{lim}$. In order to do that, we need to apply the Lagrange inversion theorem to write the $\{\F_k\}_k$ in terms of the moments $\{f_k\}_k$:
\begin{align}\label{Ff}
    \begin{dcases}
        \F_1=\frac{1}{f_1}\\
        \F_n=\frac{1}{n!}\frac{1}{f_1^n}\sum_{k=1}^{n-1}(-1)^kn^{\bar{k}} B_{n-1,k}(\hat{f}_1,\hat{f}_2,\ldots,\hat{f}_{n-k}),\qquad n\geq2\,.
    \end{dcases}
\end{align}

%\begin{description}
    %\item[Order 1] 
    \textbf{Lowest order in $\alpha$}
    
    To get the lowest order in $\alpha$ of $\rho_n$ for $n\geq2$, $\textbf{E}_\textbf{P}$ must be the maximum possible. Notice that $\sum_{l=1}^{m+1}p_l^{(m)}=j_m$, we can regard $\textbf{E}_\textbf{P}$ as the total energy of a system composed of $n-k$ separated subsystems: the $m$-th subsystem contains $j_m$ particles occupying $m+1$ energy levels $E_l=l$ for $l=1,\ldots,m+1$. Since $\textbf{P}$ is the multi-index that contains all the occupation numbers of the fictitious system, we can call it a \textit{configuration}.

    There is only one configuration $\textbf{P}^{(1)}$ that maximizes energy: the one in which all particles occupy the highest energy state of each subsystem: $p_{m+1}^{(m)}=j_m$ and $p_l^{(m)}=0$ for $l<m+1$, $\forall m$ (notice that $j_m$ could also be zero). It is easy to check that for this configuration $\textbf{E}_{\textbf{P}^{(1)}}=n+k-1$, which corresponds to the terms of order $\alpha$:
    \begin{align}
        \rho_n^{(1)}=\frac{f_1^n}{n}\sum_{k=1}^{n-1}(-1)^kn^{\bar{k}}f_1^k\sum_J\left(\prod_{m=1}^{n-k}\Omega_{m,m+1,\textbf{P}^{(1)}}(\F) \right)\alpha
    \end{align}
    where we used the fact that $\Omega_{m,l,\textbf{P}^{(1)}}(\F)=1$ for $l<m+1$. By substituting (\ref{Ff}) in $\Omega_{m,m+1,\textbf{P}^{(1)}}(\F)$, we obtain
    \begin{align}
        \rho_n^{(1)}=\frac{f_1^n}{n}\sum_{k=1}^{n-1}(-1)^{n-k-1}n^{\bar{k}}\sum_J\frac{1}{j_1!\ldots j_{n-k}!}\alpha\,.
    \end{align}
    Now, using the definition of the Bell polynomials and their relation with the Lah numbers, we have that
    \begin{equation}
        \sum_J\frac{1}{j_1!\ldots j_{n-k}!}=\frac{1}{(n-1)!}B_{n-1,k}(1!,2!,\ldots,(n-k)!)=\frac{1}{(n-1)!}L(n-1,k)\,,
    \end{equation}
    where $L(n-1,k)$ is a Lah number.
    Hence we obtain
    \begin{equation}
            \rho_n^{(1)}=\frac{f_1^n}{n!}\left(\sum_{k=1}^{n-1}(-1)^{n-k-1}n^{\bar{k}}L(n-1,k) \right)\alpha=\alpha f_1^n\,,
    \end{equation}
    where we used the following identity that relates the rising and falling factorials through the Lah numbers:
    \begin{align}
        \sum_{k=1}^{n-1}(-1)^{n-1-k}L(n-1,k)x^{\bar{k}}=x^{\underline{n-1}}=\frac{x!}{(x-n+1)!}\,,
    \end{align}
   here $x^{\underline{n}}:=\prod_{k=0}^{n-1}(x-k)$ is the falling factorial.

    %\item[Order 2]
     \textbf{Next to lowest order in $\alpha$}
    
    To get the order $\alpha^2$ we need to find the configurations $\{\textbf{P}_i^{(2)}\}_i$ for which $\textbf{E}_{\textbf{P}_i^{(2)}}=\textbf{E}_{\textbf{P}^{(1)}}-1$. Such configurations are those with all the particles in the highest possible energy level in all the systems, except for just one particle in one system, which is on the second-highest level. Here, we have to be careful since the number of particles $j_m$ in the system $m$ could also be zero, in which case all the occupation numbers of that system must vanish. Hence, $\textbf{P}_i^{(2)}$ is such that for the system $i$ (whose highest state has energy $i+1$) we have $p_{i+1}^{(i)}=j_i-1$, $p_i^{(i)}=1$ and $p_l^{(i)}=0$ for $l<i$; for all the other systems labeled by $m\neq i$, $p_{m+1}^{(m)}=j_m$ and $p_l^{(m)}=0$ for $l<m+1$ as in the previous case. This configuration is labeled by the index $i$, which runs from $1$ to $n-k$ while $j_i\neq 0$. It can be easily checked that for these configurations $\textbf{E}_{\textbf{P}_i^{(2)}}=n+k-2$, hence are associated to terms of order $\alpha^2$ in (\ref{rhoF}):
    \begin{equation}
            \rho_n^{(2)}=\frac{f_1^n}{n}\sum_{k=1}^{n-1}(-1)^kn^{\bar{k}}f_1^k \sum_J\sum_{\substack{i=1\\ j_i\neq0}}^{n-k}\Omega_{i,i+1,\textbf{P}_i^{(2)}}(\F) \Omega_{i,i,\textbf{P}_i^{(2)}}(\F) \prod_{\substack{m=1\\m\neq i}}^{n-k}\Omega_{m,m+1,\textbf{P}_i^{(2)}}(\F) \alpha^2\,.
    \end{equation}
    Also in this case, we can easily compute the $\Omega$'s by inserting the expression for the $\F$'s. Even if at first sight it seems complicated to manage the sum $\sum_{i=1}^{n-k}$ with the constraint $j_i\neq0$, after doing the calculations we find that $\rho_n^{(2)}$ is proportional to $\displaystyle \sum_{\substack{i=1\\ j_i\neq0}}^{n-k}j_i$, hence the constraint can be removed. Then, by using the same relation between the Bell polynomials and the Lah numbers that we used for order 1, we obtain
    \begin{equation}
        \begin{aligned}
            \rho_n^{(2)}=\frac{f_1^{n-2}\sigma_f^2}{n!}\left(\sum_{k=1}^{n-1}(-1)^{n-1-k}n^{\bar{k}}k\,L(n-1,k) \right)\alpha^2=\frac{n(n-1)}{2}\sigma_f^2\,(f_1)^{n-2}\alpha^2\,,
        \end{aligned}
    \end{equation}
    where we used the property of the rising factorial $kn^{\bar{k}}=n(n+1)^{\bar{k}}-nn^{\bar{k}}$.
%\end{description}

We can finally put together the results to get an expression for $\rho_n$ up to order $\alpha^2$:
\begin{align}
    \begin{dcases}
        \rho_1=\alpha f_1\\
        \rho_n=\alpha(f_1)^n+\alpha^2\frac{n(n-1)}{2}\sigma_f^2\,(f_1)^{n-2}+O(\alpha^3),\qquad n\geq2\,.
    \end{dcases}
\end{align}
This proves (\ref{final1stord}).

\section{More examples of projection of spin operators}

\subsection{The projection of a spin-$1$ operator}

For completeness, we also give another example of a spin operator, which is a spin-$1$ operator whose spectrum follows the generalized Bernoulli distribution
\begin{equation}
\rho(\lambda)=\frac{\delta(\lambda+1)+\delta(\lambda)+\delta(\lambda-1)}{3}\,.
\end{equation}
This is realized by considering the spectrum of the operator $O$ as $\{1,0,-1\}$ with multiplicities $(333,333,334)$ respectively. Numerical results are shown in Figures \ref{fig:spin1_stats} and \ref{fig:spin1_stats2}, where we find the transition point to be $\alpha=2/3$, similar to the spin-$1/2$ case. Beyond this value, degeneracy of eigenvalues is recovered. In the limit $\alpha \ll1$, the Wigner semicircle is found, as well as the level repulsion and RMT-behaved SFFs, which remain until the crossing of the transition point. The density of states for $\alpha=1/2$ exhibits a ``ghost-like" distribution, which is exactly the distribution of the sum of two free spin-$1$ operators \cite{Camargo:2025zxr}. Notably, all the continuous parts of the density of states for finite $\alpha$ can be derived using the free compression method, as shown in Section \ref{sec: FreeCompression}.

\begin{figure}[htbp]
    \centering
    
    \begin{subfigure}[b]{0.3\textwidth}
        \includegraphics[width=0.85\textwidth]{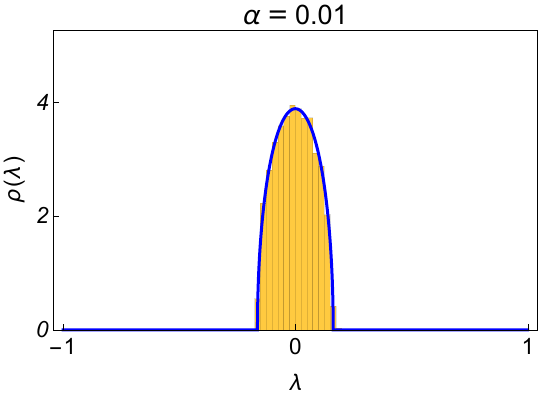}
        \caption{}
    \end{subfigure}
    \hfill
    \begin{subfigure}[b]{0.3\textwidth}
        \includegraphics[width=\textwidth]{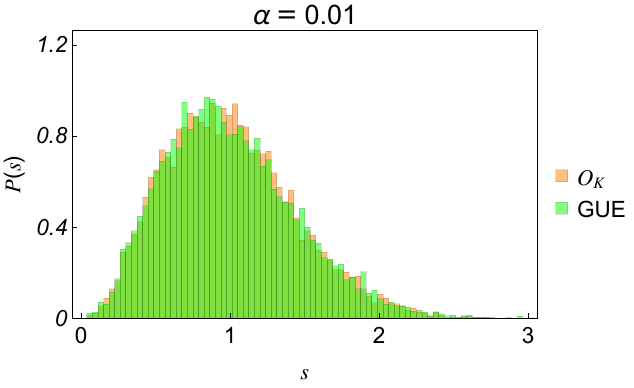}
        \caption{}
    \end{subfigure}
    \hfill
    \begin{subfigure}[b]{0.3\textwidth}
        \includegraphics[width=1.1\textwidth]{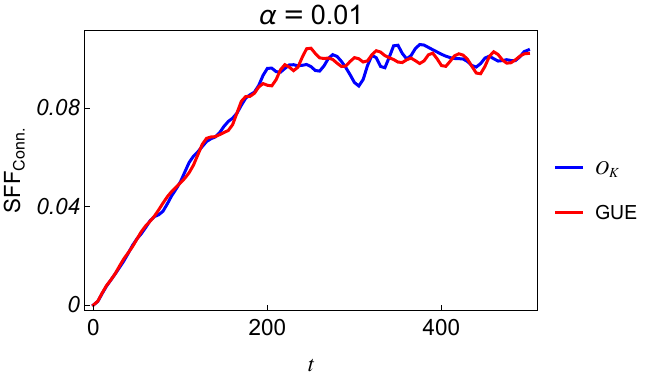}
        \caption{}
    \end{subfigure}
    
    \medskip
    
    \begin{subfigure}[b]{0.3\textwidth}
        \includegraphics[width=0.85\textwidth]{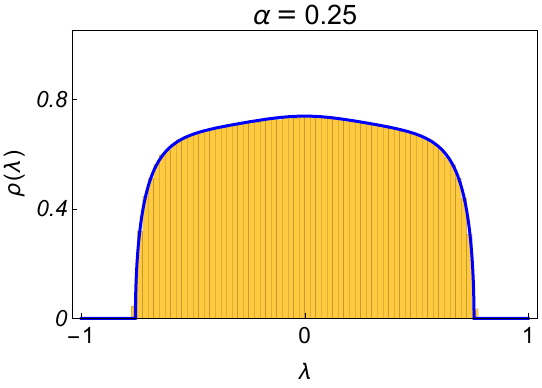}
        \caption{}
    \end{subfigure}
    \hfill
    \begin{subfigure}[b]{0.3\textwidth}
        \includegraphics[width=\textwidth]{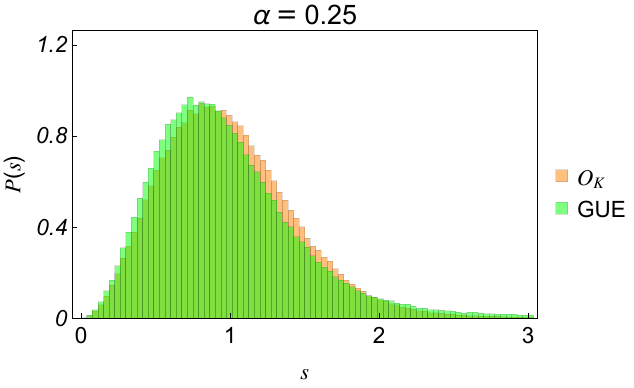}
        \caption{}
    \end{subfigure}
    \hfill
    \begin{subfigure}[b]{0.3\textwidth}
        \includegraphics[width=1.1\textwidth]{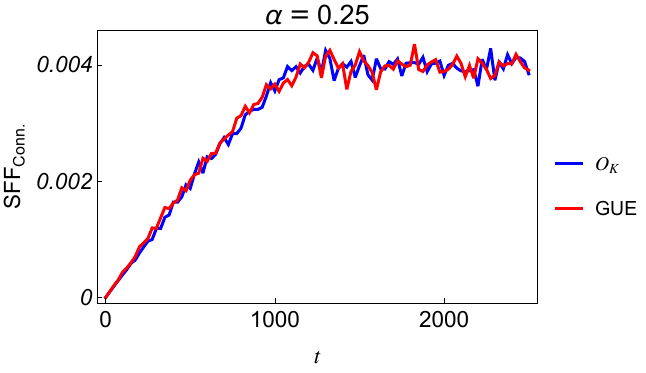}
        \caption{}
    \end{subfigure}
    
    \medskip
    
    \begin{subfigure}[b]{0.3\textwidth}
        \includegraphics[width=0.85\textwidth]{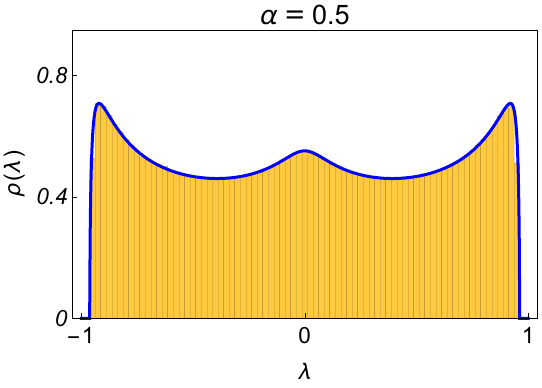}
        \caption{}
    \end{subfigure}
    \hfill
    \begin{subfigure}[b]{0.3\textwidth}
        \includegraphics[width=\textwidth]{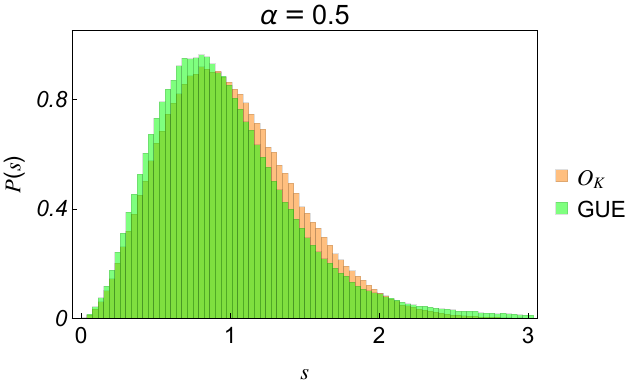}
        \caption{}
    \end{subfigure}
    \hfill
    \begin{subfigure}[b]{0.3\textwidth}
        \includegraphics[width=1.1\textwidth]{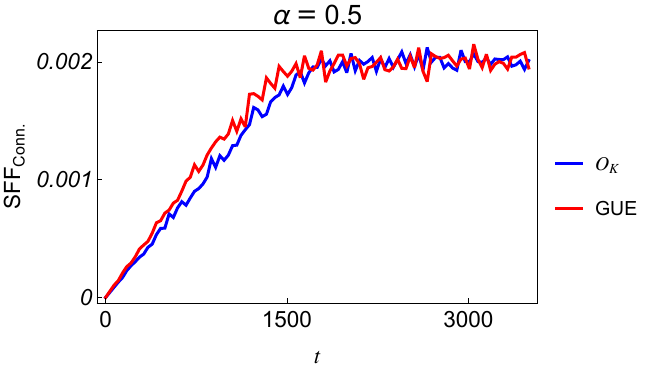}
        \caption{}
    \end{subfigure}
    
    \caption{Different statistics for the projected spin-$1$ operator for different values of $\alpha\leq\frac{1}{2}$. Left column: the numerically normalized level density (yellow), the theoretical prediction from free compression in Appendix \ref{App: DOS for spin1}. Central column: the normalized spacing distribution of $O_K$ (orange), and of the GUE (green). Right column: the connected SFF for $O_K$ (blue), and for the GUE (red)}
    \label{fig:spin1_stats}
\end{figure}

\begin{figure}[htbp]
    \centering
    
    \begin{subfigure}[b]{0.3\textwidth}
        \includegraphics[width=0.85\textwidth]{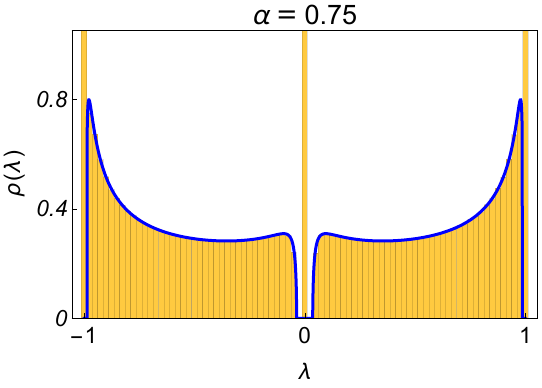}
        \caption{}
    \end{subfigure}
    \hfill
    \begin{subfigure}[b]{0.3\textwidth}
        \includegraphics[width=\textwidth]{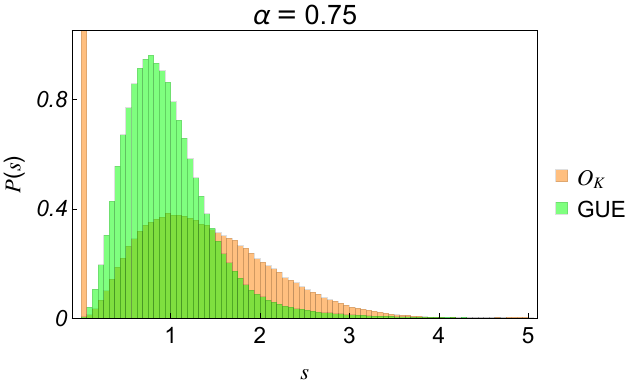}
        \caption{}
    \end{subfigure}
    \hfill
    \begin{subfigure}[b]{0.3\textwidth}
        \includegraphics[width=1.1\textwidth]{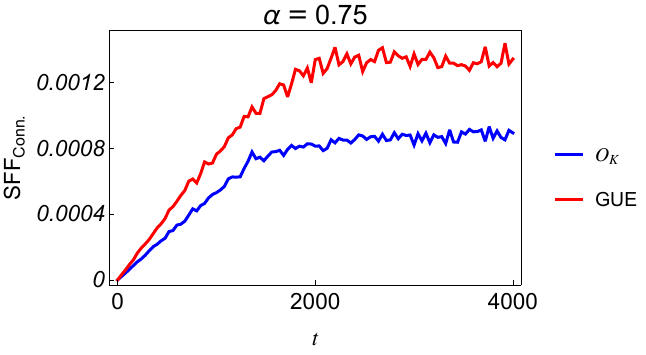}
        \caption{}
    \end{subfigure}
    
    \medskip
    
    \begin{subfigure}[b]{0.3\textwidth}
        \includegraphics[width=0.87\textwidth]{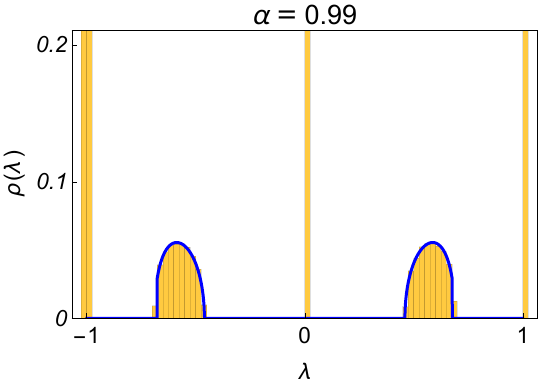}
        \caption{}
    \end{subfigure}
    \hfill
    \begin{subfigure}[b]{0.3\textwidth}
        \includegraphics[width=1.05\textwidth]{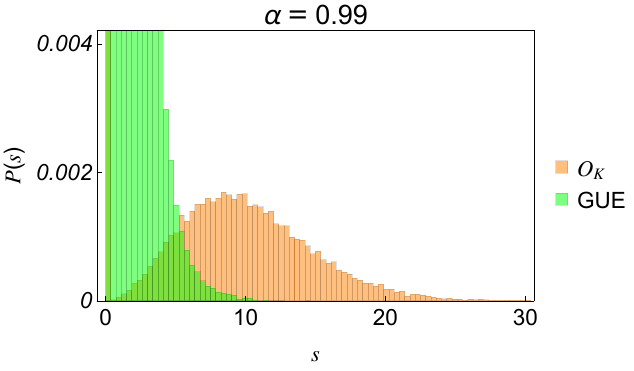}
        \caption{}
    \end{subfigure}
    \hfill
    \begin{subfigure}[b]{0.3\textwidth}
        \includegraphics[width=1.1\textwidth]{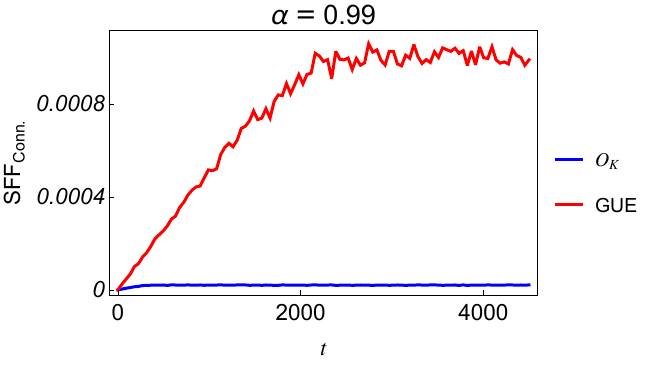}
        \caption{}
    \end{subfigure}
    
    \medskip
    
    \begin{subfigure}[b]{0.3\textwidth}
        \includegraphics[width=0.85\textwidth]{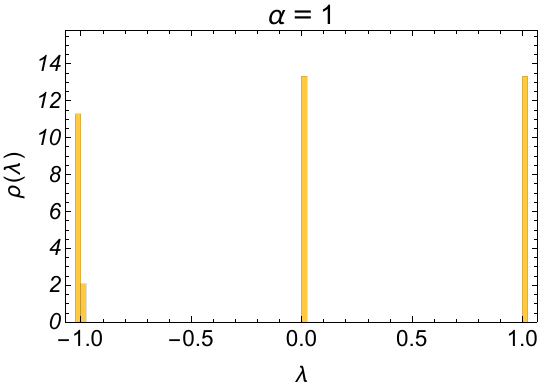}
        \caption{}
    \end{subfigure}
    \hfill
    \begin{subfigure}[b]{0.3\textwidth}
        \includegraphics[width=\textwidth]{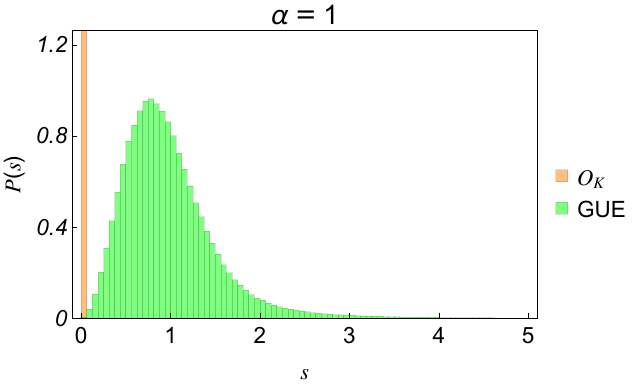}
        \caption{}
    \end{subfigure}
    \hfill
    \begin{subfigure}[b]{0.3\textwidth}
        \includegraphics[width=1.1\textwidth]{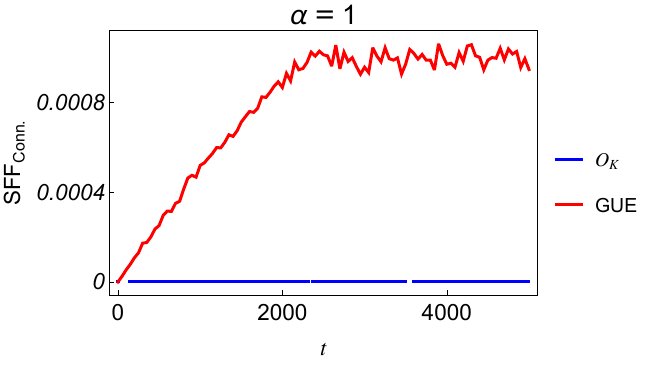}
        \caption{}
    \end{subfigure}
    \caption{Different statistics for the projected spin-$1$ operator for different values of $\alpha>\frac{1}{2}$. Left column: the numerically normalized level density (yellow), the theoretical prediction from free compression in Appendix \ref{App: DOS for spin1}. Central column: the normalized spacing distribution of $O_K$ (orange), and of the GUE (green). Right column: the connected SFF for $O_K$ (blue), and for the GUE (red)}
    \label{fig:spin1_stats2}
\end{figure}

\subsection{The projection of a 2-eigenvalue operator with different multiplicities}\label{sec:2eignum}
Here, we numerically study the random projection of a more general operator. The straightforward generalization is to consider an observable $O$ on an $N$-dimensional Hilbert space with 2 generic eigenvalues $\{a,b\}$ and respective generic multiplicities $\{\gamma,N-\gamma\}$, for $0\leq\gamma\leq N$. In its diagonal form, $O$ reads
\begin{align}
     O=
    \begin{pmatrix}
        a\,\mathbb{I}_\gamma & 0\\
        0 & b\,\mathbb{I}_{N-\gamma}
    \end{pmatrix}\,.
\end{align}
Here, the big difference is made by the unequal multiplicities, not by the fact that $a$ and $b$ are generic. The reason is that $O$ can always be related to the case $a=1$, $b=-1$ by a shift and dilation as follows:
\begin{align}
    O=\frac{a-b}{2}\tilde{S}+\frac{a+b}{2}\mathbb{I}_N
\end{align}
where
\begin{align}
    \tilde{S}=
    \begin{pmatrix}
        +\mathbb{I}_\gamma & 0\\
        0 & -\mathbb{I}_{N-\gamma}
    \end{pmatrix}\,.
\end{align}
Upon rotation to the energy eigenbasis and projection onto a subspace of dimension $K$, we get
\begin{align}
    O_K=\frac{a-b}{2}S_K+\frac{a+b}{2}\mathbb{I}_K\,,
\end{align}
where we used our usual notation $A_K\equiv P_K\,U A\,U^\dagger P_K$ for any operator $A$, which means that the relation between the eigenvalues of $O_K$ and the ones of $S_K$ is
\begin{align}
    \lambda_i=\frac{a-b}{2}s_i+\frac{a+b}{2}\,,
\end{align}
where $\{\lambda_i\}_{i=1...N}$ are the eigenvalues of $O_K$ and $\{s_i\}_{i=1...N}$ are the eigenvalues of $S_K$. This implies that the one-point eigenvalue distributions of the two operators are the same, modulo a shift and a dilation, together with a vertical rescaling to adjust for normalization:
\begin{align}\label{shift_resc}
    \rho_{O_K}(\lambda)=\frac{2}{|a-b|}\rho_{S_K}\left(\frac{2}{a-b}\left(\lambda-\frac{a+b}{2}\right)\right)\,.
\end{align}
With this in mind, we can show the results for just one choice of $\{a,b\}$, without loss of generality. What really changes the game is the fact that the multiplicities are unequal.

\begin{figure}[htbp]
    \centering
    
    \begin{subfigure}[b]{0.3\textwidth}
        \includegraphics[width=0.85\textwidth]{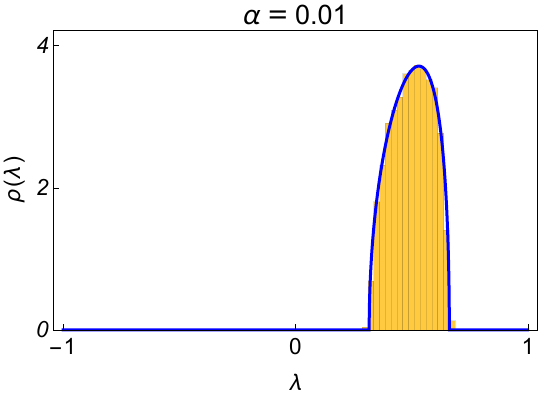}
        \caption{}
    \end{subfigure}
    \hfill
    \begin{subfigure}[b]{0.3\textwidth}
        \includegraphics[width=\textwidth]{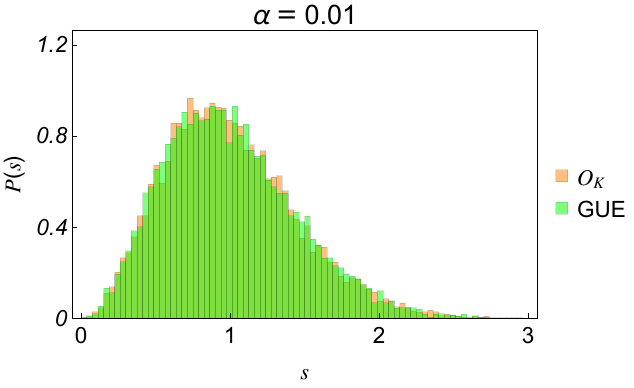}
        \caption{}
    \end{subfigure}
    \hfill
    \begin{subfigure}[b]{0.3\textwidth}
        \includegraphics[width=1.1\textwidth]{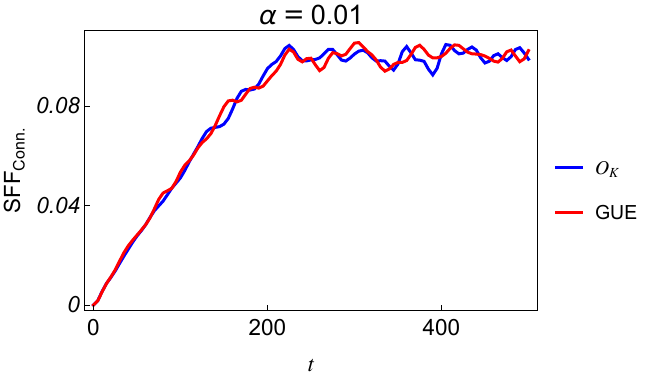}
        \caption{}
    \end{subfigure}
    
    \medskip
    
    \begin{subfigure}[b]{0.3\textwidth}
        \includegraphics[width=0.85\textwidth]{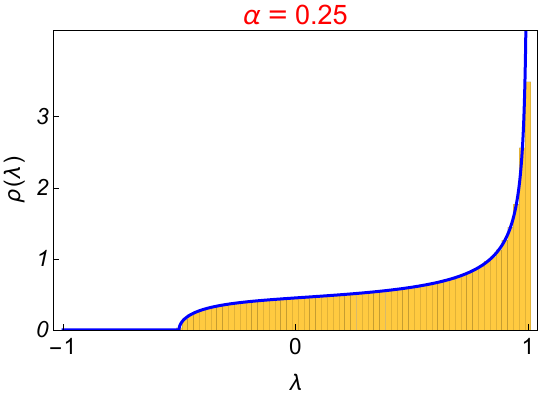}
        \caption{}
    \end{subfigure}
    \hfill
    \begin{subfigure}[b]{0.3\textwidth}
        \includegraphics[width=\textwidth]{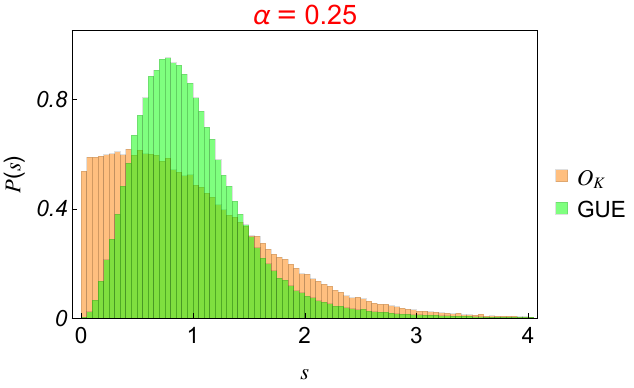}
        \caption{}
    \end{subfigure}
    \hfill
    \begin{subfigure}[b]{0.3\textwidth}
        \includegraphics[width=1.1\textwidth]{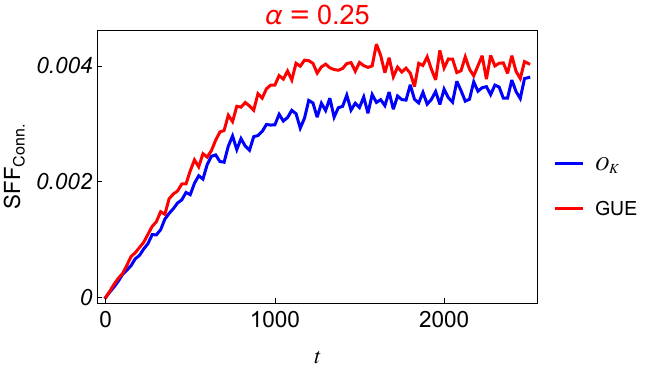}
        \caption{}
    \end{subfigure}
    
    \medskip
    
    \begin{subfigure}[b]{0.3\textwidth}
        \includegraphics[width=0.85\textwidth]{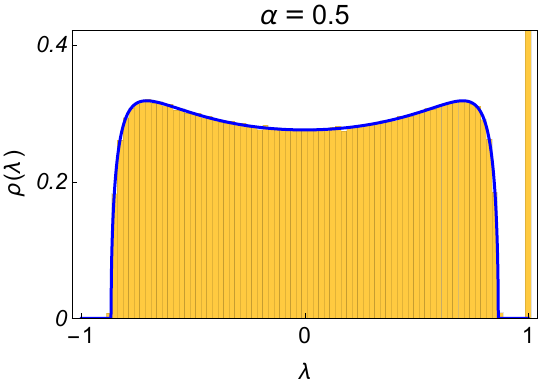}
        \caption{}
    \end{subfigure}
    \hfill
    \begin{subfigure}[b]{0.3\textwidth}
        \includegraphics[width=\textwidth]{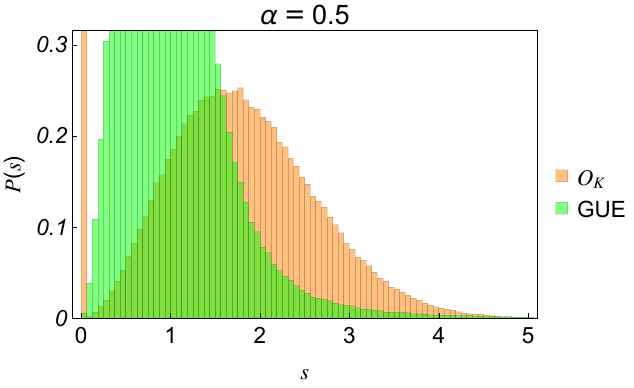}
        \caption{}
    \end{subfigure}
    \hfill
    \begin{subfigure}[b]{0.3\textwidth}
        \includegraphics[width=1.1\textwidth]{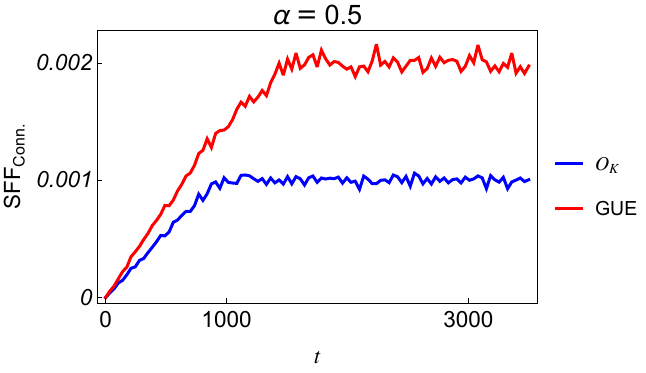}
        \caption{}
    \end{subfigure}
    
    \caption{Different statistics in the case of two different multiplicities, for various values of $\alpha\leq\frac{1}{2}$. Left column: the numerical normalized level density (yellow), the theoretical prediction of (\ref{2eigdistr}) (also from free compression in Section \ref{sec: FreeCompression}) (blue). Central column: the normalized spacing distribution of $O_K$ (orange), and of the GUE (green). Right column: the connected SFF for $O_K$ (blue), and for the GUE (red). In red, we highlight the values of $\alpha$ for which the transition occurs.}
    \label{fig:2eigop}
\end{figure}

\begin{figure}[htbp]
    \centering
    
    \begin{subfigure}[b]{0.3\textwidth}
        \includegraphics[width=0.85\textwidth]{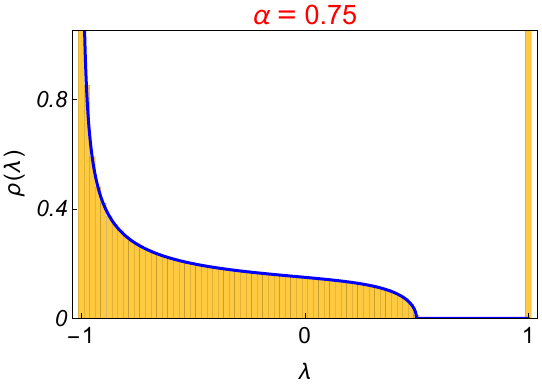}
        \caption{}
    \end{subfigure}
    \hfill
    \begin{subfigure}[b]{0.3\textwidth}
        \includegraphics[width=\textwidth]{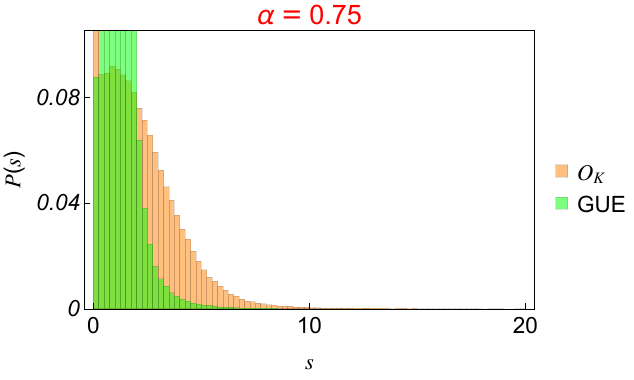}
        \caption{}
    \end{subfigure}
    \hfill
    \begin{subfigure}[b]{0.3\textwidth}
        \includegraphics[width=1.1\textwidth]{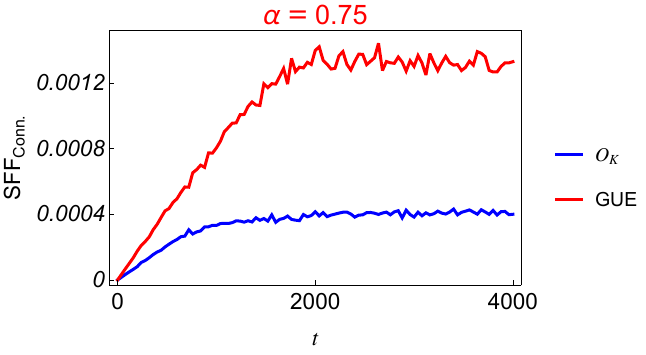}
        \caption{}
    \end{subfigure}
    
    \medskip
    
    \begin{subfigure}[b]{0.3\textwidth}
        \includegraphics[width=0.85\textwidth]{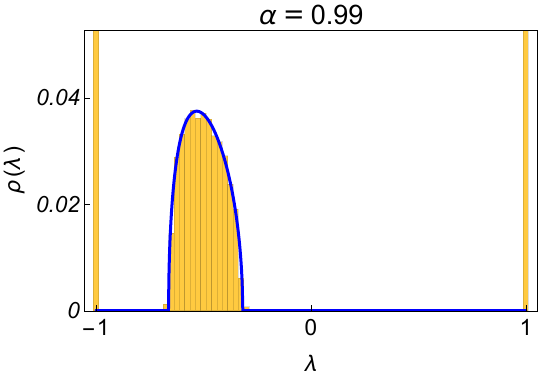}
        \caption{}
    \end{subfigure}
    \hfill
    \begin{subfigure}[b]{0.3\textwidth}
        \includegraphics[width=\textwidth]{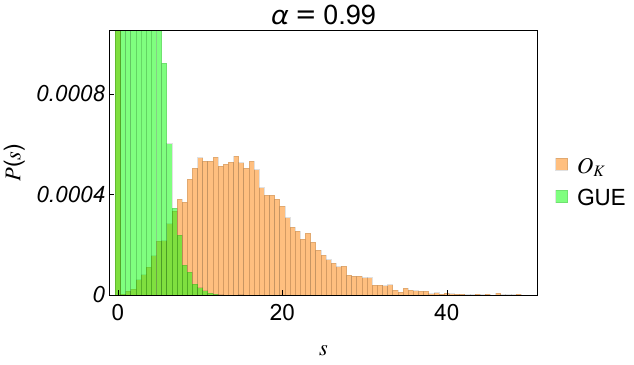}
        \caption{}
    \end{subfigure}
    \hfill
    \begin{subfigure}[b]{0.3\textwidth}
        \includegraphics[width=1.1\textwidth]{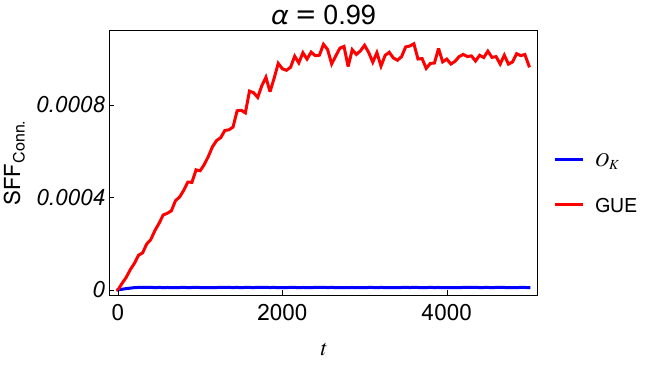}
        \caption{}
    \end{subfigure}
    
    \medskip
    
    \begin{subfigure}[b]{0.3\textwidth}
        \includegraphics[width=0.85\textwidth]{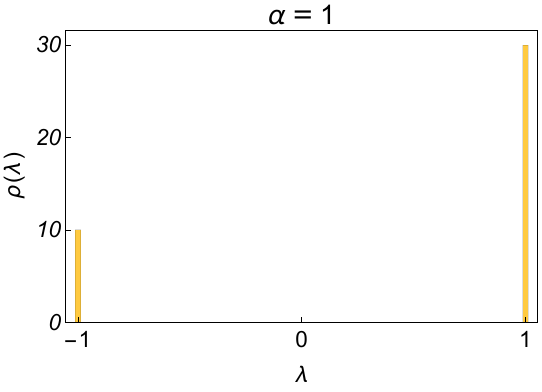}
        \caption{}
    \end{subfigure}
    \hfill
    \begin{subfigure}[b]{0.3\textwidth}
        \includegraphics[width=\textwidth]{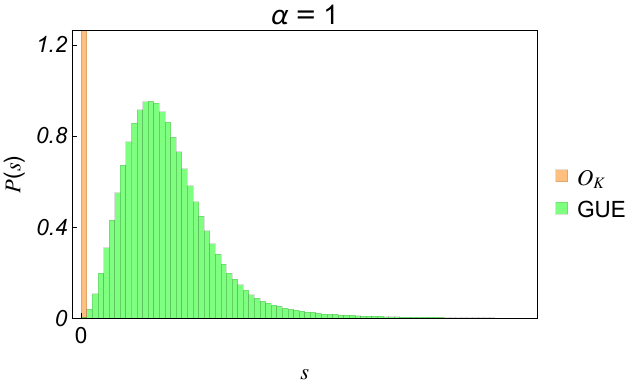}
        \caption{}
    \end{subfigure}
    \hfill
    \begin{subfigure}[b]{0.3\textwidth}
        \includegraphics[width=1.1\textwidth]{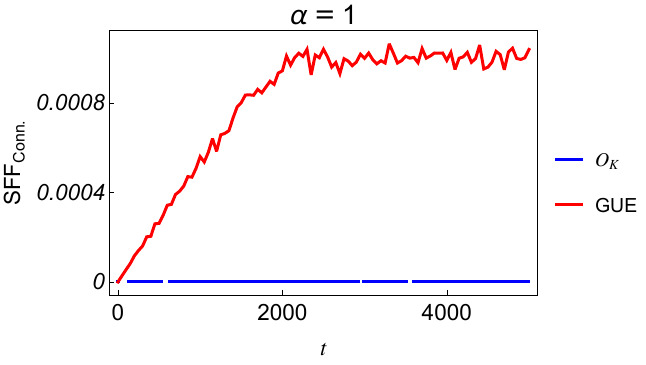}
        \caption{}
    \end{subfigure}
    \caption{Different statistics in the case of two different multiplicities, for various values of $\alpha>\frac{1}{2}$. Left column: the numerical normalized level density (yellow), the theoretical prediction of (\ref{2eigdistr}) (also from free compression in Section \ref{sec: FreeCompression}) (blue). Central column: the normalized spacing distribution of $O_K$ (orange), and of the GUE (green). Right column: the connected SFF for $O_K$ (blue), and for the GUE (red). In red, we highlight the values of $\alpha$ for which the transition occurs.}
    \label{fig:2eigop2}
\end{figure}

Let us take a look at what the distribution of the eigenvalues of $O_K$ looks like numerically, in the case where $a=1$ and $b=-1$ and for $\gamma=\frac{3}{4}N$.

In this case we separated the results in $\alpha\leq\frac{1}{2}$ (Figure \ref{fig:2eigop}) and $\alpha>\frac{1}{2}$ (Figure \ref{fig:2eigop2}) only for practical reasons, because $\alpha=\frac{1}{2}$ is not a special point anymore. As we can see from numerical results, we have two different values of $\alpha$ beyond which $O_K$ begins to recover the eigenvalues of the full $O$: beyond $\alpha=\frac{1}{4}$ we see the recovery of the eigenvalue $\{+1\}$, while beyond $\alpha=\frac{3}{4}$ we see the recovery of $\{-1\}$. In Figures \ref{fig:2eigop} and \ref{fig:2eigop2}, these transition values of $\alpha$ are highlighted in red. Moreover, the one-point eigenvalue distribution (left column) is generally asymmetric. We have proved in the Section \ref{sec: DOS from FPT} that, for all values of $\alpha\in[0,1]$, the one-point eigenvalue distribution reads
\begin{align}\label{2eigdistr}
    \rho(\lambda)=\frac{\sqrt{12\alpha(1-\alpha)-(1-2\alpha-2\lambda)^2}}{4\pi(1-\lambda^2)\mathcal{N}(\alpha,\beta)}
\end{align}
where $\mathcal{N}(\alpha,\beta)=\min(\alpha,\frac{3}{4})-\max(\alpha-\frac{1}{4},0)$, and $\beta=\gamma/N$. This distribution is shown as a solid blue line in the left column of Figure \ref{fig:2eigop}.

Despite the differences, we see some properties that seem to be the same as the equal-multiplicity case. First, when $\alpha\ll 1$, the level density, the spacing distribution, and the SFF are the same as in the GUE $H$ that satisfies
\begin{align}
    \E[H_{ij}]=0,\qquad \E[|H_{ij}|^2]=\frac{1}{N}\,,
\end{align}
and begin to deviate from it when $\alpha$ becomes bigger. In particular, the biggest deviations appear when we go beyond the transition points. Moreover, we again see that the property of eigenvalue repulsion is lost when $\alpha$ exceeds the first transition points, where degeneracy starts to dominate, thereby contributing to the level cluster.

\subsection{The projection of a 4-eigenvalue operator with different multiplicities}
\begin{figure}[htbp]
    \centering
    
    \begin{subfigure}[b]{0.3\textwidth}
        \includegraphics[width=0.85\textwidth]{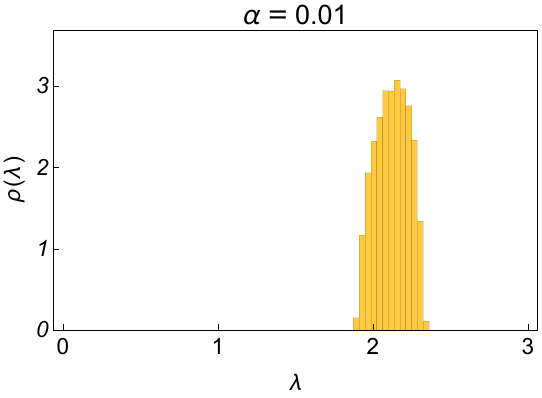}
        \caption{}
    \end{subfigure}
    \hfill
    \begin{subfigure}[b]{0.3\textwidth}
        \includegraphics[width=\textwidth]{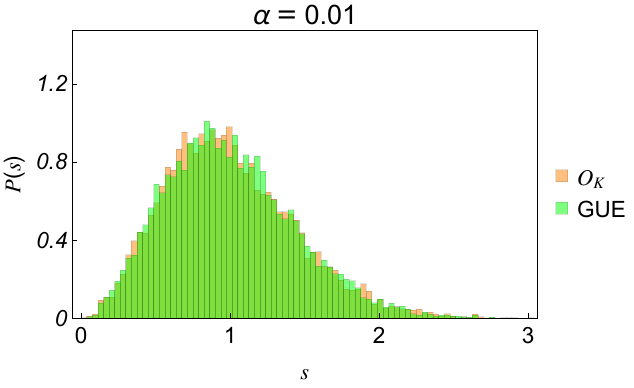}
        \caption{}
    \end{subfigure}
    \hfill
    \begin{subfigure}[b]{0.3\textwidth}
        \includegraphics[width=1.1\textwidth]{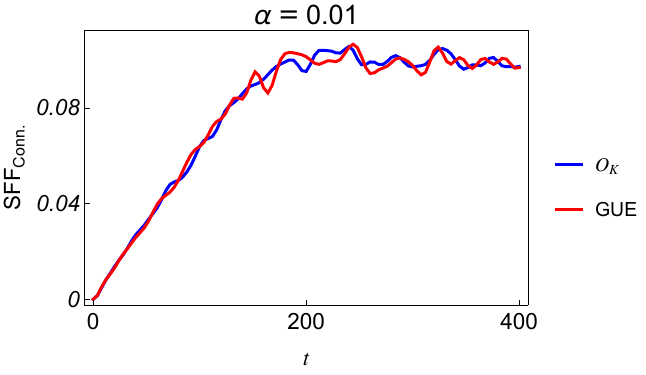}
        \caption{}
    \end{subfigure}
    
    \medskip

    \begin{subfigure}[b]{0.3\textwidth}
        \includegraphics[width=0.85\textwidth]{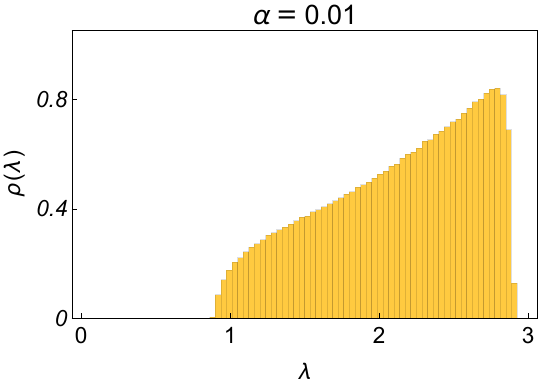}
        \caption{}
    \end{subfigure}
    \hfill
    \begin{subfigure}[b]{0.3\textwidth}
        \includegraphics[width=\textwidth]{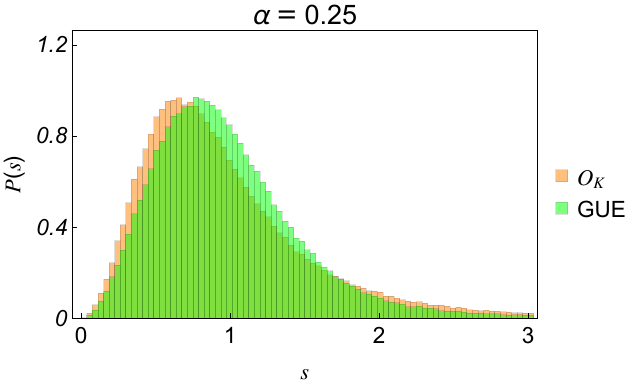}
        \caption{}
    \end{subfigure}
    \hfill
    \begin{subfigure}[b]{0.3\textwidth}
        \includegraphics[width=1.1\textwidth]{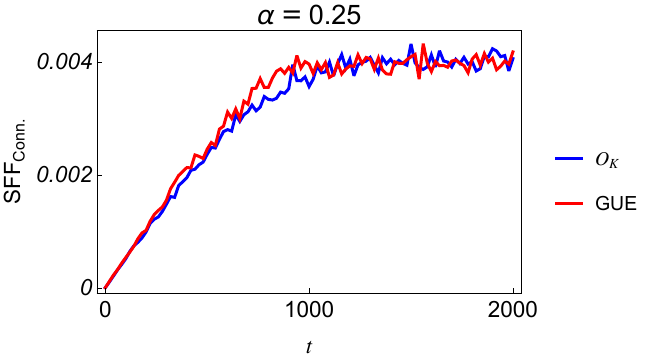}
        \caption{}
    \end{subfigure}
    
    \medskip
    
    \begin{subfigure}[b]{0.3\textwidth}
        \includegraphics[width=0.85\textwidth]{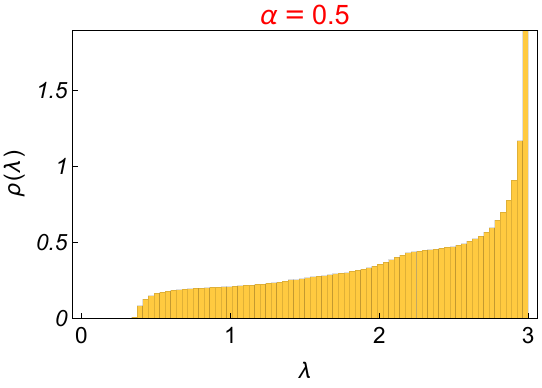}
        \caption{}
    \end{subfigure}
    \hfill
    \begin{subfigure}[b]{0.3\textwidth}
        \includegraphics[width=\textwidth]{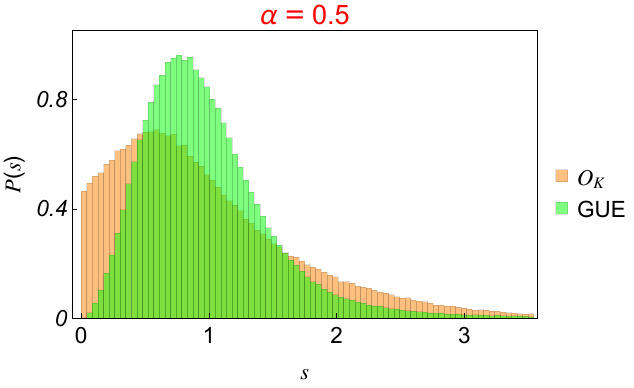}
        \caption{}
    \end{subfigure}
    \hfill
    \begin{subfigure}[b]{0.3\textwidth}
        \includegraphics[width=1.1\textwidth]{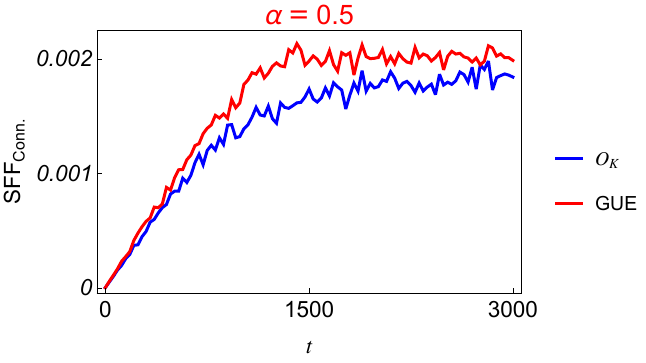}
        \caption{}
    \end{subfigure}
    
    \caption{Different statistics in the case of four eigenvalues with different multiplicities, for various values of $\alpha\leq\frac{1}{2}$. Left column: the normalized level density (yellow). Central column: the normalized spacing distribution of $O_K$ (orange), and of the GUE (green). Right column: the connected SFF for $O_K$ (blue), and for the GUE (red). In red, we highlight the values of $\alpha$ for which the transition occurs.}
    \label{fig:4eigop}
\end{figure}

\begin{figure}[htbp]
    \centering
    
    \begin{subfigure}[b]{0.3\textwidth}
        \includegraphics[width=0.85\textwidth]{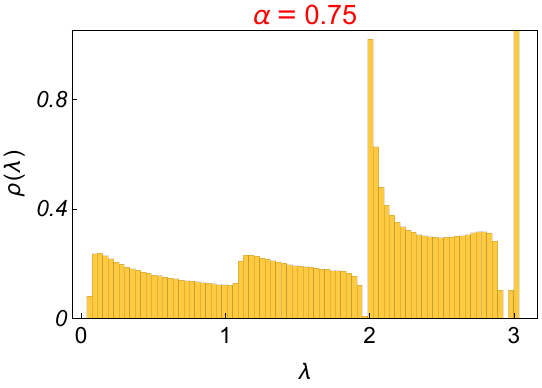}
        \caption{}
    \end{subfigure}
    \hfill
    \begin{subfigure}[b]{0.3\textwidth}
        \includegraphics[width=\textwidth]{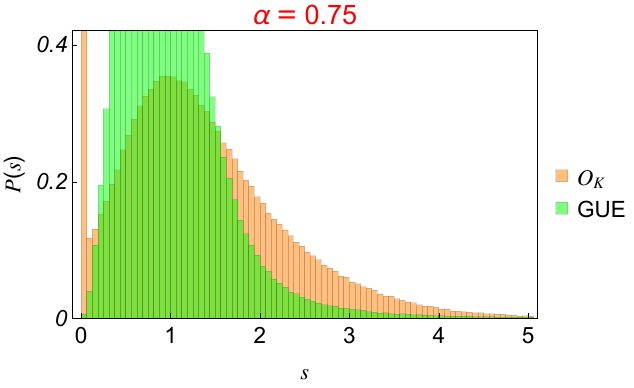}
        \caption{}
    \end{subfigure}
    \hfill
    \begin{subfigure}[b]{0.3\textwidth}
        \includegraphics[width=1.1\textwidth]{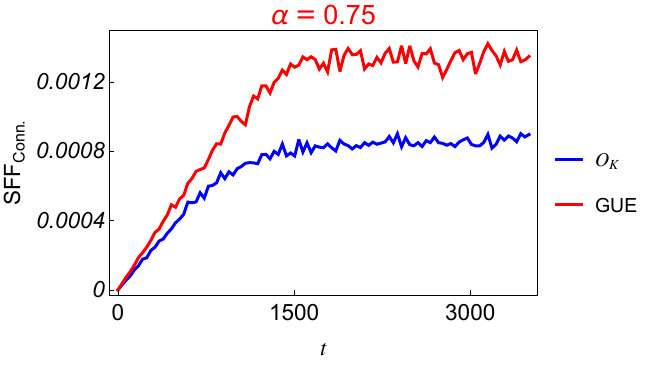}
        \caption{}
    \end{subfigure}
    
    \medskip
    
    \begin{subfigure}[b]{0.3\textwidth}
        \includegraphics[width=0.85\textwidth]{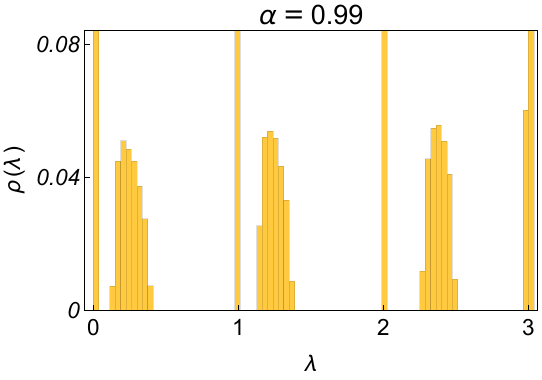}
        \caption{}
    \end{subfigure}
    \hfill
    \begin{subfigure}[b]{0.3\textwidth}
        \includegraphics[width=\textwidth]{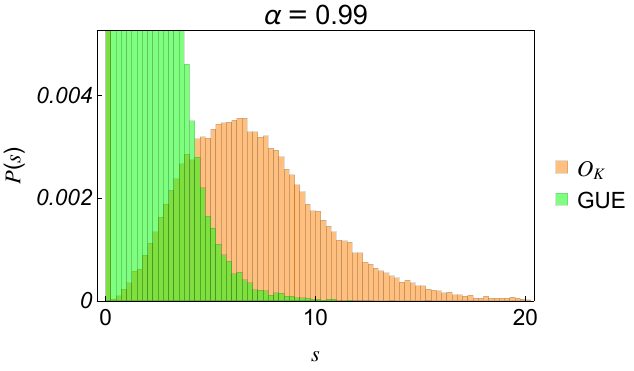}
        \caption{}
    \end{subfigure}
    \hfill
    \begin{subfigure}[b]{0.3\textwidth}
        \includegraphics[width=1.1\textwidth]{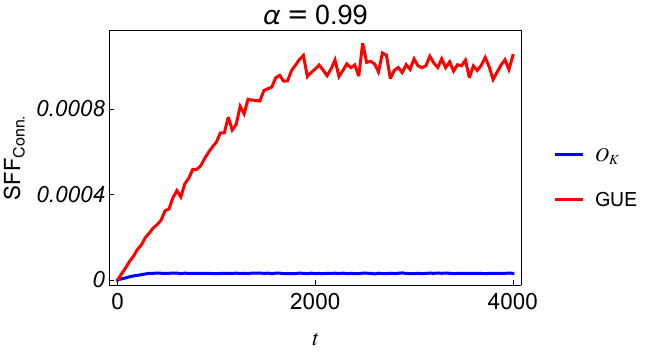}
        \caption{}
    \end{subfigure}
    
    \medskip
    
    \begin{subfigure}[b]{0.3\textwidth}
        \includegraphics[width=0.85\textwidth]{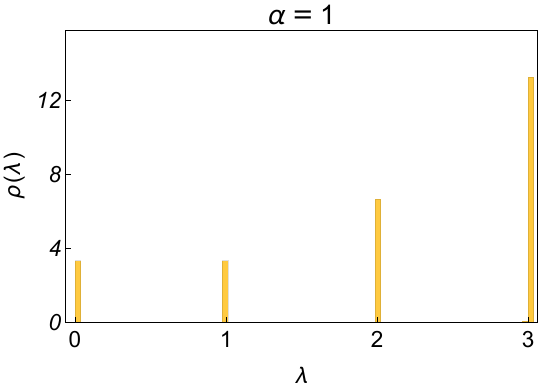}
        \caption{}
    \end{subfigure}
    \hfill
    \begin{subfigure}[b]{0.3\textwidth}
        \includegraphics[width=\textwidth]{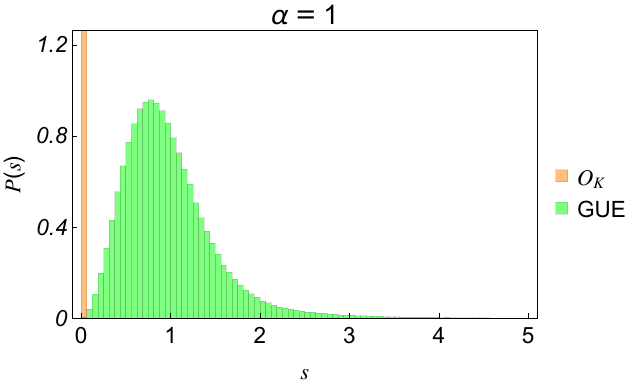}
        \caption{}
    \end{subfigure}
    \hfill
    \begin{subfigure}[b]{0.3\textwidth}
        \includegraphics[width=1.1\textwidth]{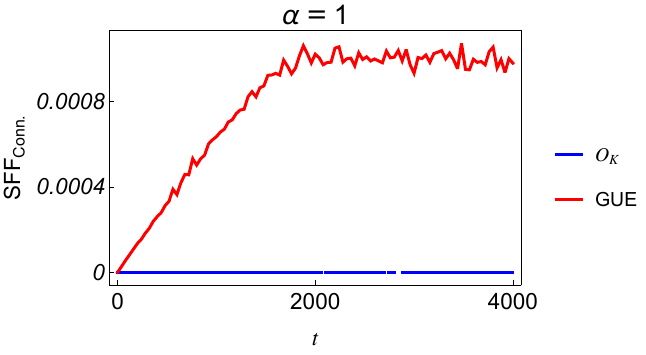}
        \caption{}
    \end{subfigure}
    \caption{Different statistics in the case of four eigenvalues with different multiplicities, for various values of $\alpha>\frac{1}{2}$. Left column: the normalized level density (yellow). Central column: the normalized spacing distribution of $O_K$ (orange), and of the GUE (green). Right column: the connected SFF for $O_K$ (blue), and for the GUE (red). In red, we highlight the values of $\alpha$ for which the transition occurs.}
    \label{fig:4eigop2}
\end{figure}
Let us get to a further generalization. Take now an $O$ with four eigenvalues $\{0,1,2,3\}$ with respective multiplicities $\left\{\frac{1}{8}N,\frac{1}{8}N,\frac{1}{4}N,\frac{1}{2}N\right\}$. We show the statistics in Figures \ref{fig:4eigop} and \ref{fig:4eigop2}, from which we can see that now we have three values of alpha beyond which $O_K$ recovers the eigenvalues of $O$: beyond $\alpha=\frac{1}{2}$ $O_K$ recovers $\{3\}$, beyond $\alpha=\frac{3}{4}$ it recovers $\{2\}$ and beyond $\alpha=\frac{7}{8}$ it recovers $\{0,1\}$. All the other properties are very similar to the previous cases: we lose eigenvalue repulsion at the transition points, and for $\alpha\ll 1$ $O_K$ shares the statistics with the GUE $H$, which has moments
\begin{align}
    \E[H_{ij}]=\mu\delta_{ij},\qquad \E[|H_{ij}|^2]=\frac{1}{N}
\end{align}
where $\mu=\frac{1}{N}\Tr O$ is the average of the eigenvalues of $O$.

\section{Density of states from Free Compression using $R$-transform for spin-$1$ operators}\label{App: DOS for spin1}
In this Appendix, we show the eigenvalue distributions of the projected spin-$1$ operators for several values of projection parameters $\alpha=0.01, 0.25,0.5,0.75,0.99$. These distributions are obtained analytically from the R-transform (\ref{eq-freeDec R-tranform sp1}), using the free compression method described in Section \ref{sec: FreeCompression}. The corresponding plots are shown in blue in Figures \ref{fig:spin1_stats} and \ref{fig:spin1_stats2} (left columns), which perfectly match the numerical results.

\begin{equation}
    \mu_{ \tilde{a}/100+ \tilde{b}/100}(\lambda)=\frac{-90000 \lambda^4+f_1^{2/3}-26427 \lambda^2-38809}{6 \sqrt{3} \pi  \lambda \left(\lambda^2-1\right) f_1^{1/3}}\,,
\end{equation}
where 
\begin{eqnarray}
    f_1&=&7645373+9 \left(30000 \left(100 \lambda ^2-397\right) \lambda ^2+1264627\right) \lambda ^2\nonumber\\
    &&+2673 \sqrt{-\lambda ^2 \left(\lambda ^2-1\right)^2 \left(30000 \lambda ^2 \left(30000 \lambda ^2+8809\right)-7645373\right)}\,.
\end{eqnarray}

\begin{equation}
 \mu_{ \tilde{a}/4+ \tilde{b}/4}(\lambda)=\frac{-144 \lambda ^4+69 \lambda ^2+f_2^{2/3}-25}{6 \sqrt{3} \pi  \lambda  \left(\lambda ^2-1\right) f_2^{1/3}}\,,
\end{equation}
where
\begin{equation}
    f_2=81 \sqrt{-\lambda ^2 \left(\lambda ^2-1\right)^2 \left(48 \lambda ^2 \left(48 \lambda ^2-23\right)-125\right)}+9 \left(192 \lambda ^4-624 \lambda ^2+307\right) \lambda ^2+125\,.
\end{equation}

\begin{equation}
 \mu_{ \tilde{a}/2+ \tilde{b}/2}(\lambda)=\frac{-36 \lambda ^4+33 \lambda ^2+f_3^{2/3}-1}{6 \sqrt{3} \pi  \lambda  \left(\lambda ^2-1\right) f_3^{1/3}}\,,
\end{equation}
where
\begin{equation}
    f_3=9 \left(24 \lambda ^4-60 \lambda ^2+35\right) \lambda ^2+27 \sqrt{-\lambda ^2 \left(\lambda ^2-1\right)^2 \left(144 \lambda ^4-132 \lambda ^2-1\right)}+1\,.
\end{equation}

\begin{equation}
 \mu_{ 3\tilde{a}/4+ 3\tilde{b}/4}(\lambda)=\frac{-144 \lambda ^4+141 \lambda ^2+f_4^{2/3}-1}{18 \sqrt{3} \pi  \lambda  \left(\lambda ^2-1\right) f_4^{1/3}}\,,
\end{equation}
where
\begin{equation}
    f_4=9 \left(48 \left(4 \lambda ^2-7\right) \lambda ^2+145\right) \lambda ^2+27 \sqrt{3} \sqrt{-\lambda ^2 \left(\lambda ^2-1\right)^2 \left(768 \lambda ^4-752 \lambda ^2+1\right)}-1\,.
\end{equation}

\begin{equation}
 \mu_{ 99\tilde{a}/100+ 99\tilde{b}/100}(\lambda)=\frac{-90000 \lambda ^4+61773 \lambda ^2+f_5^{2/3}-9409}{594 \sqrt{3} \pi  \lambda  \left(\lambda ^2-1\right) f_5^{1/3}}\,,
\end{equation}
where
\begin{eqnarray}
    f_5&=&9024057 \lambda ^2+270000 \left(100 \lambda ^2-103\right) \lambda ^4-912673 \nonumber\\
    &&+27 \sqrt{3} \sqrt{-\lambda ^2 \left(\lambda ^2-1\right)^2 \left(10000 \left(30000 \lambda ^2-20591\right) \lambda ^2+30118209\right)}\,.
\end{eqnarray}

\bibliographystyle{ytphys}
\bibliography{ref}

\end{document}